%% file: ms.tex
\pdfoutput=1
\documentclass[notes-porbanz,noinfoline]{notes-imsart}

\RequirePackage[OT1]{fontenc}

\usepackage[usenames,dvipsnames]{color}
\RequirePackage{hyperref}

\setattribute{copyright}{text}{\empty}
\hypersetup{pdfinfo={Subject={ }}}
\usepackage{bbm}
\usepackage{amsmath,amssymb,amsthm} 
\usepackage{dsfont}
\usepackage{thmtools}

\usepackage[numbers]{natbib}
\usepackage{mathtools}
\usepackage{graphicx}
\usepackage{booktabs}
\usepackage{placeins}
\usepackage{enumitem}
\usepackage[capitalize,noabbrev]{cleveref}
\usepackage{tikz}
\usetikzlibrary{calc,fit,decorations.pathreplacing}

\usepackage{picins}

\DeclareGraphicsExtensions{.pdf,.png,.jpg}

\newtheorem{theorem}{Theorem}[section]
\newtheorem{proposition}[theorem]{Proposition}
\newtheorem{corollary}{Corollary}[theorem]

\newtheorem{lemma}{Lemma}

\newtheorem*{observation}{Observation}

\theoremstyle{remark}

\makeatletter
\newcommand{\killpic}{%
  \hangindent=0pt
  \let\par=\old@par
}
\makeatother

\setlist[itemize]{itemsep=0pt, topsep=4pt}

\declaretheoremstyle[
  postheadhook = {\hspace*{\parindent}},
  mdframed={
    backgroundcolor=gray!10!white, 
    hidealllines=true, 
    innertopmargin=2pt, 
    innerbottommargin=4pt, 
    skipabove=8pt,
    skipbelow=10pt,
    nobreak=true}
]{grayboxed} 
\declaretheoremstyle[
  mdframed={
    backgroundcolor=gray!10!white, 
    hidealllines=true, 
    innertopmargin=2pt, 
    innerbottommargin=4pt, 
    skipabove=8pt,
    skipbelow=10pt,
    nobreak=true},headpunct=\empty,
]{grayboxed*} 
\declaretheorem[style=plain,qed=$\triangleleft$]{auxtheorem}
\declaretheorem[style=grayboxed,sibling=auxtheorem,name=Algorithm]{poalgorithm}
\declaretheorem[{numbered=no},style=grayboxed*,heading=\empty,name=\empty]{postep}

\newcommand{\appxref}[1]{\hyperref[#1]{Appendix~\ref{#1}}}

\tikzstyle{mybraces}=[mirrorbrace/.style={
          decoration={brace, mirror},
          decorate},brace/.style={
          decoration={brace},
          decorate}]

\def\kword#1{\textbf{#1}}
\def\braces#1{{\lbrace #1 \rbrace}}
\def\ind#1{\text{\rm\tiny #1}}
\def\ie{i.e.\ }

\def\N{\mathbb{N}}
\def\R{\mathbb{R}}
\def\argdot{\,\bullet\,}
\def\indicator{\mathds{1}} % indicator
\def\iden{\mathbb{I}} % identity matrix

\def\bfX{\mathbf{X}}

\def\calK{\mathcal{K}}

\def\vertexset{\mathcal{V}}
\def\edgeset{\mathcal{E}}
\def\deg{\text{\rm deg}}

\def\vol{\text{\rm vol}}

\def\L{\mathbf{L}} % normalized Laplacian
\def\LeigVec{\boldsymbol{\psi}} % eigenvectors of \L
\def\LeigVal{\sigma} % eigenvalues of \L

\def\P{\mathbb{P}}
\def\E{\mathbb{E}}
\def\eqdist{\overset{d}{=}}
\def\simiid{\overset{\mbox{\tiny iid}}{\sim}}

\def\almostSurely{\overset{\text{\small a.s.}}{\longrightarrow}}

\def\on{\, \vert \,}

\def\Bernoulli{\text{\rm Bernoulli}}
\def\Poisson{\text{\rm Poisson}_{+}}
\def\Uniform{\text{\rm Uniform}}
\def\GammaDist{\text{\rm Gamma}}
\def\BetaDist{\text{\rm Beta}}
\def\NB{\text{\rm NB}_{+}}
\def\NBo{\text{\rm NB}}
\def\Geom{\text{\rm Geometric}}
\def\MittagLeffler{\text{\rm Mittag-Leffler}}
\def\MN{\text{\rm Multinomial}}

\def\SRW{\text{SRW}}

\def\RW{\mathbf{RW}}

\def\SB{\mathbb{S}}
\def\RWU{\RW_{\!\ind{U}}}
\def\RWSB{\RW_{\!\ind{SB}}}
\def\ACL{\mathbf{ACL}}

\def\seqParams{\theta}

\def\seqObs{x}
\def\seqState{z}
\def\seqObs{X}
\def\seqState{Z}

\def\stub[#1][#2][#3]{\delta^{(#1)}_{#2}(#3)} % indicator of edge stub
\def\jvec{\vec{j}}
\def\kvec{\vec{k}}

\def\Gall{G_{1:T}}
\def\Glatent{G_{2:(T-1)}}
\def\feasibleSet[#1][#2]{\mathcal{F}_{#1}(G_{#2})} % feasible set of graphs
\def\kAll{\mathbf{K}}
\def\kNot{\mathbf{K}_{-t}}
\def\Iall{\mathbf{B}}
\def\INot{\mathbf{B}_{-t}}

\def\P{\mathbbm{P}}

\def\sigalg[#1]{\mathcal{A}_{#1}}

\def\dmin[#1]{\underline{d}^{(#1)}}
\def\DegMat{D}

\def\vDist{\mu}
\def\eDist{\nu}

\def\nball{\widetilde{\mathcal{N}}}
\def\RWprob{Q}

\def\postSpacet[#1]{\widehat{\mathcal{G}}_{#1}(G_T)}

\def\map[#1][#2]{\varphi_{#1,#2}}

\def\mapSet[#1][#2]{\Phi_{#1,#2}}

\def\argdot{\text{\raisebox{0.25ex}{\,\tiny$\bullet$\,}}}

\DeclarePairedDelimiterX\bigCond[2]{[}{]}{#1 \;\delimsize\vert\; #2}
\newcommand{\conditional}[3][]{\E_{#1}\bigCond*{#2}{#3}}

\begin{document}
\begin{frontmatter}
  \title{Random Walk Models of Network Formation and Sequential Monte Carlo Methods for Graphs}
  \runtitle{Random Walk Network Models}
  \begin{aug}
    \author{\fnms{Benjamin\ }\snm{Bloem-Reddy}\ead[label=e1]{benjamin.bloem-reddy@stats.ox.ac.uk}}
    \and
    \author{\fnms{Peter\ }\snm{Orbanz}\corref{}\ead[label=e2]{porbanz@stat.columbia.edu}}
    \runauthor{Bloem-Reddy and Orbanz}
    \affiliation{University of Oxford and Columbia University}

    \address{Department of Statistics\\
      24--29 St. Giles'\\
      Oxford OX1 3LB, UK\\
      \printead{e1}
    }

    \address{Department of Statistics\\
      1255 Amsterdam Avenue\\
      New York, NY 10027, USA\\
      \printead{e2}
    }
  \end{aug}
  \maketitle
  \begin{abstract}
    We introduce a class of generative network models that insert edges by
    connecting the starting and terminal vertices of a random
    walk on the network graph. 
    Within the taxonomy of statistical network models, this class is distinguished
    by permitting the location of a new edge to explicitly depend on the
    structure of the graph, but being nonetheless statistically and computationally
    tractable. In the limit of infinite walk length, the model converges
    to an extension of the preferential
    attachment model---in this sense, it can be motivated 
    alternatively by asking what preferential attachment is an approximation to.
    Theoretical properties, including the limiting degree sequence,
    are studied analytically.
    If the entire history of the graph is observed, parameters can be
    estimated by maximum likelihood. If only the final graph is available, its history can be imputed
    using MCMC. We develop a class of sequential Monte Carlo algorithms that are
    more generally applicable to sequential network
    models, and 
    may be of interest in their own right. The model parameters can be recovered from a single
    graph generated by the model. Applications to 
    data clarify the role of the random walk length as a 
    length scale of interactions within the graph.
  \end{abstract}
\end{frontmatter}

\section{Introduction} 
\label{sec:introduction}

We consider problems in network analysis where the observed data constitutes a single graph.
Tools developed for such problems include, for example, graphon models 
\citep{Borgs:Chayes:Lovasz:Sos:Vesztergombi:2008,Hoff:Raftery:Handcock:2002,Bickel:Chen:Levina:2011,Ambroise:Matias:2012:1,Wolfe:Olhede:2013:1,Gao:Lu:Zhou:2015:1}, 
the subfamily of stochastic blockmodels \citep[e.g.][]{Goldenberg:etal:2010},
preferential attachment graphs \citep{Barabasi:Albert:1999}, 
interacting particle models \citep{Liggett:2005:1}, and others.
In this context, a statistical model is a family ${\mathcal{P}=\lbrace P_{\theta},\theta\in\mathbf{T}\rbrace}$
of probability distributions on graphs,
and sample data is explained either as a single graph drawn from the model, or as a subgraph of such a graph.

Applied network analysis distinguishes between network formation mechanisms that are
\emph{endogenous} (are an outcome of the current structure of the network), and those that are
\emph{exogenous} (are driven by external effects), see e.g.\ 
\citep{Tomasello:Perra:Tessone:Karsai:Schweitzer:2014:1,Minhas:Hoff:Ward:2016:1}. 
We develop methodology to account for endogenous mechanisms. That necessarily introduces
stochastic dependence between existing and newly emerging structure in the network.
The technical challenge is to identify models that nonetheless
are statistically and computationally tractable.

The approach taken in the following is generative, in the sense that it formulates a model of how a graph grows
as a stochastic process. 
Specifically, the class of models we study generate a graph by inserting one edge at a time:
Initialize a graph as a single edge, with its two terminal vertices. For each new edge, 
select a vertex $V$ in the current graph at random.
\begin{itemize}
  \item With a fixed probability ${\alpha\in(0,1]}$, add a new vertex and connect it
  to $V$.
  \item Otherwise, start a random walk at $V$, and connect its terminal vertex
  to $V$.
\end{itemize}
The result is a graph or multigraph, with a random number of vertices labeled in order of occurrence.
Generative network models are studied in applied probability \citep[e.g.][]{Durrett:2006},
where they are often motivated by how generative mechanisms determine properties of the graph.
Perhaps the best-known examples are preferential attachment graphs, in which a simple generative
mechanism (reinforcement based on degree-biased selection of vertices) explains an observable characteristic of networks
(power law degrees). 
For the purposes of statistics, a family of generative distributions
is applicable as a model if (i) it can be parametrized in a meaningful way and (ii) if parameter inference
from observed data is possible. As the ensuing discussion shows, inference using MCMC techniques
is possible for a range of generative models, including the one sketched above.

The endogenous mechanism characterizing the model above is the random walk, which lets vertices
connect ``through the network''. At each step, an endogenous mechanism is chosen with probability
$(1-\alpha)$, or an exogenous one with probability $\alpha$.
In contrast to the models listed at the outset,
the locations of newly inserted edges thus depend on those of previously
inserted ones, 
and the effect of this dependence is clearly visible in graphs generated by
the model (see \cref{fig:examples:uniform} in \cref{sec:model} for examples).
The random walk step itself can be regarded as a model of interactions:
\begin{itemize}
\item 
  Long walks can connect
  vertices far apart in the graph. The (expected) length of the walk can hence
  be regarded as a range scale.
\item 
  Where convenient, the walk may be interpreted explicitly,
  as a process on the network (e.g.\ users in a social network forming connections through
  other users). More generally, it biases connections
  towards vertices reachable along multiple paths.
\end{itemize}
We show in the following that random walk models can explain a range of graph structures as the outcome of
interactions mediated by the network,
and are statistically and analytically tractable.\\[-2pt]

{\noindent\textit{Background.}}
At a conceptual level, the work presented here is motivated by the question of 
how statistical network models handle stochastic dependence within a network;
the discussion in this paragraph is purely heuristic.
Consider the link graph of a network, regarded as a system of 
random variables (representing edges and vertices). Informally, for structure to occur in a network
with non-negligible probability, the variables must be dependent. The graphs below
are (i) a protein-protein interaction network (see \cref{sec:experiments}), and
(ii) its reconstruction sampled from a graphon model fitted to (i), and (iii) a
reconstruction from the random walk model:
\begin{center}
  \begin{tikzpicture}
    \path[use as bounding box]
      (-.6,-2) rectangle (9,2);
      \node at (-.5,0) {\includegraphics[width=4.5cm,angle=-20]{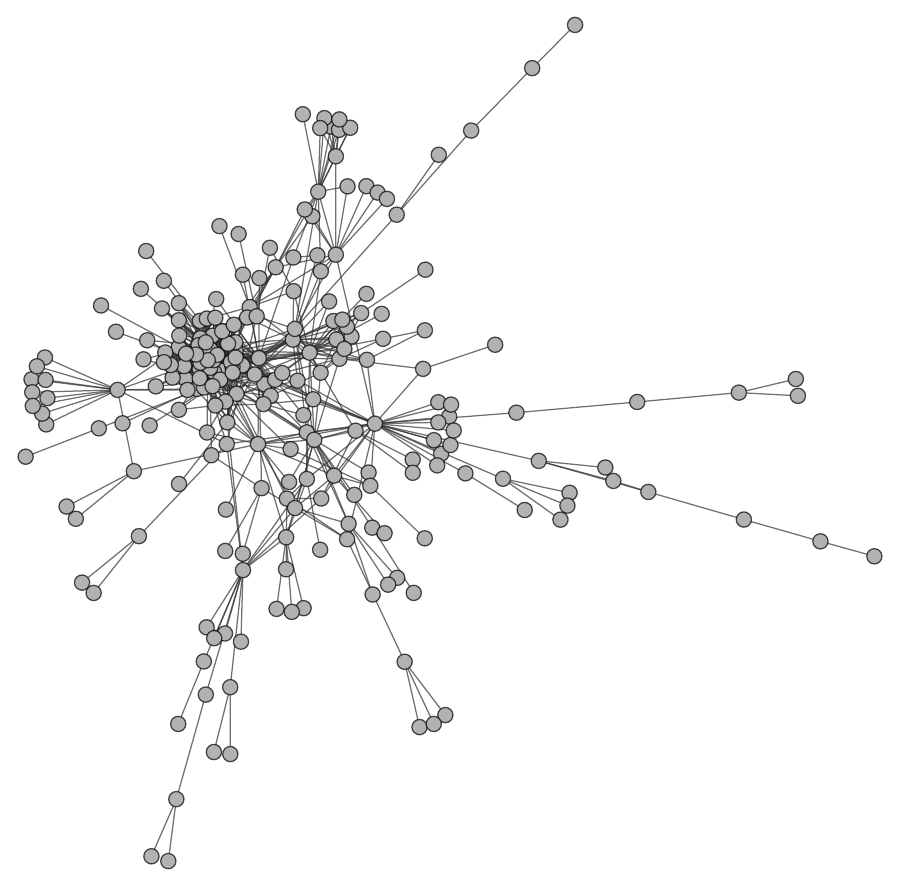}};
      \node[rotate=90] at (4.5,0) {\includegraphics[width=4.5cm]{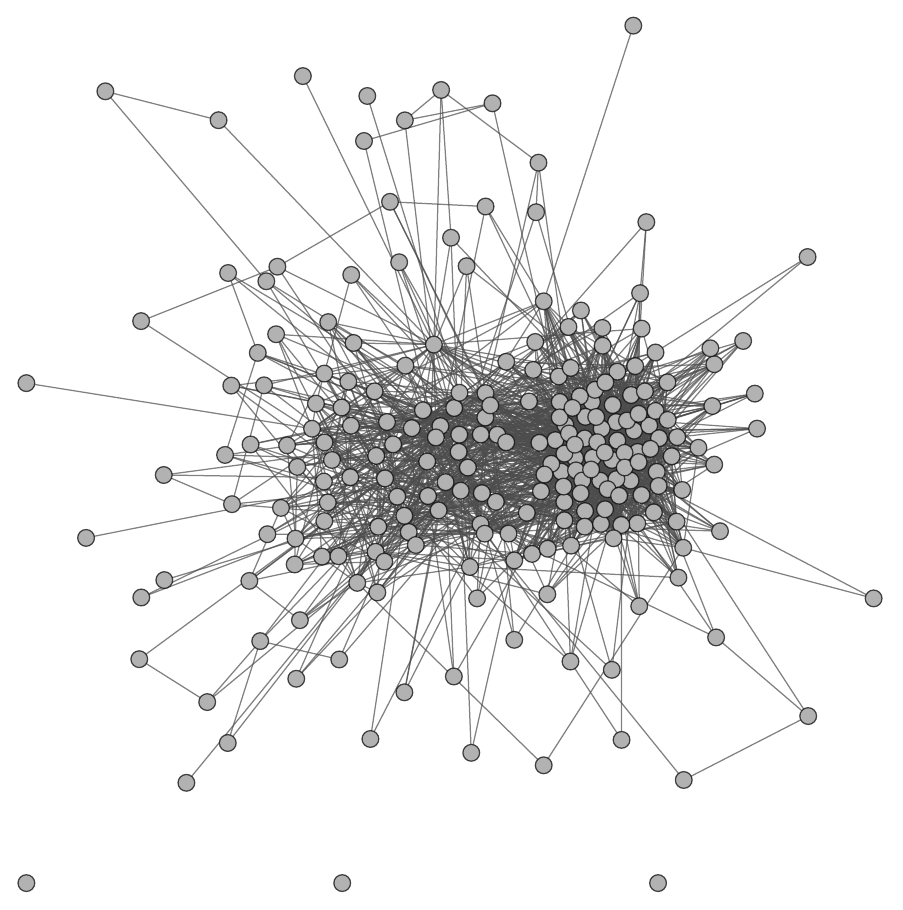}};
      \node[rotate=50] at (9.9,-.4) {\includegraphics[width=5.5cm]{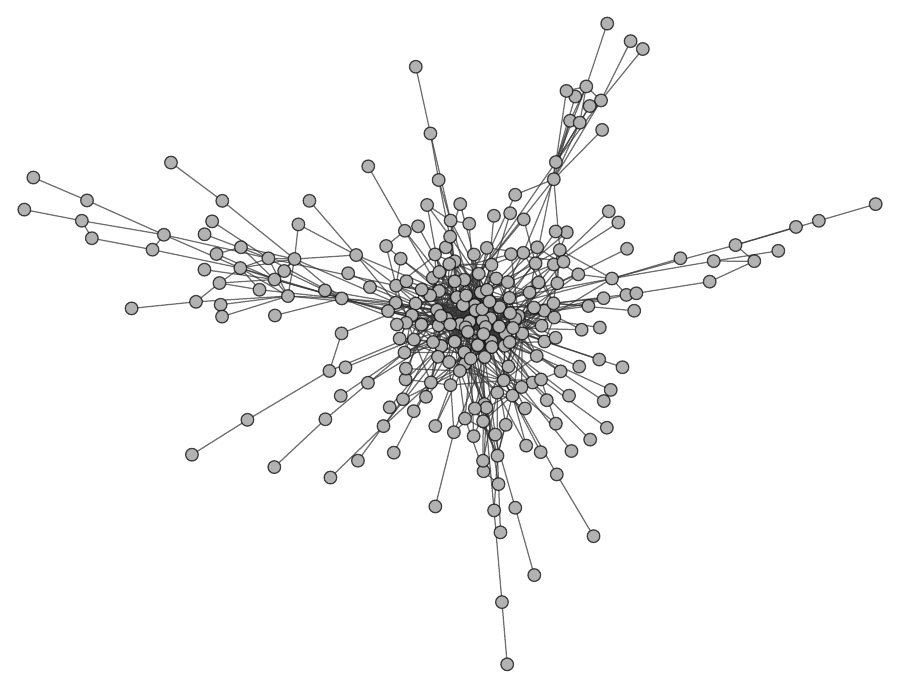}};
      \node at (-3,0) {\scriptsize (i)};
      \node at (2.2,0) {\scriptsize (ii)};
      \node at (8.1,0) {\scriptsize (iii)};
  \end{tikzpicture}
\end{center}
The data set in (i) contains pendants 
(several degree-1 vertices linked to a single vertex), isolated chains of edges, etc.
For these to arise at random requires local dependence between edges on different length scales, and 
they are conspicuously absent in (ii).
That is not a shortcoming
of the method used to fit the model, but inherent to graphons per se, since these models
constrain edges to be conditionally independent given certain vertex-wise information.
If the missing structure is relevant to statistical analysis, its absence constitutes a
form of model misspecification. To put the problem into context,
it is useful to compare some commonly used models:
\begin{itemize}
\item Graphon models (as mentioned above).
  These have recently received considerable attention in statistics, 
  and include stochastic blockmodels \citep{Goldenberg:etal:2010}, and
  many models popular in machine learning \citep[see][]{Orbanz:Roy:2015}.
\item Models that randomize a graph to match a given degree sequence,
  such as the configuration model \citep{Durrett:2006}. 
\item Models describing a formation process, such as preferential attachment (PA) models \citep{Barabasi:Albert:1999}.
\end{itemize}
All of these models constrain dependence by
rendering edges conditionally independent given some latent variable: 
In a graphon model, this latent variable is a local edge density; in the related model 
of \citet{Caron:Fox:2017},
a scalar variable associated with each vertex; in recent so-called edge exchangeable models \citep{Crane:Dempsey:2016,Williamson:2016,Cai:etal:2016},  
a random measure associated with the edges; in a configuration model, the degree sequence;
in PA models, the vertex arrival times and a vector of vertex-specific variables \citep{Bloem-Reddy:Orbanz:2017aa}. 
Some of these models 
have been noted to be misspecified for network analysis problems. For example:
\begin{itemize}
\item In the case of graphon models, various symptoms of this problem have been noted 
  \citep[e.g.][]{Borgs:Chayes:Cohn:Zhao:2014:1,Orbanz:Roy:2015}:
  The generated graph is dense; unless it is $k$-partite or disconnected, the distance between any two
  vertices in an infinite graph is almost surely $1$ or $2$; etc.
  A graphon can encode any fixed pattern on some 
  number $n$ of vertices, but
  this pattern then occurs on every possible subgraph of size $n$ with fixed probability.
\item Configuration models are popular in probability (due to their simplicity), 
  but have limited use in statistics unless the quantity of interest is the degree sequence itself:
  They are ``maximally random'' given the degrees, in a manner similar to exponential families being 
  maximally random given a sufficient statistic, and are thus insensitive to any structure not captured by the degrees. Other models based on degrees (PA models, the model of \citet{Caron:Fox:2017}, so-called rank one edge exchangeable models \citep{Janson:2017aa}) are similarly constrained.
\end{itemize}
In each instance of misspecification listed, the constraints on dependence are culpable. 
Nonetheless, such constraints are imposed for
good reasons. One is tractability: As in any system of multiple
random variables, forms of stochastic dependence amenable to mathematical analysis and statistical inference
are rare. 
Another is the basic design principle that a statistical model should be
as simple as possible: The presence, not absence, of dependence
requires justification.

The approach taken here is to make a modeling assumption on the network formation
process: Edges are formed via interactions mediated by the network, in the form of random walks. Consequently, no known form of exchangeability is present in random walk models. 
Whether that modeling assumption is considered adequate must depend on the problem at hand. 
A non-exchangeable model is best suited to data for which the order of the edges might be informative, even if it is unknown. This type of problem arises, for example, in genetics and evolutionary biology \citep[e.g.][]{Navlakah:Kingsford:2011}. 
Conceptually, the 
random walk models studied in the following allow slightly more intricate forms of stochastic dependence
than the models discussed above, without sacrificing applicability. 
\\

\noindent\textit{Article overview}.
Random walk models are defined in \cref{sec:model}.
They can be chosen to generate either multigraphs 
or simple graphs, which are sparse in either case.
A quantity that plays a key role in both theory and inference is the history of the graph,
\ie the order $\Pi$ in which edges are inserted. In dynamic networks 
\citep[e.g.][]{Kolaczyk:2009,Durrett:2006}, 
where a graph is
observed over time, $\Pi$ is an observed variable. If only the final graph is observed, $\Pi$
is latent.
\cref{sec:properties} establishes some theoretical properties:
\begin{itemize}
\item If the entire history $\Pi$ of a multigraph is observed,
  model parameters can be estimated by maximum likelihood (\cref{sec:mle}).
\item Under suitable conditions, the limiting degree distribution can be characterized (\cref{thm:degree:distribution}).
  The generated graphs can exhibit power law properties.
\item Conditionally on the times at which vertices are inserted, 
  the normalized degree sequence converges to an almost sure limit. 
  Results of this type are known for preferential attachment models. 
  Each limiting relative degree can be generated marginally by a sampling scheme reminiscent of
  ``stick-breaking''. See \cref{thm:degree:sequence}.
\item If the length of the random walk tends
  to infinity, the model converges to a generalization of the preferential attachment model
  (\cref{prop:limiting_ACL}).
\end{itemize}
If only the final graph is observed, inference is still possible, by treating the
history $\Pi$ as a latent variable and imputing it using a sampling algorithm.
This problem is of broader interest, since it permits
the application of sequential or dynamic network models, including the PA model and many
others, to a single observed graph.
In \cref{sec:inference}, we show how the graph history can be sampled, under any sequential network model
satisfying a Markov and a monotonicity property, using a Sequential Monte Carlo (SMC) algorithm. Using this algorithm, we derive inference algorithms for sequential network models
that can be designed as SMC or particle
MCMC methods \citep{Andrieu:Doucet:Holenstein:2010}. 
In \cref{sec:experiments}, we report empirical results and applications to data:
\begin{itemize}
\item Applicability of the model is evaluated by performing parameter inference
  on synthetic data sampled from the model. Parameters indeed can be recovered from
  a single observed graph.
\item The role of the random walk parameter can be understood in more detail by
  relating it to the mixing time of a simple random walk on the data set; see
  \cref{sec:experiments:scale}.
\item Comparisons between the random walk model and other network models are given
  in \cref{sec:model:fitness}; we also discuss issues raised by such how comparisons are performed.
\item The latent order $\Pi$ can be related to measures of vertex centrality (\cref{sec:vertex:centrality}).
\end{itemize}
Data and code are available at \url{https://github.com/ben-br/random_walk_smc}.

\section{Random walk graph models: Definition}
\label{sec:model}

Throughout, $g$ generically denotes a fixed graph, and $G$ a random graph.
A \kword{random walk} on a connected graph $g$, started at vertex $v$, 
is a random sequence ${(v,V_1,\ldots,V_k)}$ of vertices, such that neighbors in the sequence
are connected in $g$. 
Figuratively, a walker repeatedly moves along an 
edge to a randomly selected neighbor, until
$k$ edges have been crossed. Edges and vertices may be visited multiple times.
The random walk is \kword{simple} if the next vertex is selected with uniform probability
from the neighbors of the current one. The distribution of the terminal vertex $V_k$ of
a simple random walk on $g$, of length $k$, and starting at $v$, is denoted ${\SRW(g,v,k)}$. 
The model considered in the following is, like many sequential network models, most easily
specified as an algorithm:
\begin{poalgorithm}[Random walk model on multigraphs]
\label{scheme:multigraph}
\begin{itemize}
\item Fix ${\alpha\in(0,1]}$ and a distribution $P$ on $\mathbb{N}$. The initial
  graph $G_1$ consists of a single edge connecting two vertices.
\item For ${t=2,\ldots,T}$, generate $G_{t}$ from ${G_{t-1}}$ as follows:
\begin{enumerate}[label=(\arabic*)]
    \item Select a vertex $V$ of $G_{t-1}$ at random (see below). 
    \item With probability $\alpha$, attach a new vertex to $V$.
    \item Else, draw ${K\sim P}$ and connect $V$ to 
    ${V'\sim\SRW(G_{t-1},V,K)}$.
\end{enumerate}
\end{itemize}
\end{poalgorithm}
\noindent
The vertex $V$ in (1) is either sampled uniformly from the current vertex set,
or from the \kword{size-biased} (or \kword{degree-biased}) distribution
${\mathbb{S}(v)=\deg(v)/\sum_{v'\in \vertexset(g)}\deg(v')}$, 
where $\deg(v)$ is the degree of vertex $v$ in $g$.
In principle, the initial graph $G_1$ can be chosen as an arbitrary, fixed
``seed graph'', but we work with a single edge throughout.
The random walk models endogenous formation
of edges along existing paths in the graph---for illustration, 
consider a meeting process in a social network, where users
meet through one or several acquaintances.

\cref{scheme:multigraph} generates multigraphs, since the two vertices 
selected by the random walk may already be connected (resulting
in a multi-edge), or may coincide (resulting in a self-loop).
To generate simple graphs instead, step (3) above is replaced by:
\begin{postep}
  \begin{enumerate}[label=(\arabic*)]
    \addtocounter{enumi}{2}
    \renewcommand{\labelenumi}{(\arabic{enumi}')}
  \item
    Else, draw ${K\sim P}$ and 
      ${V'\sim\SRW(G_{t-1},V,K)}$.
      Connect $V$ to $V'$ if they are distinct and not connected; 
      else, attach
      a new vertex to $V$.    
  \end{enumerate}
\end{postep}
Either sampling scheme defines the law of a sequence ${(G_1,G_2,\ldots)}$ of graphs.
We denote this law ${\RWU(\alpha,P)}$ if the vertex $V$ in (1) is chosen uniformly,
or ${\RWSB(\alpha,P)}$ in the size-biased case. Each such law is defined both on
multigraphs and on simple graphs, depending on which sampling scheme is used.

We usually choose the length of the random walk as a Poisson variable: 
Denote by $\Poisson(\lambda)$ the Poisson distribution with parameter $\lambda$, shifted by 
1, \ie the law of ${K+1}$, where $K$ is Poisson. We write ${\RW(\alpha,\lambda)}$ for
${\RW(\alpha,\Poisson(\lambda))}$ where convenient. Examples of graphs generated by this distribution 
are shown in Figure~\ref{fig:examples:uniform}.

\begin{figure}[tb]
  \makebox[\textwidth][c]{
    \resizebox{\textwidth}{!}{
    \begin{tikzpicture}
  \begin{scope}
    \begin{scope}[xshift=-5cm]
      \node  (gt) {
        \includegraphics[width=4cm,angle=90]{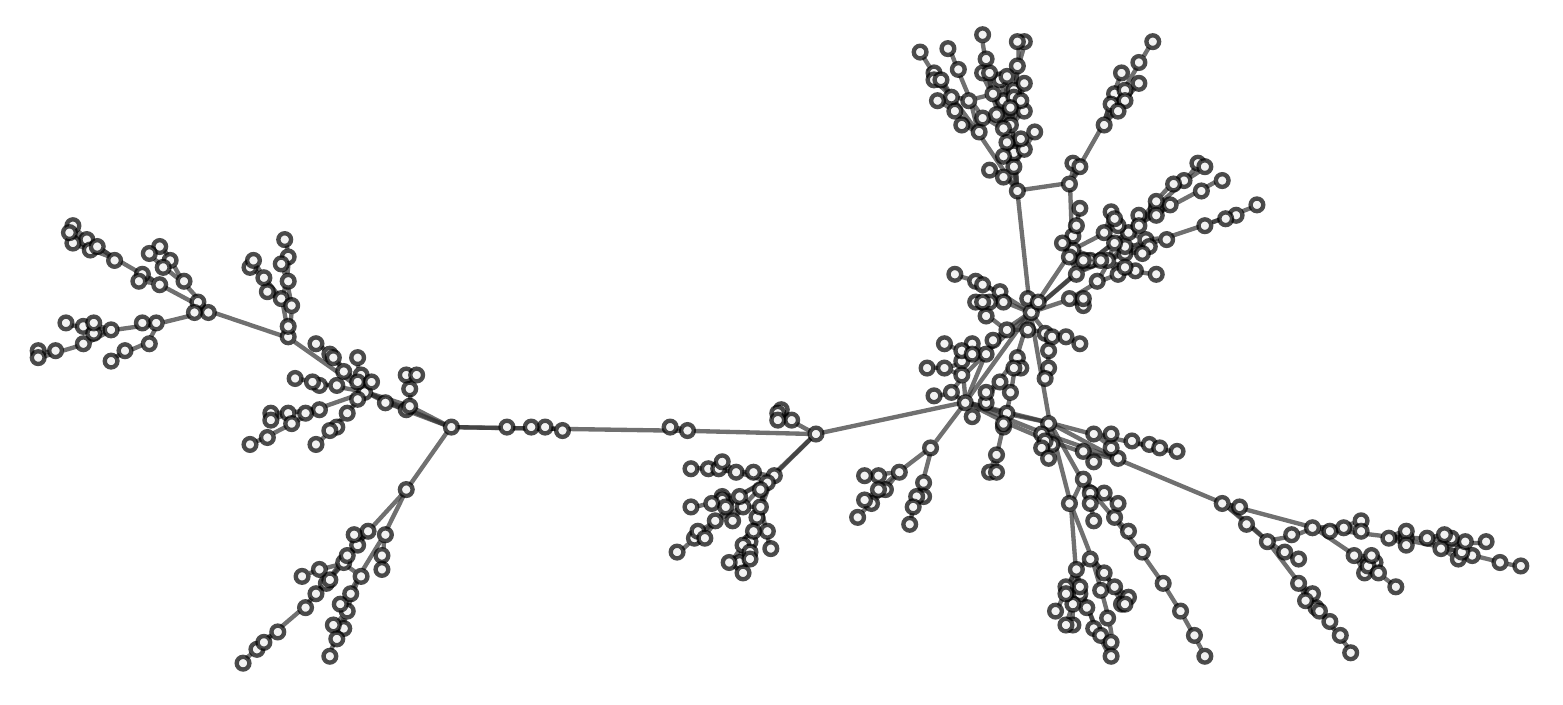}
      };
       \node [below=1.9cm, align=flush center] at (gt) {
	   \footnotesize{{$\alpha=0.5, \lambda=2$}}
      };
    \end{scope}
    \begin{scope}[xshift=-2cm]
      \node  (g5) {
        \includegraphics[width=4cm,angle=90]{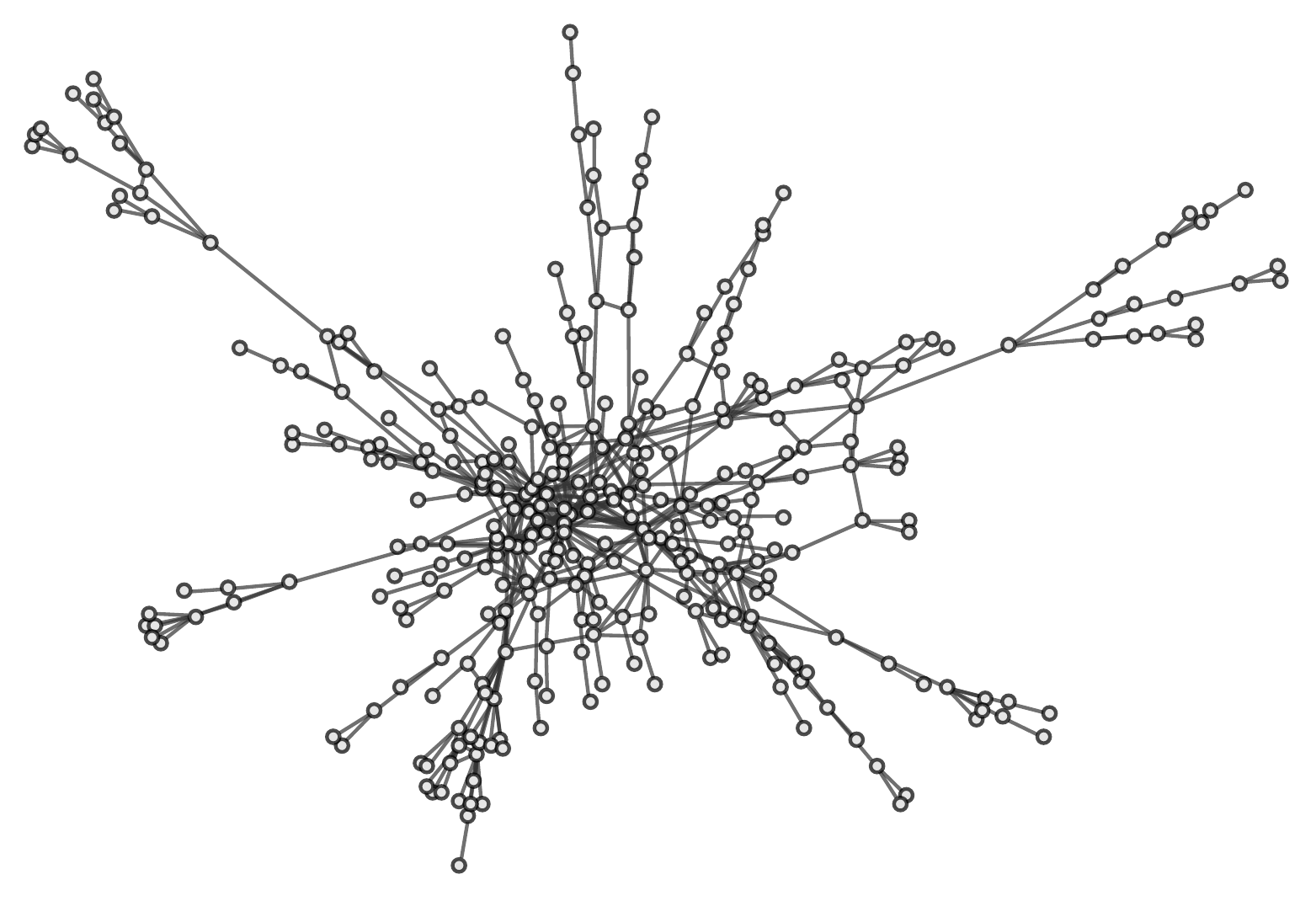}
      };
      \node [below=1.9cm, align=flush center] at (g5) {
     \footnotesize{{$\alpha=0.5, \lambda=4$}}
      };
    \end{scope} 
    \begin{scope}[xshift=1cm]
      \node  (g5) {
        \includegraphics[width=4cm,angle=90]{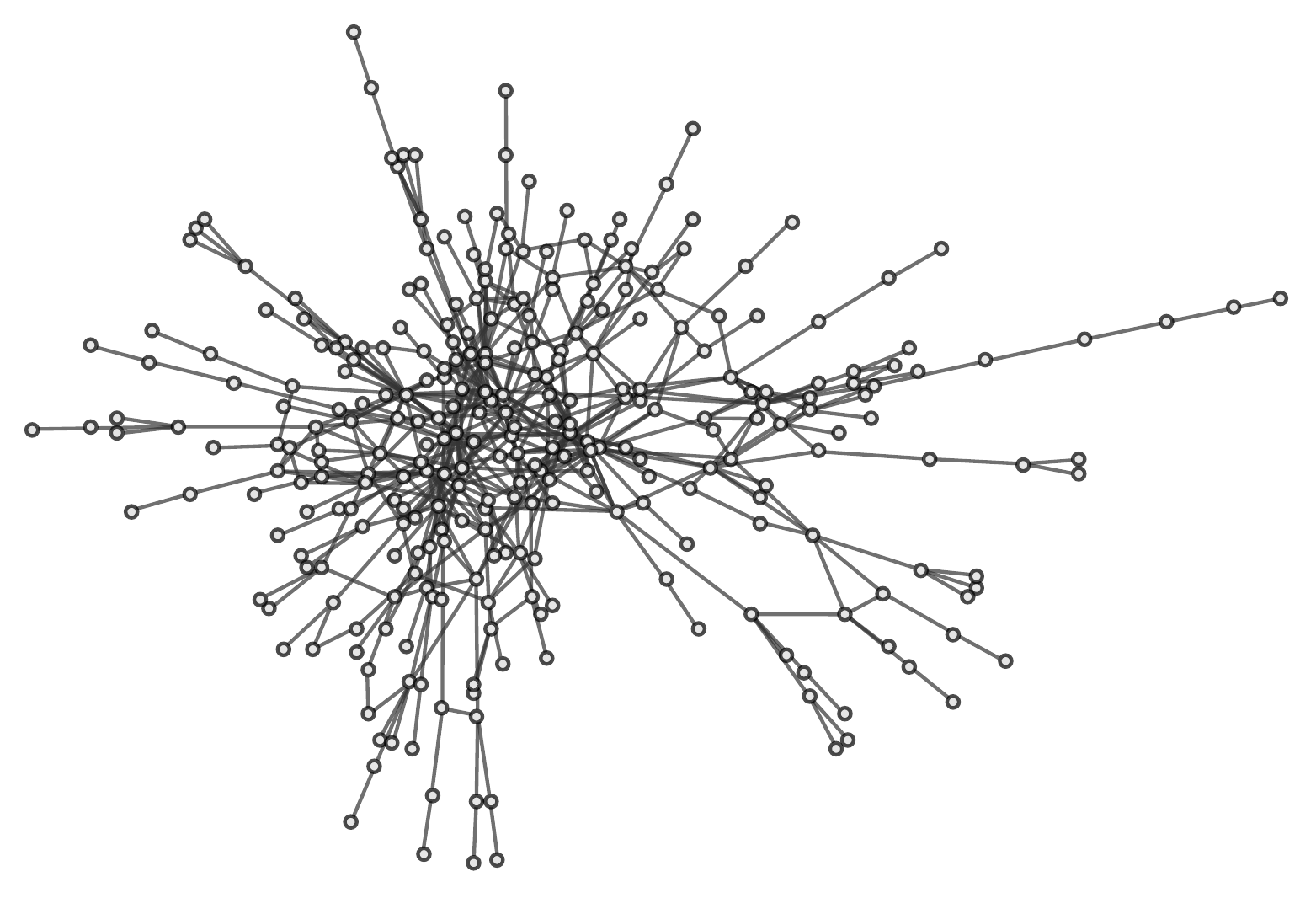}
      };
      \node [below=1.9cm, align=flush center] at (g5) {
     \footnotesize{{$\alpha=0.5, \lambda=8$}}
      };
    \end{scope} 
    \begin{scope}[xshift=-8cm]
      \node (g1){
        \includegraphics[width=4cm,angle=90]{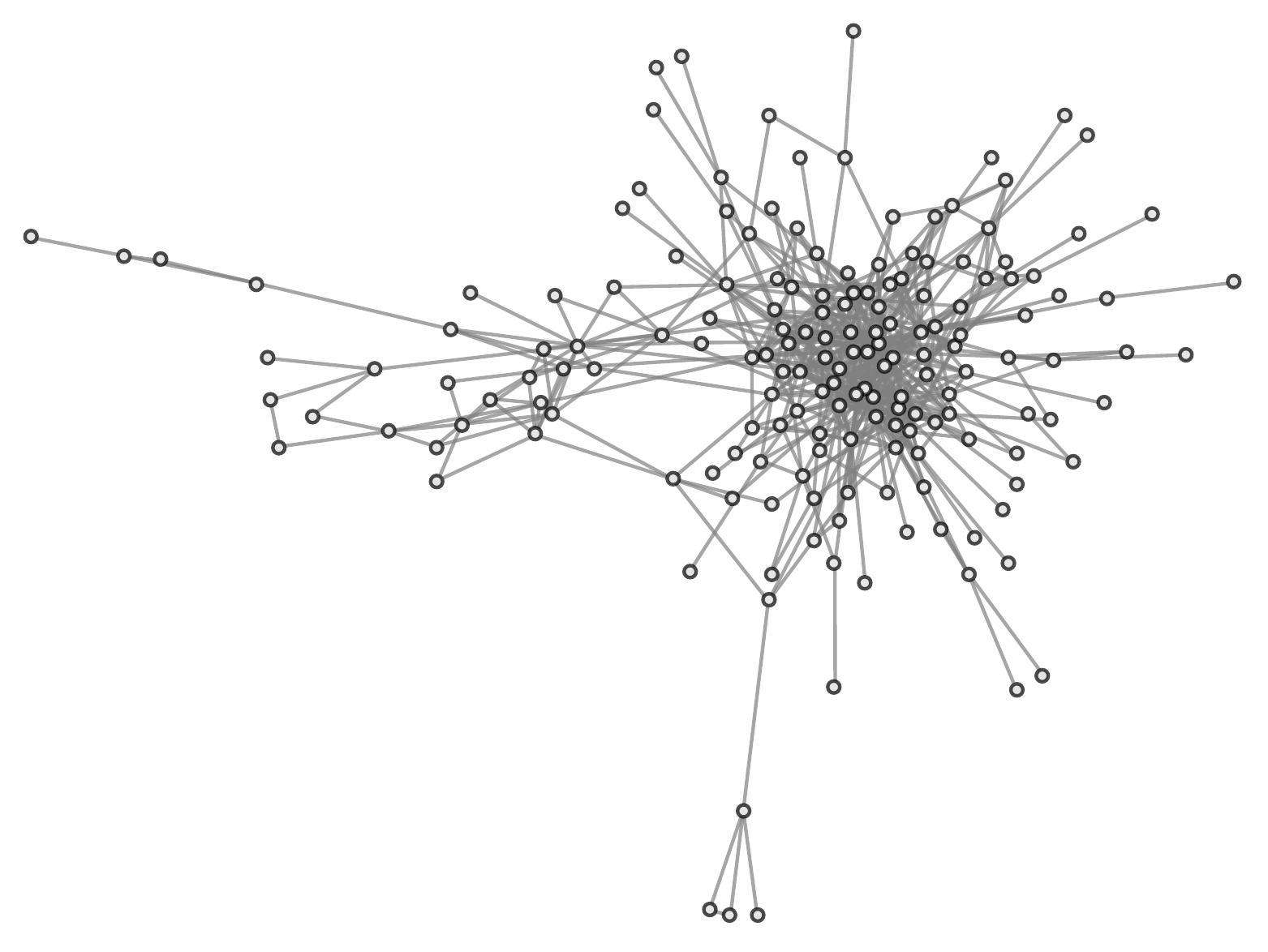}
       
      };
      \node [below=1.9cm, align=flush center] at (g1) {
      	 \footnotesize{{$\alpha=0.1, \lambda=4$}}
      };
    \end{scope}
    \begin{scope}[xshift=4cm]
      \node  (g3) {
        \includegraphics[width=4cm,angle=90]{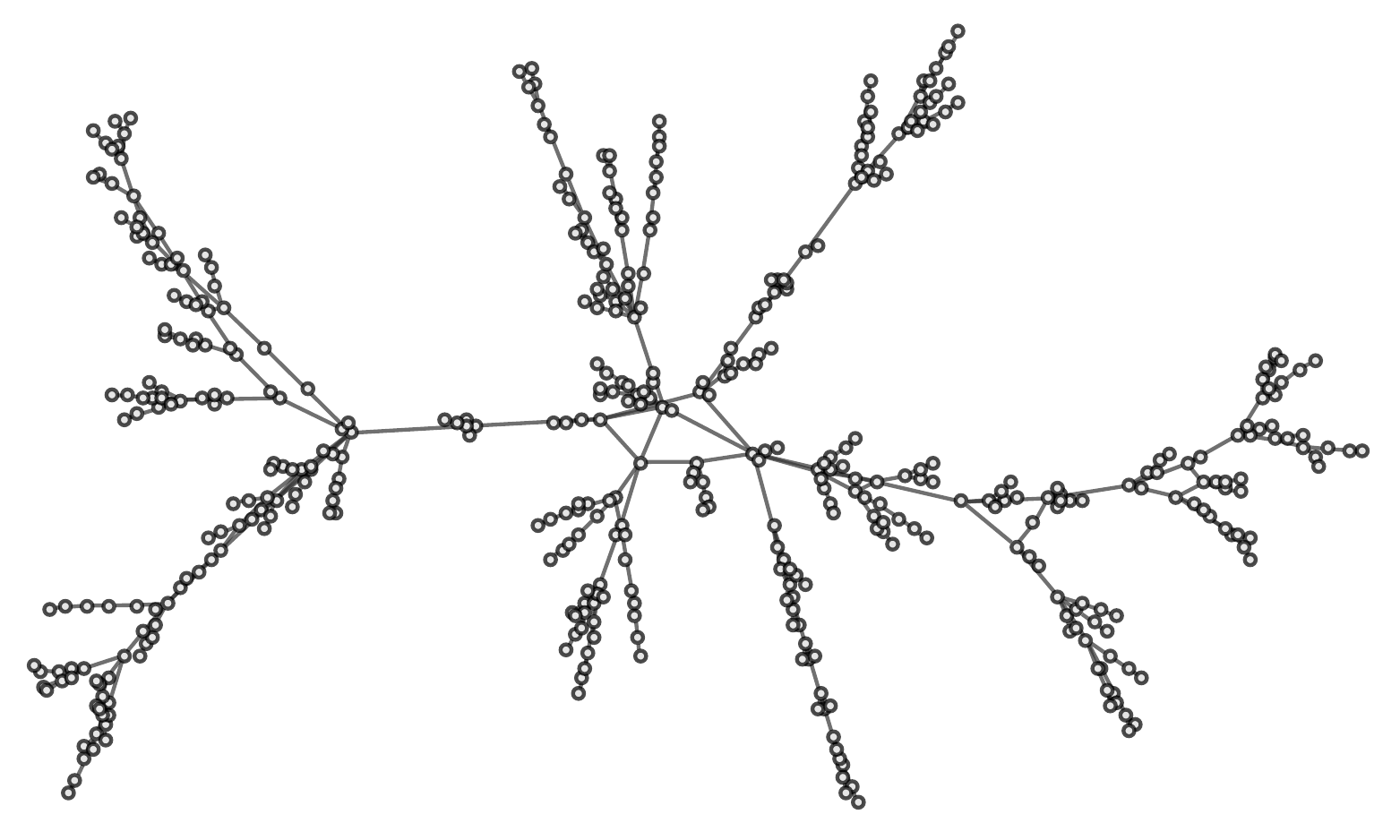}
      };
      \node [below=1.9cm, align=flush center] at (g3) {
	   \footnotesize{{$\alpha=0.9, \lambda=4$}}
      };
    \end{scope}
  \end{scope}
  \begin{scope}[yshift=-4.8cm]
    \begin{scope}[xshift=-5cm]
      \node  (gt) at (.3,0) {
        \includegraphics[width=4cm,angle=90]{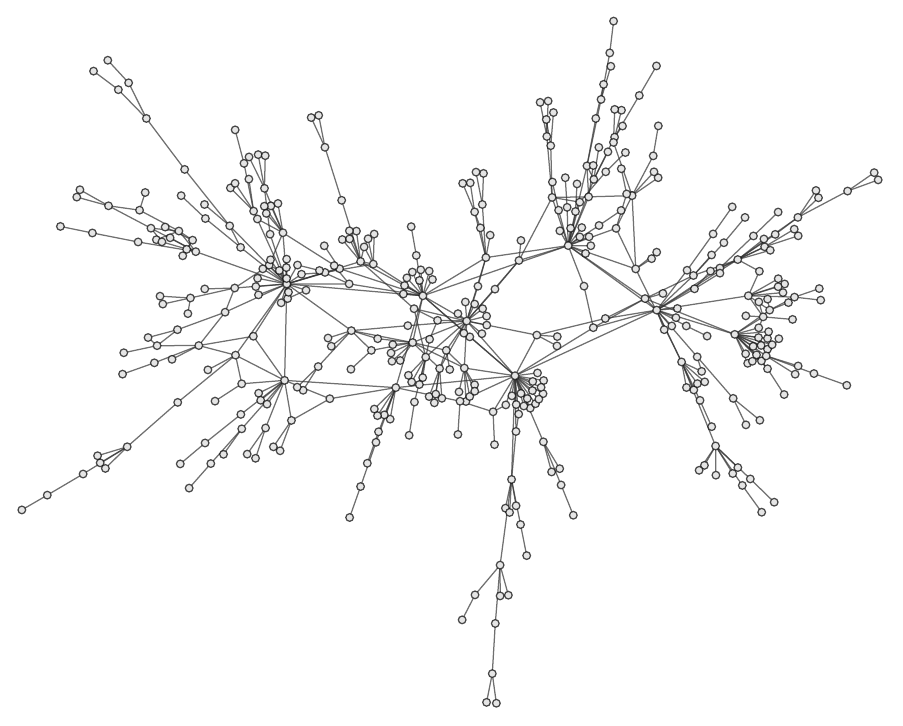}
      };
       \node [below=1.9cm, align=flush center] at (gt) {
	   \footnotesize{{$\alpha=0.5, \lambda=2$}}
      };
    \end{scope}
    \begin{scope}[xshift=-2cm]
      \node  (g5) {
        \includegraphics[width=4cm,angle=90]{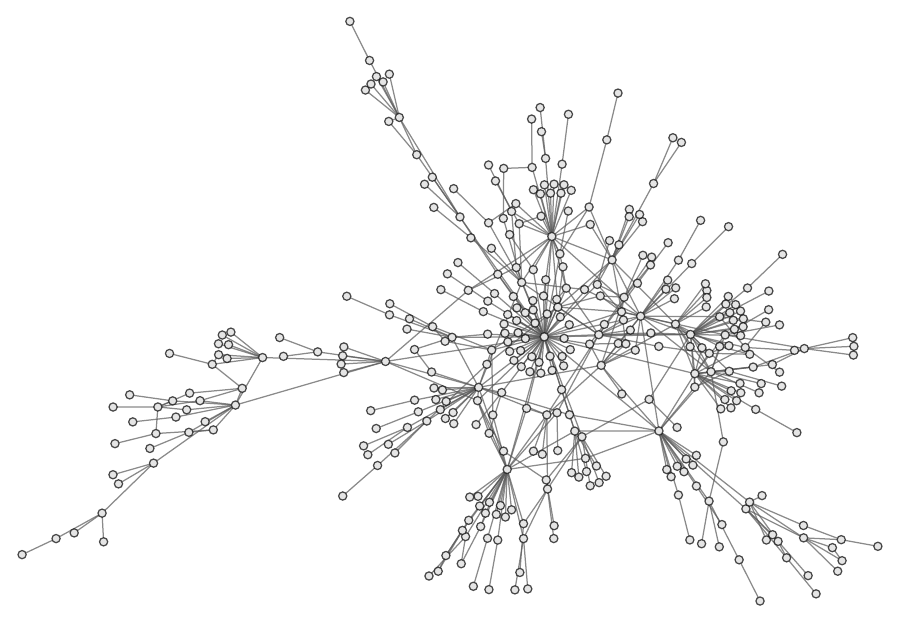}
      };
      \node [below=1.9cm, align=flush center] at (g5) {
     \footnotesize{{$\alpha=0.5, \lambda=4$}}
      };
    \end{scope} 
    \begin{scope}[xshift=1cm]
      \node  (g5) {
        \includegraphics[width=4cm,angle=90]{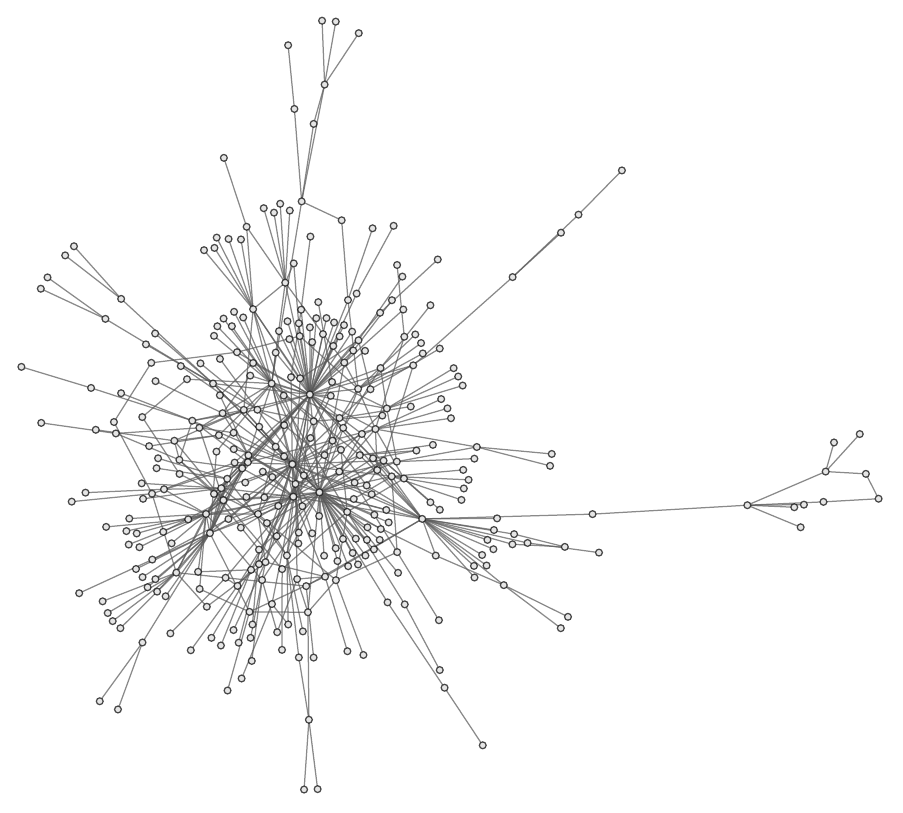}
      };
      \node [below=1.9cm, align=flush center] at (g5) {
     \footnotesize{{$\alpha=0.5, \lambda=8$}}
      };
    \end{scope} 
    \begin{scope}[xshift=-8cm]
      \node (g1){
        \includegraphics[width=4cm,angle=90]{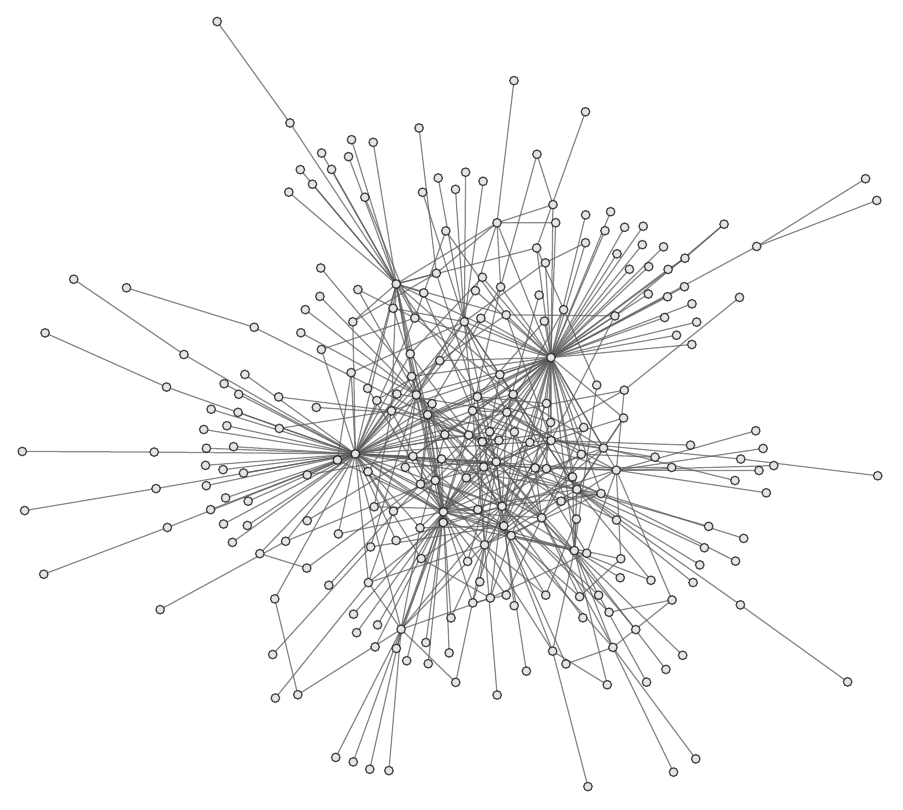}
       
      };
      \node [below=1.9cm, align=flush center] at (g1) {
      	 \footnotesize{{$\alpha=0.1, \lambda=4$}}
      };
    \end{scope}
    \begin{scope}[xshift=4cm]
      \node  (g3) {
        \includegraphics[width=4cm,angle=90]{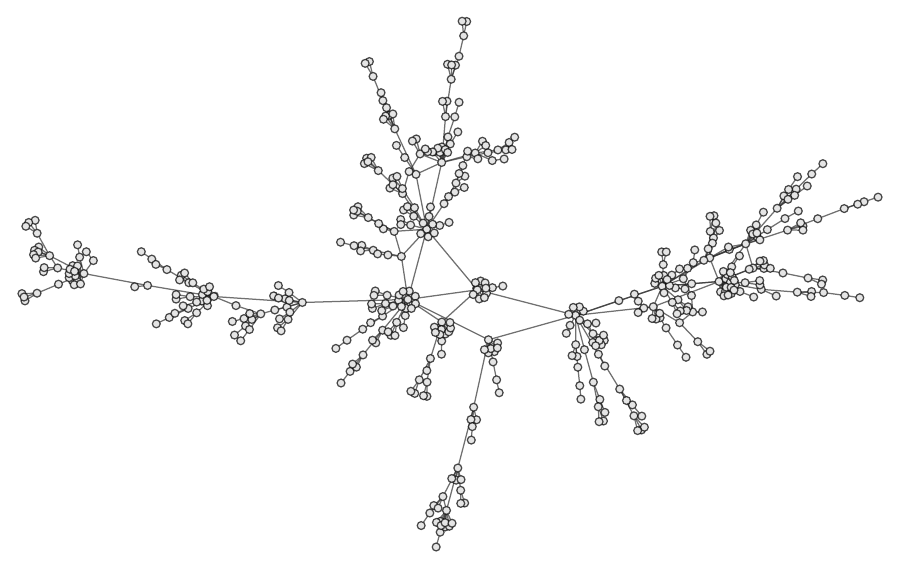}
      };
      \node [below=1.9cm, align=flush center] at (g3) {
	   \footnotesize{{$\alpha=0.9, \lambda=4$}}
      };
    \end{scope}
  \end{scope}
\end{tikzpicture}

  }}
  \caption{Examples of simple graphs generated by a random walk model: 
    $\RWU(\alpha,\text{Poisson}_+(\lambda))$ distribution (top row),
    and by a $\RWSB(\alpha,\text{Poisson}_+(\lambda))$ distribution (bottom row).}
  \label{fig:examples:uniform}\label{fig:examples:sb}
\vspace{-.5cm}
\end{figure}
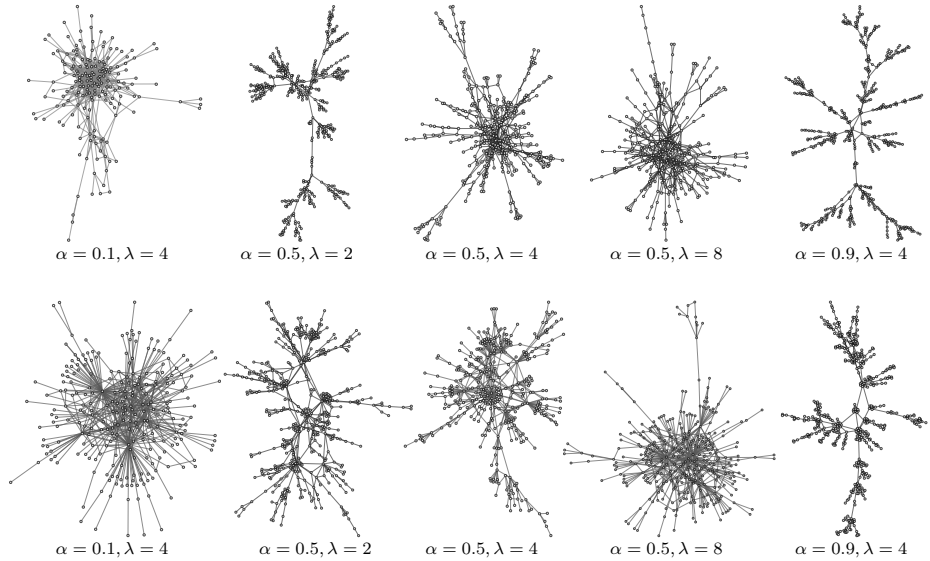

\section{Model properties}
\label{sec:properties}

We first observe that graphs generated by the model are sparse by construction.
Let ${(G_1,\ldots,G_t,\ldots)}$ be a graph sequence generated by any $\RW$ model on simple or multigraphs,
with parameter $\alpha$, where $G_1$ has a single edge connecting two vertices. Denote by $\vertexset(G_t)$ and $\edgeset(G_t)$ the set of vertices and of edges, respectively, in $G_t$. For any $\RW$ model, each vertex is added with one edge and there is a constant positive probability at each step of adding a new vertex, which leads to the following:
\begin{observation}
  Graphs generated by the random walk model are sparse:
  In a sequence ${(G_1,G_2,\ldots)\sim\RW(\alpha,P)}$, the number of edges grows as
  ${|\edgeset(G_t)|=\Theta(|\vertexset(G_t)|)}$.
\end{observation}
If the model is run for an infinite number of steps (with ${\alpha>0}$), it generates an infinite
number of vertices. Sparsity implies this infinite vertex set is not exchangeable. Since
vertices inserted earlier have larger expected degree, finite subsequences of vertices are not
exchangeable, either.

\subsection{Mixed random walks and the graph Laplacian}

Most results in the sequel involve the law of a simple random walk. For a walk of fixed
length, this law is determined by the graph's Laplacian matrix. 
If the length is random, one has to mix against its distribution; 
as we show below, the mixed law of the random walk still has closed form. 

Consider an undirected graph $G$ with $n$ vertices, adjacency matrix
$A$ and degree matrix ${D=\text{\rm diag}(\deg(v_1),\ldots,\deg(v_n))}$.
The probability that a simple random walk started at a vertex $u$ terminates at $v$ after exactly $k$
steps is the entry $(u,v)$ of the matrix ${(D^{-1}A)^k}$.
If the length of the walk is a random variable $K$ with law $P$, the marginal probability of reaching $v$ from $u$ is hence
\begin{align}
  \label{eq:mixed:rw:probability}
  \P\lbrace V_{end}=v|V_0=u,\, G\rbrace 
  \quad=\quad
  \bigl[ \textstyle\sum_{k}(D^{-1}A)^{k} P\lbrace K=k\rbrace \bigr]_{uv}\;
  = \bigl[ \E_P\bigl[\left(D^{-1}A\right)^k\bigr] \bigr]_{uv} \;.
\end{align}
A useful way to represent the information in ${D^{-1}A}$ is as the matrix
\begin{equation}
  \mathbf{L}:=D^{1/2}(\iden_n-D^{-1}A)D^{-1/2}\qquad\text{ with }\iden_n:=\text{identity matrix}\;,
\end{equation}
known as the \kword{normalized graph Laplacian}
\citep[e.g.,][]{Chung:Yau:2000:1}. Letting the probability generating function (p.g.f.) of $K$ be $H_P(z) = \E_P[z^K]$, the law of a random walk of random length is obtained as follows:
\begin{proposition} \label{prop:mixed_random_walks}
  Let $g$ be a connected, undirected graph on $n$ vertices, and let $K$ have law $P$ with p.g.f. $H_P(z)$. Let $(\LeigVal_i)_{i=1}^n$ be the eigenvalues, and $(\LeigVec_i)_{i=1}^n$ the eigenvectors, of $\L$. Then for a simple random walk ${(V_0,\ldots,V_{K})}$
  of random length $K$ on $g$,
  \begin{equation} \label{eq:Q:matrix}
      \P\lbrace V_{end}=v|V_0=u,\, g\rbrace 
      =
      \bigl[
        D^{-1/2}\mathcal{K}_P D^{1/2}
        \bigr]_{uv} \quad \text{with} \quad \mathcal{K}_P := \sum_{i=1}^{n} H_P(1 - \LeigVal_i)\LeigVec_i \LeigVec_i' \;.
  \end{equation}
\end{proposition}
For a $\Poisson(\lambda)$ length, the result specializes to ${\mathcal{K}_{\lambda +}:=(\iden_n - \L)e^{-\lambda\L}}$. If
over-dispersion is a concern, one can replace the Poisson by a negative binomial distribution with parameters
$r$ and $p$, again shifted by one to the positive integers, and denoted $\NB(r,p)$. In that case, one obtains
${\mathcal{K}_{r,p +}:=(\iden_n - \L)\bigr(\iden_n+{\textstyle\frac{p}{1-p}}\L\big)^{-r}}$. 
These matrices 
arise in other contexts: 
The \kword{heat kernel} $\mathcal{K}_{\lambda}=e^{-\lambda\L}$ and the \kword{regularized Laplacian kernel} $\mathcal{K}_{r,p} = \bigr(\iden_n+{\textstyle\frac{p}{1-p}}\L\big)^{-r}$ have applications in collaborative filtering,
semi-supervised learning and manifold learning \citep{Smola:Kondor:2003,Lafferty:Lebanon:2005}.
Both act as smoothing operators on functions defined
on the vertex set, by damping the eigenvalue spectrum \citep{Smola:Kondor:2003,Vishwanathan:etal:2010}. 
The expected length of the random walk can be interpreted as a length
scale (see \cref{sec:experiments}), and a similar idea was previously used by \citet{Pons:Latapy:2006} to define a distance between vertices of a graph: if the distribution over the terminal vertex of a random walk of (fixed) length $k$ started from vertex $u$ is close to that started from vertex $v$, then $u$ and $v$ are ``close.'' 
It can be regularized by mixing the walk length against a Poisson distribution, as suggested in \citet[][Sec. 3.4]{Pons:Latapy:2006}. We note that, using \cref{prop:mixed_random_walks} above, the Pons--Latapy metric can be regularized more generally by an arbitrary walk length distribution.

\subsection{Maximum likelihood estimation for fully observed sequences}
\label{sec:mle}

In the special case of ``dynamic'' multigraphs, where the entire history is
observed, model parameters can be estimated by maximum likelihood.
Let ${(v_s,v_s')}$ be the vertices connected in step ${s}$, and 
${B_s:=\indicator\braces{\min\braces{\deg_{s-1}(v_s),\deg_{s-1}(v_s')} = 0}}$.
That is, the edge was generated by step (2) of the sampling scheme if and only if ${B_s=1}$. Define
\begin{equation}
  \label{eq:MLE:RW:probability}
  \RWprob^{\lambda}_s:=
  D_s^{-1/2}\mathcal{K}_{\lambda +,s} D_s^{1/2} \;,
\end{equation}
and denote by $\bar{d}_s(v)$ the relative degree of vertex $v$ in $G_s$.
\begin{proposition}
  Let ${(G_1,\ldots,G_t)\sim\RWSB(\alpha,\Poisson(\lambda))}$ be a random multigraph sequence,
  where $G_1$ has a single edge connecting two vertices. Then
  \begin{equation*} 
    \hat{\alpha}_t:=\frac{N_t-2}{t-1}
    \quad\text{ and }\quad
    \hat{\lambda}_t:=\arg\max_{\lambda\geq 0} \prod_{s=2}^t\bigl(
        [Q^{\lambda}_{s-1}]_{v_s v_s'} \bar{d}_{s-1}(v_{s})
        +
        [Q^{\lambda}_{s-1}]_{v_s' v_s} \bar{d}_{s-1}(v_{s}')
        \bigr)^{1 - B_s}
  \end{equation*}
  are maximum likelihood estimators of $\alpha$ and $\lambda$, respectively.
\end{proposition}
The same holds for $\RWU(\alpha,\Poisson(\lambda))$, with ${1/N_{s-1}}$ substituted for 
$\bar{d}_{s-1}(v_s)$ in $\hat{\lambda}_t$. To change the walk length distribution,
one must replace $\mathcal{K}_{\lambda+}$
in \eqref{eq:MLE:RW:probability} with the appropriate matrix from \eqref{eq:Q:matrix}.
For simple graphs, maximum likelihood estimation is still possible in principle, but 
more complicated since the effects of $\alpha$ and $\lambda$ are no longer 
independent. This manifests in step (3') in \cref{sec:model}, 
where the probability of generating an additional vertex 
depends on the outcome of the random walk. Conceptually,
it is due to the fact that a simple graph censors evidence of
repeatedly connecting two vertices.

\subsection{Asymptotic degree properties} \label{sec:asymptotic:degrees}

A property of network models thoroughly studied in the theoretical literature 
is how vertex degrees behave as the graph grows large. For
PA graphs, limiting degree distributions can
be determined analytically \citep{Durrett:2006}.
Our next results describe analogous properties for random walk models. 
In the proofs, effects of the random walk are mitigated by using invariance of the degree-biased distribution
under $\mathcal{K}_P$ in \eqref{eq:Q:matrix} %$\mathcal{K}^{\lambda}$ and $\mathcal{K}^{r,p}$ 
to reduce to techniques developed for preferential attachment.

Consider a $\RWSB$ random multigraph, and 
let $p_d$ be the probability that a vertex sampled uniformly at random has degree $d$. 
The proof of the next result shows that, if the average probability (over vertices) of 
inserting a self-loop vanishes for large $t$, then
\begin{align} \label{eq:deg:distn}
    p_d 
    & = \rho
    \frac{\Gamma(d) \Gamma(1+\rho)}{\Gamma(d + 1 + \rho)} 
    \qquad\text{ where }\qquad
    \rho:=1+\frac{\alpha}{2-\alpha}\;.
\end{align}
as ${t\rightarrow\infty}$. This is the \kword{Yule--Simon distribution}
with parameter ${\rho}$ \citep[e.g.][]{Durrett:2006}. 
To state the result, let $\vertexset_{d,t}$ be the set of vertices in $G_t$ with degree $d$, and $m_{d,t} = |\vertexset_{d,t}|$.
\begin{theorem}[Degree distribution] \label{thm:degree:distribution}
  Let a sequence of multigraphs ${(G_1,G_2,\dots)}$ have law $\RWSB(\alpha,P)$, 
  for some distribution $P$ on $\N$, such that for each $d\in\mathbb{N}$,
  \begin{equation}
    \label{condition:self:loops}
    \frac{1}{m_{d,t}}\sum_{v\in\vertexset_{d,t}} P(V_{\text{end}}=v|V_0=v,G_t)
    = o(1) \qquad\text{ as } t\to\infty \;.
  \end{equation}
 Then the scaled number of vertices in $G_t$ with degree $d$ follows a power law in $d$ with exponent ${1+\rho}$. 
 In particular, ${\frac{m_{d,t}}{\alpha t}\rightarrow p_d}$ in probability, for all ${d\in\mathbb{N}_+}$
 as ${t\rightarrow\infty}$.
\end{theorem}

\noindent Condition \eqref{condition:self:loops} controls the number of self-loops, by requiring that
the probability of a random walk ending where it starts vanishes for ${t\rightarrow\infty}$, separately
for each degree $d$. Simulations indicate that it holds for many walk length distributions, including Poisson.

\def\DEG{\textbf{deg}}

Denote by ${\DEG_t:=(\deg_t(v_1),\deg_t(v_2), \dots)}$ the \kword{degree sequence}, where $v_{j}$ is the $j$-th vertex to appear in the graph sequence. 
For PA models, the limiting sequence $\DEG_{\infty}$ can be studied analytically
\citep[see e.g.][]{Durrett:2006}.
It is closely related to a P\'{o}lya urn, and like the limit of an urn, 
$\DEG_{\infty}$ is itself a random variable. Informally, edges created early have 
sufficiently strong influence on later edges
that randomness does not average out asymptotically.
The joint law of $\DEG_{\infty}$ can be obtained explicitly
\citep{Mori:2005,Pekoz:Rollin:Ross:2017}.
In some cases, it also admits constructive representations:
Using only sequences of independent elementary random variables, one can generate a random sequence whose joint 
distributions
are identical to those of the limiting degree sequence \citep{Mori:2005,James:2015aa,Bloem-Reddy:Orbanz:2017aa}. Such representations are
closely related to ``stick-breaking constructions'' \citep[e.g.][]{Pitman:2006}.

The next result similarly describes the limiting degree sequence $\DEG_{\infty}$ of the $\RWSB$ model. 
The relevant technical tool is to condition
on the order in which edges occur: For a graph sequence $(G_1,G_2,\ldots)$,
denote by $S_j$ the time index at which vertex $v_j$ is 
inserted in the model (that is, $G_{S_j}$ is the first graph containing $v_j$), and let ${S_{1:r}:=(S_1,\ldots,S_r)}$.
Let $\Sigma_{\beta}$ be a positive stable random variable with index $\beta\in(0,1)$, with Laplace transform ${\E[ e^{-t\Sigma_{\beta}} ] = e^{-t^{\beta}}}$, and $f_{\beta}$ its density. Define $\Sigma_{\beta,\theta}$, for $\theta > -\beta$, as a random variable with the polynomially tilted density  $f_{\beta,\theta}(s) \propto s^{-\theta} f_{\beta}(s)$. The variable ${M_{\beta,\theta}:=\Sigma_{\beta,\theta}^{-\beta}}$ is said to have \kword{generalized Mittag--Leffler distribution} 
\citep{Pitman:2006,James:2015aa}.

\begin{theorem}[Degree sequence] \label{thm:degree:sequence}
  Let a sequence of multigraphs ${(G_1,G_2,\dots)}$ have law $\RWSB(\alpha,P)$, 
  for some distribution $P$ on $\N$.
  Conditionally on $(S_1,S_2,\ldots)$, the 
  scaled degree sequence converges jointly to a random limit: For each
  $r\in\N_+$, 
    \begin{equation} \label{eq:deg:seq}
      t^{-1/\rho}(\deg_t(v_1), \deg_t(v_2), \dots, \deg_t(v_r))\,\big\vert\,S_{1:r}
      \quad\xrightarrow{t\rightarrow\infty}\quad (\xi_1, \xi_2, \dots, \xi_r)\,\big\vert\,S_{1:r}
    \end{equation}
    almost surely.
    Each limiting conditional law ${\mathcal{L}(\xi_j \mid S_j)}$, for ${j\geq 1}$, 
    can be represented constructively: 
    Let ${M_j\sim\MittagLeffler\left(\rho^{-1}, S_j-1 \right)}$, 
    ${B_j\sim\BetaDist\left(1,\rho(S_j-1)\right)}$, and \\${\psi_j \sim \BetaDist\left( S_j, S_{j+1} - S_j \right)}$ be conditionally independent given $S_j$. Then for $1 \leq i < j$
    \begin{equation} \label{eq:marginal:construction}
      \xi_j \big\vert S_j \;\eqdist \; M_j \cdot B_j\big\vert S_j \;\eqdist\; \xi_i \prod_{k=i}^{j-1} \psi_k^{1/\rho} \;\big\vert S_{i:j} \;.
    \end{equation}
    Each marginal law ${\mathcal{L}(\xi_j)}$ is uniquely determined by the sequence of moments
    \begin{align} \label{eq:marginal:moments}
      \E\bigr[ \xi_j^k \bigl] 
      & = \frac{ \Gamma(k+1) \Gamma(j-1) }{ \Gamma(j-1 + \frac{k}{\rho}) }
      \alpha^{\frac{k}{\rho}}
      \,_2F_1\bigl(1 + {\textstyle\frac{k}{\rho}},{\textstyle\frac{k}{\rho}}; 
      j-1+{\textstyle\frac{k}{\rho}};1-\alpha\bigr)
    \end{align}
    where $\,_2F_1(a,b;c;z)$ is the ordinary hypergeometric function.
\end{theorem}
We note the moment ${\E\bigr[ \xi_j^k \bigl]}$ scales as
${\Gamma(k+1)\alpha^{\frac{k}{\rho}} j^{-\frac{k}{\rho}} \cdot(1 + O(j^{-1}))}$ for 
${j\to\infty}$.  
The simple constructive representation \eqref{eq:marginal:construction}
seems only to exist marginally, in contrast to the 
PA model, where a recursive analogue of \eqref{eq:marginal:construction} holds even jointly 
\citep{James:2015aa}.

\subsection{Relation to preferential attachment and other models}
\label{sec:pa}

\citet{Aiello:Chung:Lu:2002} introduced a generalization of preferential attachment graphs:
Fix ${\alpha\in(0,1]}$. At each step, select two vertices ${V,V'}$ independently from
the degree-biased distribution. With probability $\alpha$, insert a new vertex and connect
it to $V$; otherwise, connect $V$ and $V'$. The model is denoted ${\ACL(\alpha)}$,
and can be regarded as a natural extension of Yule--Simon processes from sequences to graphs.
\begin{proposition} \label{prop:limiting_ACL}
  The limit in distribution ${\RWSB(\alpha,\infty):=\lim_{\lambda\to\infty} \RWSB(\alpha,\lambda)}$ exists for every $\alpha$, and ${\RWSB(\alpha,\infty)=\ACL(\alpha)}$ if both models start with the same seed graph.
\end{proposition}
This is of course due to the fact that, as ${\lambda\rightarrow\infty}$, the 
terminal vertex of simple random walk has law approaching the degree-biased distribution. The same is true for any random walk model in the limit where the walk length distribution becomes degenerate at $\infty$ (analogously to $\lambda\to\infty$). 
For ${\alpha=1}$ (and any $\lambda$), both the random walk model and the $\ACL(\alpha)$ model coincide with the Barab\'{a}si--Albert preferential attachment tree.

Random walk models generate networks one edge at a time; this specification of a graph may be viewed as a sequence of edges. Recent work on so-called edge exchangeable graphs \citep{Crane:Dempsey:2016,Williamson:2016,Cai:etal:2016,Janson:2017aa} define network models in terms of an exchangeable sequence of edges. Edges in random walk models are not exchangeable. The limiting $\ACL$ model, however, and a subclass of edge exchangeable graphs both belong to a larger class of models whose distributions are characterized by the arrival times of new vertices \citep{Bloem-Reddy:Orbanz:2017aa}.

\vspace{-0.5em}
\subsection{Covariate information}
\label{sec:covariates}

We briefly mention two possible extensions not further considered in our discussion. 
Nodal covariates $\bfX$ represent exogenous effects on edge formation; when available, they can be incorporated by biasing the probability of a random
walk started at $u$ to end at $v$ by a function of the covariate values at $u$ and $v$. 
A simple way to do so is via a kernel matrix $\calK_{\bfX}$, where $[\calK_{\bfX}]_{uv}$ expresses a suitable notion of similarity between $u$ and $v$, based on $\bfX$. To this end, let ${\calK_{P + \bfX}:=\calK_P + \calK_{\bfX}}$. Define ${\deg_{P+\bfX}(u)=\sum_{v}[\calK_{P + \bfX}]_{uv}}$, and let $D_{P+\bfX}$ be the diagonal matrix with entries ${[D_{P+\bfX}]_{uu}=\deg_{P+\bfX}(u)}$. The probability of a new edge between $u$ and $v$ is then
\begin{align*}
  \P\{ (u,v) | V_0=u, g  \} = [D_{P+\bfX}^{-1} \calK_{P + \bfX}]_{uv} \;.
\end{align*}
For example, the kernel can be chosen to express \emph{homophily}, the tendency of similar vertices to form edges \citep[e.g.][]{Hoff:2008}. Similarly, one can bias the ``popularity'' of a vertex by adding a scalar offset to each degree that depends on covariates: Choose functions $f_i$ with range $(-1,\infty)$ (to ensure the biased degrees
remain positive), and replace matrices in \eqref{eq:mixed:rw:probability} by 
${D'=\text{diag}(\deg(v_1)+f_1(\bfX),\ldots,\deg(v_n)+f_n(\bfX))}$, and $A_{i,j}'=A_{i,j}(1 + f_i/\deg(v_i))$. The inference methods developed below
remain applicable, provided the functions $f_i$ depend only on exogenous covariates, not on the graph.
Temporal covariates, such as age, can also be incorporated into the order in which vertices appear. 

\section{Particle methods for sequential network models} \label{sec:inference}

\newcommand{\prior}[1]{P_{\!\text{\scalebox{.7}{$[#1]$}}}}

Our modeling approach is to use a suitably parametrized sequential random graph
as a statistical model. One could similarly use a preferential attachment graph, or some
other sequential model considered more appropriate for a given task. The approach poses
an inference problem:
A sequential network model is assumed, but only the final graph $G_T$ generated by the model is observed,
as opposed to the entire history $\Gall:=(G_1,\ldots,G_T)$ of the generative process. 
This section develops Markov chain Monte Carlo methods for inference in sequential random graph models,
and in random graph models in particular.

Given the choice of $G_1$---typically a graph consisting of a single edge---$T$ is always
uniquely determined by the number of edges in the observed graph, since each step of the generative
process inserts exactly one edge. 
The methods developed here assume a Bayesian setup, with a prior distribution $\mathcal{L}(\theta)$
on the model parameter $\theta$,
where ${\mathcal{L}(\argdot)}$ generically denotes the law of a random variable.
They sample the posterior distribution 
${\mathcal{L}(\theta|G_T)}$, and are applicable to sequential network models satisfying two properties:
\begin{itemize}
\item[(P1)] The sequence $(G_1,\ldots,G_T)$ of graphs generated by the models forms a Markov chain on 
  the set of finite graphs: ${G_{t+1}\perp\!\!\perp(G_1,\ldots,G_{t-1})\,|\,G_t}$, for each ${t<T}$.
\item[(P2)] The sequence is strictly increasing, in the sense that ${G_t\subsetneq G_{t+1}}$ almost surely.
\end{itemize}
These hold for the random walk models, but also for many other network formation models,
such as preferential attachment graphs, fitness models, and vertex copying models
\citep[e.g.][]{Newman:2009,Goldenberg:etal:2010}.
Despite the attention these models have received in the literature, little work exists on inference
(see \cref{sec:sub:inference:related} for references).

If only a single graph $G_T$ is observed, inference requires 
that the unobserved history is integrated out. The result is a likelihood of the form
\begin{align} \label{eq:likelihood}
  p_{\seqParams}(G_T \mid G_1) = \int p_{\seqParams}(G_T, \Glatent \mid G_1)\; d\Glatent \;.
\end{align}
Since the variables $G_t$ take values in large combinatorial
sets, the integral amounts to a combinatorial sum that is typically intractable. The strategy is to approximate the
integral with a sampler that imputes the unobserved graph sequence, noting that only valid sequences which lead to $G_T$ will have non-zero likelihood \eqref{eq:likelihood}. The sequential nature of the models makes
SMC and particle methods the tools of choice.

\subsection{Basic SMC} \label{sec:sub:basic:SMC}

We briefly recall SMC algorithms: The canonical application is a state space model. Observed is
a sequence ${\seqObs_{1:T} = (\seqObs_1,\dots,\seqObs_T)}$. The model explains the sequence
using an unobserved sequence of latent states $\seqState_{1:T}$, and defines three quantities:
(1) A Markov kernel $q^t_{\seqParams}(\argdot \mid \seqState_{t-1})$ that models transitions between latent states.
(2) An emission distribution $p^t_{\seqParams}$ that explains each observation as 
  ${\seqObs_t \sim p^t_{\seqParams}(\argdot \mid \seqState_{t})}$.
(3) A vector $\seqParams$ collecting the model parameters.
In the simplest case, $\seqParams$ is fixed. The inference target is then the posterior distribution
$\mathcal{L}_{\seqParams}(\seqState_{1:T}\mid \seqObs_{1:T})$ of the latent state sequence.

A SMC algorithm generates some number ${N\in\mathbb{N}}$ of state sequences 
${\seqState_{1:T}^1,\ldots,\seqState_{1:T}^{N}}$,
and then approximates the posterior as a sample average over these sequences, weighted by their respective
likelihoods. The imputed states $\seqState_{1:t}^i$ are called \kword{particles}.
Since the sequence of latent states is Markov, particles can be generated sequentially as
${\seqState_{t}\sim q^t_{\seqParams}(\seqState_{t}|\seqState_{t-1})}$. In cases where sampling from $q^t_{\seqParams}$ is not tractable or desirable,
$q^t_{\seqParams}$ it is additionally approximated by a simpler proposal kernel $r^t_{\seqParams}$. The likelihood, and the accuracy 
of approximation
of $q^t_{\seqParams}$ by $r^t_{\seqParams}$, are taken into account by computing normalized weights
\begin{align} \label{eq:generic:smc:weights}
  w^i_t := \frac{\tilde{w}^i_{t}}{ \sum_{j=1}^{N}\tilde{w}^j_{t} }
  \qquad\text{ where }\qquad
  \tilde{w}^i_{t} 
    = p^t_{\seqParams}(\seqObs_{t} | \seqState^i_{t}) \cdot \frac{  q^t_{\seqParams}(\seqState^i_{t} \mid \seqState^i_{t-1}) }{r^t_{\seqParams}(\seqState^i_{t} \mid \seqState^i_{t-1})} \;.
\end{align}
In step ${t}$, SMC generates particles ${\seqState_{t}^i}$ by first resampling from the previous particles
${(\seqState_{1:t-1}^i)_i}$ with probability proportional to their weights:
\begin{equation*}
  \seqState^i_{1:t} = (\seqState_{1:t-1}^{A^i_{t}}, \seqState_t^i)
  \qquad\text{ where }\quad
  \seqState_{t}^i \sim r^{t}_{\seqParams}(\argdot \mid \seqState_{t-1}^{A^i_{t}})
  \text{ and }
  A^i_{t} \sim \MN(N,(w^i_{t-1})_{i=1}^{N})\;.
\end{equation*}
Resampling a final time after the $T$-th step yields a complete array ${(\seqState_{1:T}^i)_{i}}$, with which the posterior is approximated as the average ${\mathcal{L}_{\seqParams}( dz_{1:T} \mid \seqObs_{1:T}) \approx 
    \frac{1}{N} \sum_{i=1}^{N} \delta_{\seqState^i_{1:T}}(dz_{1:T})}$.
Alternatively, the final resampling step can be omitted, in which case the posterior is approximated with a weighted average. Either approximation is asymptotically unbiased as $N\to\infty$. 
See \citet{Doucet:Johansen:2011} for a thorough review.

\subsection{SMC algorithms for graph bridges} \label{sec:sub:SMC}

\begin{figure}
  \makebox[\textwidth][c]{
    \resizebox{\textwidth}{!}{
  \begin{tikzpicture}
    \node at (0,-.3) {\includegraphics[width=3cm,angle=90]{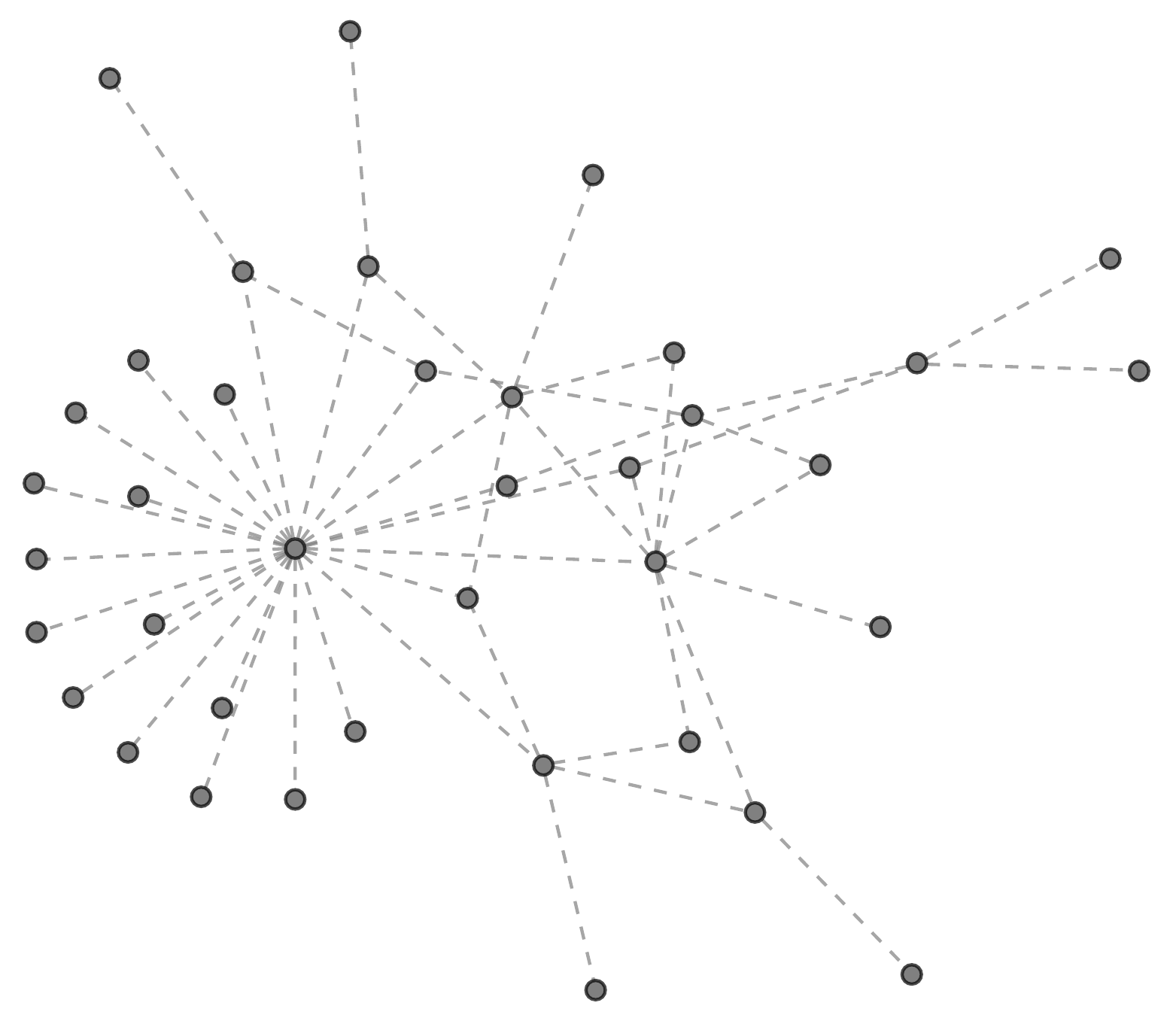}};
    \node at (12,-.3) {\includegraphics[width=3cm,angle=90]{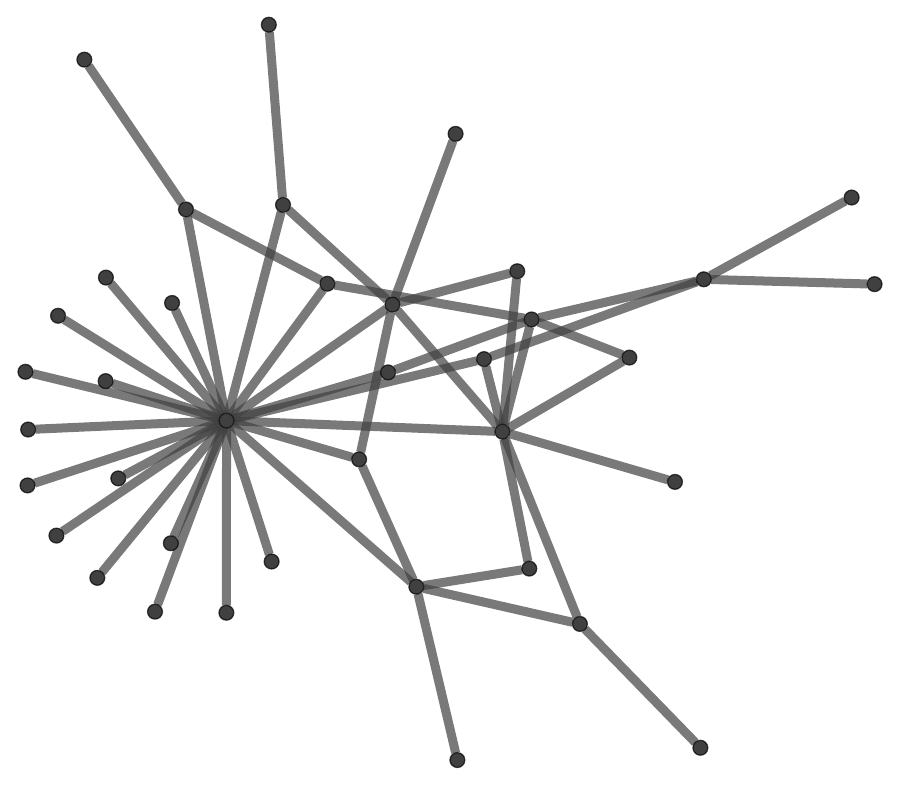}};

    \node at (3,1.4) {\includegraphics[width=3cm,angle=90]{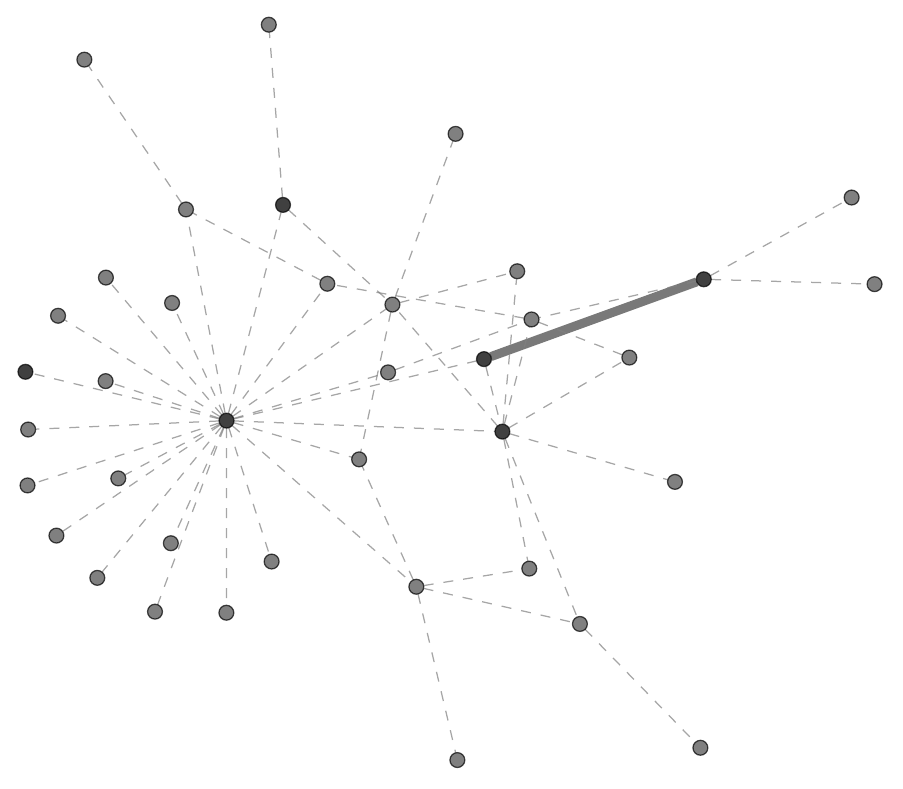}};
    \node at (6,1.4) {\includegraphics[width=3cm,angle=90]{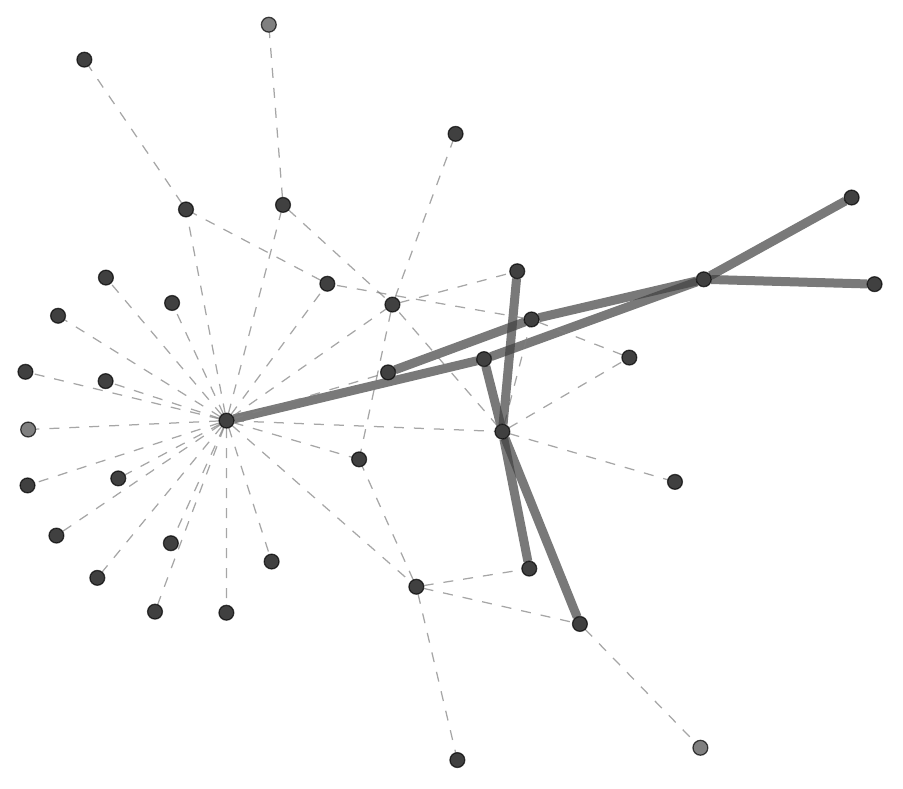}};
    \node at (9,1.4) {\includegraphics[width=3cm,angle=90]{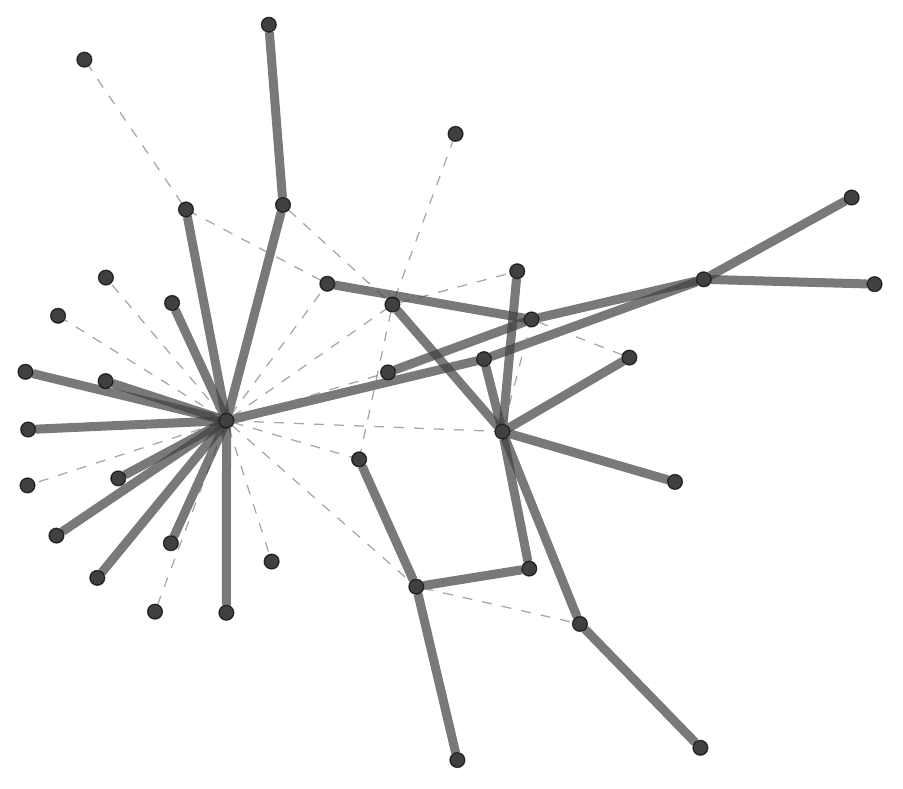}};

    \node at (3,-2) {\includegraphics[width=3cm,angle=90]{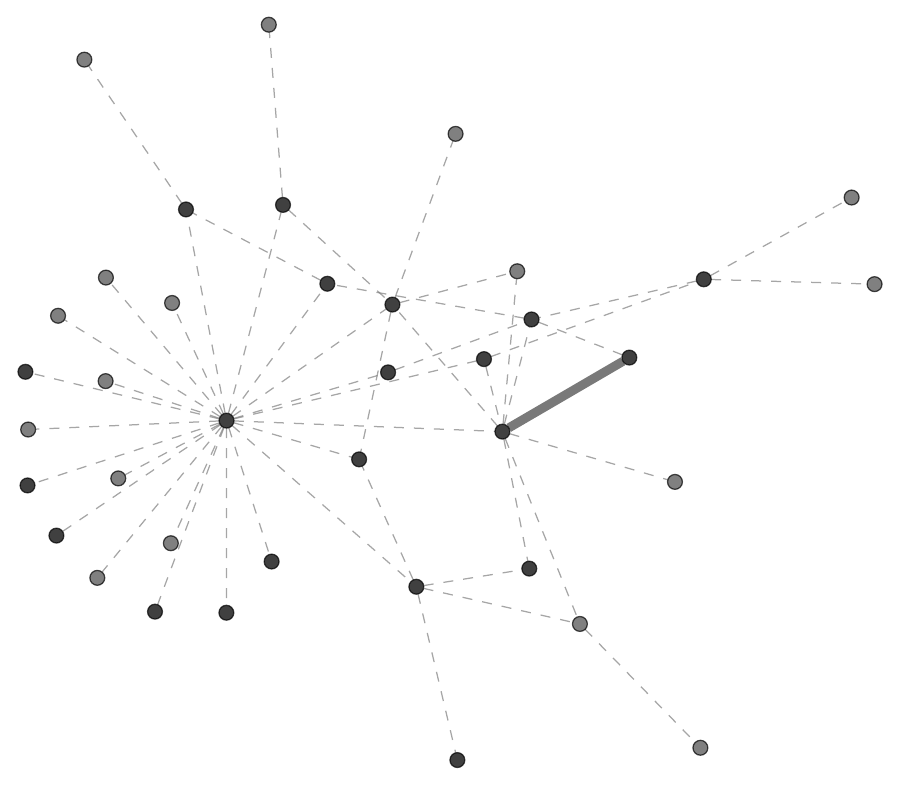}};
    \node at (6,-2) {\includegraphics[width=3cm,angle=90]{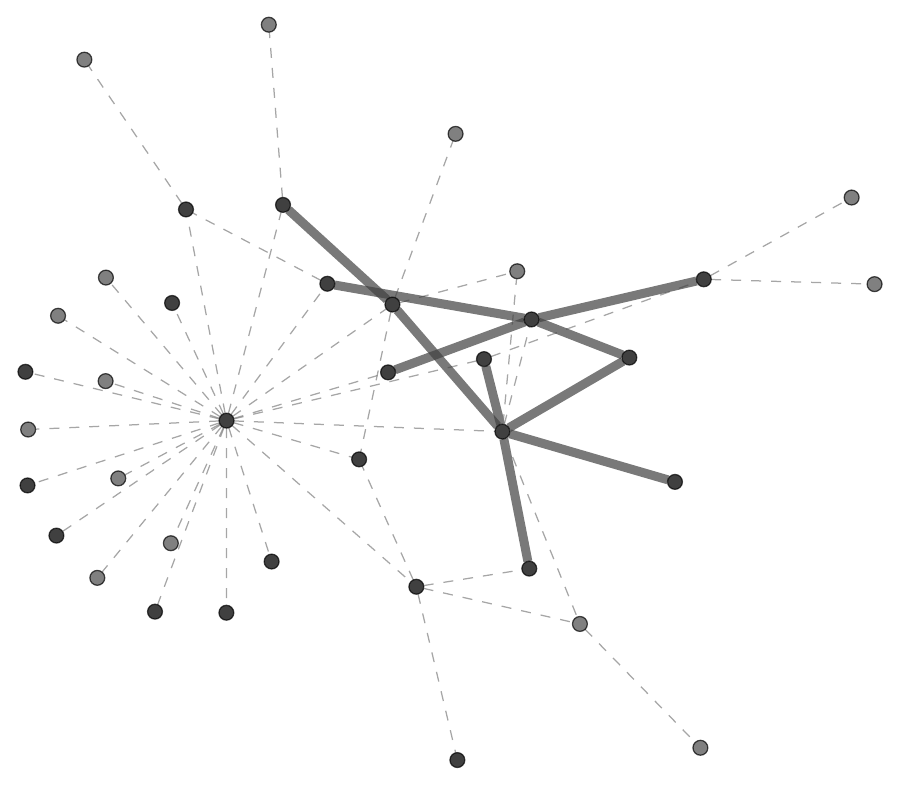}};
    \node at (9,-2) {\includegraphics[width=3cm,angle=90]{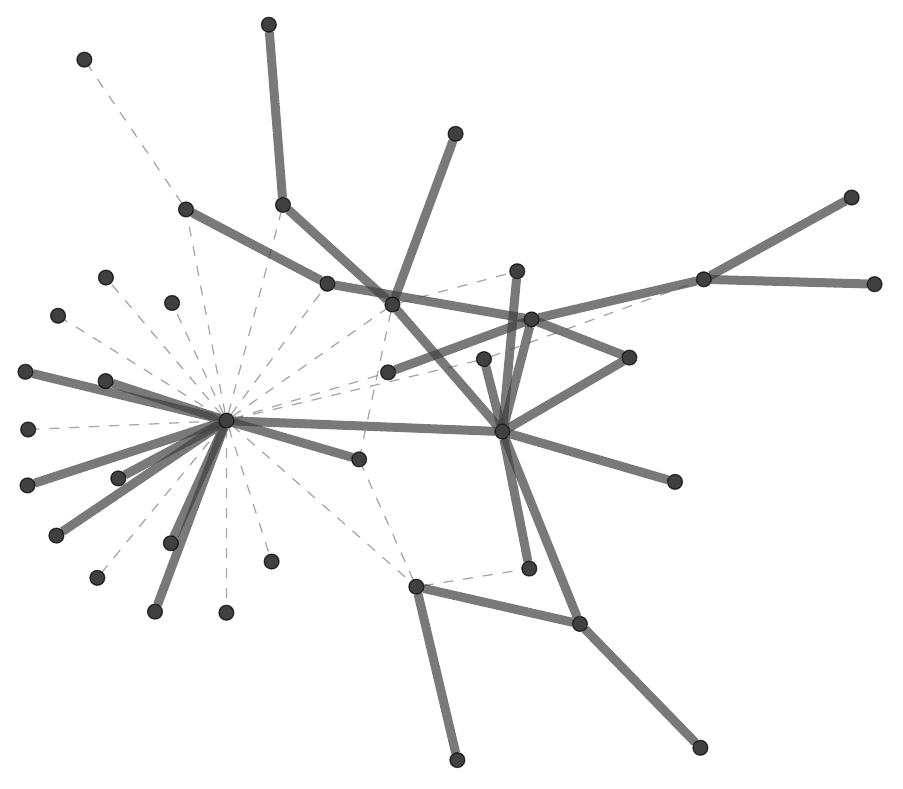}};

    \node at (0,-4) {\footnotesize Vertex set of $G_{50}$};
    \node at (3,-4) {\footnotesize $G_{1}$};
    \node at (6,-4) {\footnotesize $G_{10}$};    
    \node at (9,-4) {\footnotesize $G_{30}$};
    \node at (12,-4) {\footnotesize $G_{50}$};
  \end{tikzpicture}
  }}
  \caption{Two graph bridges generated by \cref{alg:SMC}: A graph $G_{50}$ is drawn from a $\RWU(\alpha,P)$ model, 
    and two graph bridges $G_{1:50}$ are sampled conditionally on the fixed input ${G_{50}}$. Shown
    are the graphs $G_{1}$, $G_{10}$ and $G_{30}$ of each bridge.}
  \label{fig:SMC:bridge}
\end{figure}
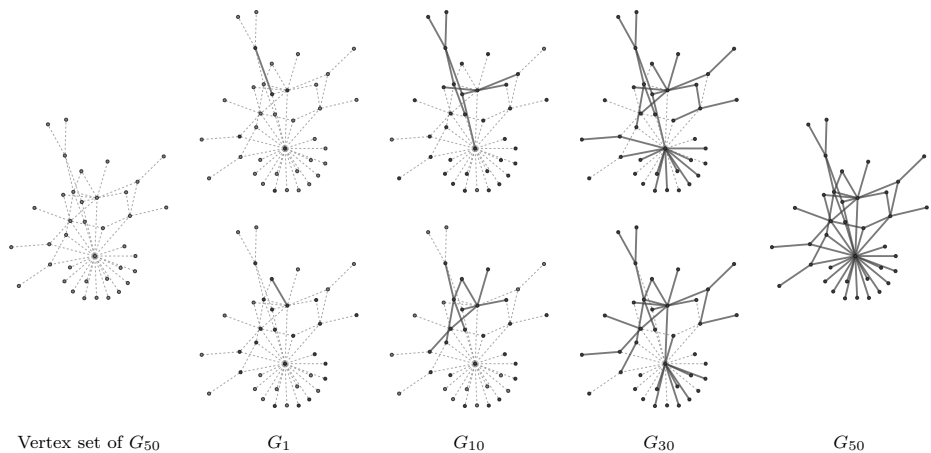

Consider a sequential network model satisfying properties (P1) and (P2) above, with model parameters $\seqParams$. For now, $\seqParams$ is fixed, and the objective is to reconstruct the
graph sequence $G_{1:T}$ from its observed final graph $G_T$ and a fixed initial graph $G_1$, \ie the relevant posterior distribution is ${\mathcal{L}_{\seqParams}(G_{1:T} \mid G_T,G_1)}$.
By the Markov property (P1), each step in the process is governed by a Markov kernel
\begin{equation*}
  q_{\seqParams}^t(g \mid g'):=P(G_t=g \mid G_{t-1}=g')\;.
\end{equation*}
The task of a sampler is hence to impute
the conditional sequence ${G_{2:(T-1)} \mid G_1,G_T}$, which is a stochastic process conditioned
on its initial and terminal point, also known as a \kword{bridge}. 
Compared to the SMC sampling algorithm for state space models, the unobserved sequence ${G_2,\ldots,G_{T-1}}$
takes the role of the hidden states ${z_{1:T}}$, whereas $G_1$ and the single observed graph $G_T$ replace the observed
sequence $x_{1:T}$. The emission likelihood $p^t_{\seqParams}(x_t|z_t)$ is replaced by a bridge likelihood,
of $G_T$ given a graph $G_t$.

The relevant likelihood functions, however, are themselves intractable:
If $g_s$ is a fixed graph at index $s$ in the sequence, the probability of observing $g_t$
at a later point $t>s$ is 
\begin{equation*}
  p_{\seqParams}^{t,s}(g_t|g_s):=\int \Bigl(\prod_{i=s}^{t-1} q^{i+1}_{\seqParams}(g_{i+1}|g_{i})\Bigr)dg_{s+1}\cdots dg_{t-1}
  \qquad\text{ whenever }t>s\;.
\end{equation*}
For a candidate graph $G_t$ generated by the sampler, the likelihood under observation $G_T$
is the \kword{bridge likelihood} ${L_{\seqParams}^t(G_t):=p_{\seqParams}^{T,t}(G_T|G_t)}$.
This likelihood is intractable unless ${t=T-1}$; indeed, for ${t=1}$, it
is precisely the integral \eqref{eq:likelihood} we set out to approximate.

A sequence ${G_{1:T}}$ satisfying (P2) is equivalent to an enumeration of the $T$ edges of $G_T$
in order of occurrence. Let $\pi$ be a permutation of ${\braces{1,\ldots,T}}$, \ie an ordered 
list containing each element of the set exactly once, and
\begin{equation*}
  \begin{split}
  \pi_{t}(G_T):= &\text{ graph obtained by deleting all edges of $G_T$ not in $\pi_{t}$,
  and isolated vertices.}
  \end{split}
\end{equation*}
A random sequence $G_{1:T}$ can then be specified as a pair $(\Pi,G_T)$, for a \emph{random} permutation
$\Pi$ of the edges in $G_T$, with ${G_t=\Pi_{t}(G_T)}$. Given $G_T$, not every permutation $\pi$ is a
valid candidate value for the unobserved order:
$\pi$ must describe a sequence of steps of non-zero probability from 
the initial graph to $G_T$, and we hence have to require
\begin{equation}
  L_{\seqParams}^t(\pi_{t}(G_T))>0\qquad\text{ for all }t\leq T\;.
\end{equation}
If so, we call $\pi$ \kword{feasible} for $G_T$.
For $G_T$ given and ${G_t=\Pi_{t}(G_T)}$, we use $G_t$ and $\Pi_t$ interchangeably. 
The target distribution of the SMC bridge sampler at step $t$ is then
\begin{equation}
  \label{eq:bridge:SMC:target}
  \gamma_{\seqParams}(t):=L_{\seqParams}^t(\Pi_t)P(\Pi_{t})=L_{\seqParams}^t(\Pi_t)\prod_{s=1}^{t-1}q_{\seqParams}^{s+1}(\Pi_{s+1} \mid \Pi_s)\;,
\end{equation}
which satisfies the recursion
\begin{equation*}
  \gamma_{\seqParams}(t)=
  \begin{cases}
    \frac{L_{\seqParams}^t(\Pi_t)}{L_{\seqParams}^{t-1}(\Pi_{t-1})} q_{\seqParams}^{t}(\Pi_{t}|\Pi_{t-1})\gamma_{\seqParams}(t-1)
    & \text{if }\Pi_t \text{ is feasible}\\
    0 & \text{otherwise}
  \end{cases} \;.
\end{equation*}
The intractable part is the bridge likelihood ratio ${L_{\seqParams}^t/L_{\seqParams}^{t-1}}$.
Define
\begin{equation*}
  h_t(\Pi_t):=
    \begin{cases}
      \mathds{1}\braces{L_{\seqParams}^t(\Pi_t)>0} & \text{ for all }t < T-1 \\
      L_{\seqParams}^{t}(\Pi_{t}) & \text{ for } t = T-1
    \end{cases} \;.
\end{equation*}
As \cref{prop:smc} below shows, 
the apparently crude approximation ${L_{\seqParams}^t\approx{h_t}}$
still produces asymptotically unbiased samples from the posterior $\mathcal{L}_{\seqParams}(\Gall \mid G_T, G_1)$;
this is based on methodology developed in
\citep{delMoral:Murray:2015} for bridges in continuous state spaces.
If $\Pi_t$ is feasible, substituting $h_t$ into \eqref{eq:bridge:SMC:target} yields the surrogate recursion
\begin{equation*}
  \gamma_{\seqParams}(t)=
    \frac{ q_{\seqParams}^{t}(\Pi_{t} \mid \Pi_{t-1}) }{ r_{\seqParams}^t(\Pi_{t} \mid \Pi_{t-1}) } 
    \gamma_{\seqParams}(t-1) \qquad \text{ for } \quad t < T-1 \;.
\end{equation*}
The proposal kernel $r_{\seqParams}^t$ is hence chosen as the truncation of $q_{\seqParams}^t$ to feasible permutations,
\begin{equation} \label{eq:smc:bridge:proposal}
  r_{\seqParams}^t(\Pi_t \mid \Pi_{t-1}):=\frac{\mathds{1}\braces{L_{\seqParams}^t(\Pi_t)>0}q_{\seqParams}^{t}(\Pi_{t} \mid \Pi_{t-1})}{\tau_{\seqParams}^t (\Pi_{t-1})}
  \quad\text{ with }\quad\tau_{\seqParams}^t:= \sum_{ \pi_t : \\ L_{\seqParams}^t(\pi_t)>0 } q_{\seqParams}^{t}(\pi_{t} \mid \Pi_{t-1})\;.
\end{equation}
For particles ${G_t^i=\Pi^i_t(G_T)}$, the unnormalized 
SMC weights \eqref{eq:generic:smc:weights} are 
\begin{align} \label{smc:bridge:weights}
  \tilde{w}^i_t = 
  \begin{cases}
    \tau^t_{\seqParams}(\Pi_{t-1}^i) & \text{ if } t < T-1 \\
    q_{\seqParams}^t(\Pi \mid \Pi^i_{T-1})\; \tau^{T-1}_{\seqParams}(\Pi_{T-2}^i) & \text{ if } t=T-1
  \end{cases} \;.
\end{align}
The sampling algorithm then generates a graph bridge as follows:
\begin{poalgorithm}[Bridge sampling]
\label{alg:SMC}
\begin{itemize}
\item Initialize ${G_1^i:=G_1}$, $\Pi_1^i \sim \Uniform\{1,\dots,T\}$, and ${w_1^i:=1/N}$ for each ${i\leq N}$.
\item For ${t=2,\ldots,T-1}$, iterate:
  \begin{itemize}
  \item Resample indices ${A^i_{t} \sim\MN(N,(w_{t-1}^i)_i)}$.
  \item Draw ${\Pi_t^i\sim r_{\seqParams}^t(\argdot \mid \Pi_{t-1}^{A^i_{t}})}$ as in \eqref{eq:smc:bridge:proposal} for each $i$.
  \item Compute weights as in \eqref{smc:bridge:weights} and normalize to obtain ${w_t^i}$.
  \end{itemize}
\item Resample $N$ complete sequences ${G_{1:T}^i = \Pi^i\sim\MN(N,(w_{T-1}^i)_i)}$.
\end{itemize}
\end{poalgorithm}
See \cref{fig:SMC:bridge} for an illustration.
Computing $h_t(\Pi_t)$ and $\tau^t_{\seqParams}(\Pi_{t-1})$ is simplified by the constraints on $\Pi$:
Given $\Pi_{t-1}$, the requirement that $\Pi_t$ must again be a restriction of a permutation $\Pi$
implies ${\Pi_t = (\Pi_{t-1}, s_t)}$, for some ${s_t \in \{1,\dots,T\}\setminus \Pi_{t-1}}$. 
Within this set, the $s_t$ for which $\Pi_t$ is feasible are simply those
edges in $G_T$ connected to $\Pi_{t-1}(G_T)$.

\begin{proposition} \label{prop:smc}
  Let ${q_{\seqParams}^t}$, for ${t=1,\ldots,T}$, be the Markov kernels defining a sequential network model that satisfies
  conditions (P1) and (P2). Given an observation $G_T$, \cref{alg:SMC} produces samples that are asymptotically unbiased as $N\to\infty$: For any bounded function $f$ on graph sequences,
  ${\frac{1}{N} \sum_{i=1}^{N} f(\Gall^i) \rightarrow \E[f(\Gall) \mid G_T, G_1]}$ in probability
  as ${N \to \infty}$, where the expectation is evaluated with respect to the model
  posterior. 
\end{proposition}

The particle MCMC methods of the next section will require an unbiased estimate of the bridge likelihood, ${L_{\seqParams}^1(G_1)}$, which we state in the following: 
\begin{proposition} \label{prop:unbiased:bridge:likelihood}
  Let ${q_{\seqParams}^t}$, for ${t=1,\ldots,T}$, be the Markov kernels defining a sequential network model that satisfies conditions (P1) and (P2), and let ${r_{\seqParams}^t}$ be the corresponding proposal kernels. Given an observation $G_T$ and a fixed $G_1$, the bridge likelihood estimator
  \begin{align}
    \label{eq:SMC:marginal}
      \hat{L}_{\seqParams}^1 
        & := \prod_{t=2}^{T-1} \left[ \left( \sum_{i=1}^N \frac{\tilde{w}^i_t}{N} \right)
          \left(\frac{ \sum_{i=1}^N  h_{t-1}(G_T \mid G^i_{t-1}) \tilde{w}^i_{t-1} }{ \sum_{i=1}^N  \tilde{w}^i_{t-1} }  \right)   \right] \;.
  \end{align}
  is positive and unbiased: $\hat{L}^1_{\seqParams}>0$ and ${\E[\hat{L}^1_{\seqParams}] = L_{\seqParams}^1(G_1) = p_{\seqParams}(G_T \mid G_1)}$, for any $N\geq 1$.
\end{proposition}

{\noindent\textit{Variance reduction.}} 
If estimates exhibit high variance, it is straightforward to modify \cref{alg:SMC} to 
use adaptive resampling \citep[see][]{delMoral:Murray:2015}, and to replace multinomial resampling
above by residual or stratified resampling \citep{Doucet:Johansen:2011}. A variance reduction technique due to \citet{Fearnhead:Clifford:2003} takes advantage of the discrete state space by enumerating possible states $\Pi_{t+1}^i$ given the current state $\Pi_{t}^i$, creating $N'>N$ ``virtual'' particles before down-sampling to propagate $N$ particles. 
Finally, if a given model admits a more bespoke approximation $h_t$ to the bridge likelihood,
this approximation can be substituted for $h_t$, following \citet{delMoral:Murray:2015}. In
this case, some of the equations above require minor modifications.

\subsection{Parameter inference} \label{sec:particle:MCMC}

\def\SeqParams{\Theta}

\cref{alg:SMC} generates a history of a graph under a model with fixed parameter vector $\seqParams$. 
For parameter inference, the parameters are treated as a random variable $\Theta$, with 
prior distribution $\mathcal{L}(\SeqParams)$,
and the task is to generate samples from the joint posterior ${\mathcal{L}(\SeqParams,G_{1:T}|G_1,G_T)}$.
The sample space of the SMC sampler above is thus extended by the domain of $\SeqParams$. The bridge likelihood ${L_{\SeqParams}^1(G_1)}$ is a marginal likelihood that naturally
leads to particle MCMC methods \citep{Andrieu:Doucet:Holenstein:2010}, which use SMC to compute
an unbiased estimate ${\hat{L}^1_{\SeqParams}}$. 
Substituting into the Metropolis--Hastings acceptance ratio of a proposal $\tilde{\SeqParams}$ from a (yet to be specified)
proposal distribution $\tilde{q}$ yields
\begin{align} \label{eq:acceptance:ratio}
  \mathbb{P}\lbrace\text{ accept $\tilde{\SeqParams}$ }\rbrace
  =
  \frac{ 
    \hat{L}_{\tilde{\SeqParams}}^1 \cdot \pi(\tilde{\SeqParams}) 
  }{ 
    \hat{L}_{\SeqParams}^1 \cdot \pi(\SeqParams) 
  } 
  \cdot
  \frac{ 
    \tilde{q}(\SeqParams \mid \tilde{\SeqParams}) 
  }{ 
    \tilde{q}(\tilde{\SeqParams} \mid \SeqParams) 
  } \;.
\end{align}
By \cref{prop:unbiased:bridge:likelihood}, a positive unbiased estimate of the bridge likelihood
is given by \eqref{eq:SMC:marginal}; with \eqref{eq:acceptance:ratio}, we obtain a particle marginal Metropolis--Hastings (PMMH) sampler:
\begin{poalgorithm}[ ]
  \label{alg:PMMH}
  \begin{itemize}
  \item Initialize ${\SeqParams^0\sim\pi}$.
  \item For ${j=1,\ldots,J}$ iterate:
    \begin{enumerate}
    \item Draw candidate a value ${\tilde{\SeqParams}\sim\tilde{q}(\argdot \mid \SeqParams^{j-1})}$.
    \item Run \cref{alg:SMC} with parameter $\tilde{\SeqParams}$ to compute
      ${\hat{L}_{\SeqParams}^1}$ as in \eqref{eq:SMC:marginal}.
    \item Accept $\tilde{\SeqParams}$ with probability \eqref{eq:acceptance:ratio} and set ${\SeqParams^j:=\tilde{\SeqParams}}$; else set ${\SeqParams^j:=\SeqParams^{j-1}}$.
    \item If $\tilde{\SeqParams}$ is accepted, select a single graph sequence $G_{1:T}^j$ by resampling from the particles output by \cref{alg:SMC}; else set $G_{1:T}^j = G_{1:T}^{j-1}$.
    \end{enumerate}
  \item Output the sequence ${(\SeqParams^1,G_{1:T}^1),\ldots,(\SeqParams^J,G_{1:T}^J)}$.
  \end{itemize}
\end{poalgorithm}
The algorithm asymptotically samples the joint posterior
${\mathcal{L}(\SeqParams,G_{1:T} \mid G_1,G_T)}$, or the
marginal posterior ${\mathcal{L}(\SeqParams \mid G_1,G_T)}$ if the output of step (d) is omitted:
\begin{proposition} \label{prop:pmmh}
  If the proposal density $\tilde{q}(\argdot|\argdot)$ is chosen such that the Metropolis--Hastings sampler
  defined by \eqref{eq:acceptance:ratio} is irreducible and aperiodic, \cref{alg:PMMH} is a PMMH sampler.
  The marginal distributions ${\mathcal{L}(\SeqParams^j,\Gall^j)}$ of its output sequence satisfy
  \begin{align}
    \| \mathcal{L}(\SeqParams^j,\Gall^j)- \pi(\argdot \mid G_T, G_1)  \|_{\ind{TV}} 
    \xrightarrow{j\to\infty} 0\;.
  \end{align}
  This is true regardless of the sample size $N$ generated by \cref{alg:SMC} in step (b).
\end{proposition} 

\subsection{Particle Gibbs for random walk models} \label{sec:particle:gibbs}

The computational cost of \cref{alg:PMMH} stems mostly from rejections
in (c), each of which requires an additional execution of  (b). This can
be addressed by turning the algorithm
into a Gibbs sampler, which eliminates both rejections and the need for a 
proposal kernel $\tilde{q}$.
The result is a particle Gibbs (PG) sampler 
\citep{Andrieu:Doucet:Holenstein:2010}. 
Such an algorithm is described below, now specifically for the random walk model.

For a ${\RW(\alpha,\text{Poisson}_{+}(\lambda))}$ model, 
the parameter takes the form ${\SeqParams=(\alpha,\lambda)}$, and we fix a beta
distribution $\prior{\alpha}$ as prior for $\alpha$, and a gamma prior $\prior{\lambda}$ for $\lambda$. 
To sample a sequence ${G_1,\ldots,G_T}$ from the model, one can generate two
i.i.d.\ sequences, ${\mathbf{B}=(B_1,\ldots,B_T)}$ of $\Bernoulli(\alpha)$ variables,
and ${\mathbf{K}=(K_1,\ldots,K_T)}$ of $\Poisson(\lambda)$ variables. 
In step $t$, a new edge is inserted if ${B_t=1}$; otherwise,
two vertices are connected by a random walk of length $K_t$. 
As shown in \cref{sec:inference_appx}, the entries $B_t$, $K_t$ and the model parameters
can be integrated out of the kernel.
$\mathbf{B}$ and $\mathbf{K}$ can hence be sampled separately from the SMC steps, rather than inside
\cref{alg:SMC}, which improves exploration.

In its $j$th iteration, the algorithm updates $\mathbf{B}$ and $\mathbf{K}$ by looping over the indices ${t\leq T}$ of $\Gall$.
Since Gibbs samplers condition on every update immediately, vectors maintained
by the sampler are of the form
\begin{equation*}
  \mathbf{{B}}^{j}_{-t}:=(B^{j+1}_1,\ldots,B^{j+1}_{t-1},B^j_t,\ldots,B^j_T)
  \quad\text{ and }\quad
  \mathbf{{K}}^{j}_{-t}:=(K^{j+1}_1,\ldots,K^{j+1}_{t-1},K^j_t,\ldots,K^j_T)
\end{equation*}
The posterior predictive distributions of ${{B}_t^{j+1}}$ and
$K_t^{j+1}$ given these vectors can be obtained in closed form (see \cref{sec:inference_appx}).
We abuse notation
and let the index ${j=0}$ refer to the prior predictive distribution, \ie
we write 
${\mathcal{L}(\mathbf{B}^1|\mathbf{B}^{0},G_{1:T}^{0}):=\int\mathcal{L}(\mathbf{B}^1 \mid \alpha)\prior{\alpha}(d\alpha)}$,
and similarly for ${\mathcal{L}(\mathbf{K}^1 \mid \mathbf{K}^{0},G_{1:T}^{0})}$.

\begin{poalgorithm}[Particle Gibbs sampling for $\RW(\alpha,\Poisson(\lambda))$ models.]
  \label{alg:PG}
  \begin{itemize}
  \item For ${j=0,\ldots,J}$:
  \begin{enumerate}
  \item Draw ${(\mathbf{B}^{j+1},\mathbf{K}^{j+1}) \sim\mathcal{L}(\mathbf{B}^{j+1},\mathbf{K}^{j+1} \mid \mathbf{B}^{j},\mathbf{K}^{j},G_{1:T}^{j})}$.
  \item Using \cref{alg:SMC}, with
    ${q_{\alpha,\lambda}^t(\argdot \mid G^{j}_{t-1},\mathbf{B}_{-t}^{j+1},\mathbf{K}_{-t}^{j+1})}$ and $N\geq 2$,
    update ${G^{j+1}_{1:T}}$ and $A^{j+1}_{2:T}$. At each step $t$, substitute
    the previous iteration's bridge $G^j_{1:t}$ for the particle with index $A^j_t$ (see below).
  \item Draw ${\alpha^{j+1} \mid \mathbf{B}^{j+1}}$ and ${\lambda^{j+1} \mid \mathbf{K}^{j+1}}$ from their conjugate posteriors.
  \end{enumerate}
  \item Output the sequence ${(G_{1:T}^1,\alpha^1,\lambda^1),\ldots,(G_{1:T}^J,\alpha^J,\lambda^J)}$.
  \end{itemize}
\end{poalgorithm}
The algorithm constitutes a blocked Gibbs sampler.
The parameter values $\alpha^j$ and $\lambda^j$ generated in (c) are not used in the execution of the sampler; they
only serve as output. 
Step (b) is a \kword{conditional SMC} step \citep[see][]{Andrieu:Doucet:Holenstein:2010} that conditions on $G_{1:T}^j$, and makes it a valid Gibbs sampler. 

\begin{corollary}
  \Cref{alg:PG} is a valid particle Gibbs sampler. For any ${N \geq 2}$, it generates a sequence 
  ${(\Gall^j, \alpha^j, \lambda^j)_{j}}$ whose law converges to the model posterior,
  \begin{align}
    \|  \mathcal{L}(\Gall^j, \alpha^j, \lambda^j) - \mathcal{L}(G_{1:T},\alpha,\lambda \mid G_T, G_1)\|_{\ind{TV}} \xrightarrow{j\to\infty} 0\;.
  \end{align}
\end{corollary}

\subsection{Related work} \label{sec:sub:inference:related}

Existing work on the use of sampling algorithms for inference in sequential network models
mainly follows one of two general approaches: Importance sampling~\citep{Wiuf:etal:2006}, or 
approximate Bayesian computation (ABC) schemes~\citep{Thorne:Stumpf:2012}. ABC relies on heuristic approximations of the likelihood, and is not guaranteed to sample from the correct target distribution.
The basic approach used above---conduct inference by imputing graph sequences generated by a suitable sampler---seems
to be due to \citet{Wiuf:etal:2006}, based on earlier work by \citet*{Griffiths:Tavare:1994aa} on ancestral inference for coalescent models. More recently, similar approaches have appeared in the ``network archeology'' literature \citep{Navlakah:Kingsford:2011}. See, e.g. \citet{Young:2018aa} and references therein. 
The work of \citet{Wang:Jasra:deIorio:2014,Jasra:etal:2015} is related to ours in that they employ particle MCMC, for particular sequential models. \citet{Bloem-Reddy:etal:2018} exploit the simple probabilistic structure in preferential attachment models to construct a Gibbs sampler that infers the vertex arrival order. 
All of these methods are, in short, applicable 
if the sequential model in question is very simple. Otherwise, they suffer (1) from the
high variance of estimates that is a hallmark of importance sampling~\citep{Doucet:Johansen:2011},
or (2) infeasible computational cost.
The algorithm of \citet{Wang:Jasra:deIorio:2014}, for example, samples backwards through the sequence
generated by the model, and for each reverse step ${G_t \to G_{t-1}}$ requires computing the probability
of every possible forward step ${G'_{t-1} \to G_{t}}$ under the given model; even for the (still rather
simple) random walk model, that is no longer practical.

Inference techniques developed for models of longitudinal network data are similar in spirit to those developed here, though they typically require observations at each time step. A MCMC algorithm for continuous-time dynamic network models observed at discrete intervals was developed by \citet{Koskinen:Snijders:2007}; that requires a network of fixed size and are not suitable for growing networks.

\section{Experimental evaluation} \label{sec:experiments}

\begin{figure}[t] 
\makebox[\textwidth][c]{
    \resizebox{\textwidth}{!}{
    \begin{tikzpicture}
      \path[use as bounding box]
      (-3,-2.2) rectangle (11.6,1.9);
      \node at (0,0) {\includegraphics[width=6cm]{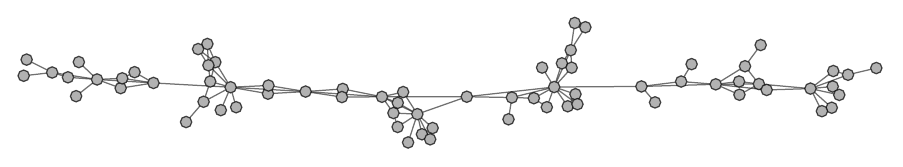}};
      \node at (5.2,0) {\includegraphics[width=4cm]{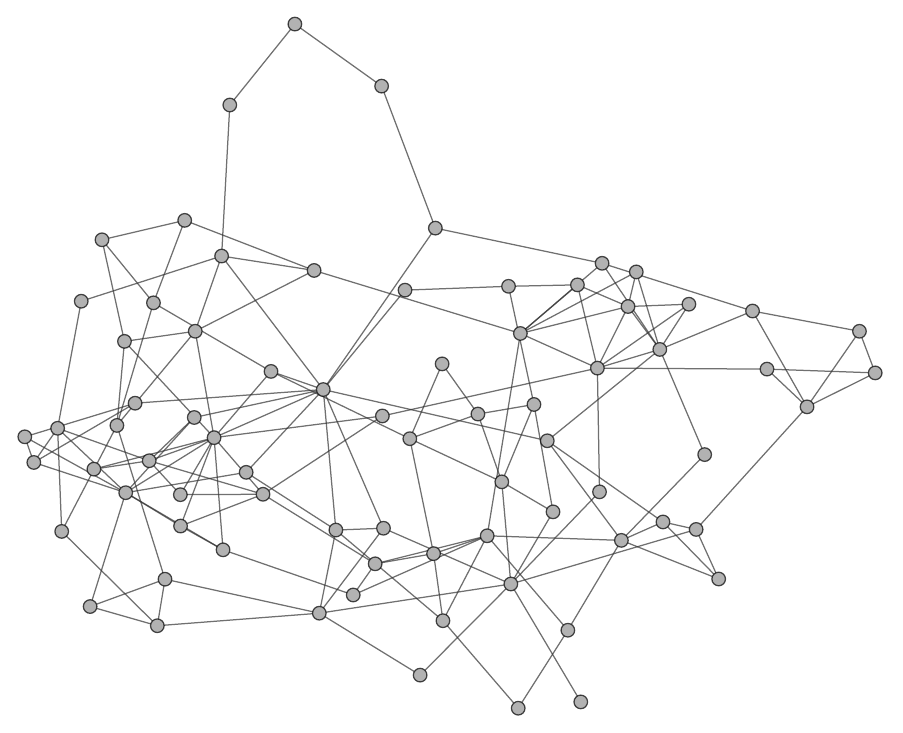}};
      \node at (9.7,0) {\includegraphics[width=4.5cm,angle=-20]{interactomeAdj.pdf}};
      \node at (0,-2) {\footnotesize (i)};
      \node at (5.2,-2) {\footnotesize (ii)};
      \node at (9.7,-2) {\footnotesize (iii)};
    \end{tikzpicture}
}}
	\caption{Network data examples: (i) the largest connected component of the NIPS co-authorship network, 2002-03 \citep{Globerson:etal:2007}; (ii) San Juan Sur family ties \citep{Loomis:etal:1953}; (iii) protein-protein interactome \citep{Butland:etal:2005}.}
	\label{fig:network_data}
        \vspace{-.5cm}
\end{figure}

This section evaluates properties of the model on real-world data and compares its performance
to other network models. 
We also discuss the interpretation of the random walk parameter as a length scale, 
and of the latent order $\Pi$ as a measure of centrality.

We consider three network datasets, shown in \cref{fig:network_data}, that exhibit a range of characteristics.  The first is the largest connected component (LCC) of the NIPS co-authorship network in 2002-03, extracted from the data used in~\cite{Globerson:etal:2007}. As shown in~\cref{fig:network_data}, it has a global chain structure connecting highly localized communities. The second dataset represents ties between families in San Juan Sur (SJS), a community in rural Costa 
Rica~\citep{Loomis:etal:1953,Pajek:data}, and is chosen here as an example of a network with small diameter. 
The third is the protein-protein interactome (PPI) of~\cite{Butland:etal:2005}. The PPI exhibits features such as chains, pendants, and heterogeneously distributed hubs. 
Basic summary statistics are as follows:
\begin{center}
  \makebox[\textwidth][c]{
    \resizebox{\textwidth}{!}{
  \begin{tabular}{cccccc}
    \emph{Dataset} & \emph{Vertices} & \emph{Edges} & \emph{Clustering coeff.} & \emph{Diameter} & \emph{Mean $L_2$-mixing time of r.w.} \\
    \midrule
    NIPS & 70 & 114 & 0.70 & 14 & $57.3$ \\
    SJS & 75 & 155 & 0.32 & 7 & $8.4$ \\
    PPI & 230 & 695 & 0.32 & 11 & $9.1$ \\
    \\
  \end{tabular}
}}
\end{center}
None of these data sets is particularly large: Sampler diagnostics show graphs of this size suffice to reliably recover model parameters under the random walk model.

Using \cref{alg:PG}, the posterior distributions of $\alpha$ and $\lambda$ for both the 
uniform and the size-biased ${\RW(\alpha,\text{Poisson}_{+}(\lambda))}$ model can be sampled.
Kernel-smoothed posterior distributions under all three datasets are shown in~\cref{fig:data_posterior}. Although posteriors in the uniform and size-biased case are very similar, these models generate graphs with different characteristics for identical parameter values (see \cref{sec:model:fitness}).

\subsection{Sampler diagnostics (synthetic data)}
\label{sec:experiments:diagnostics}

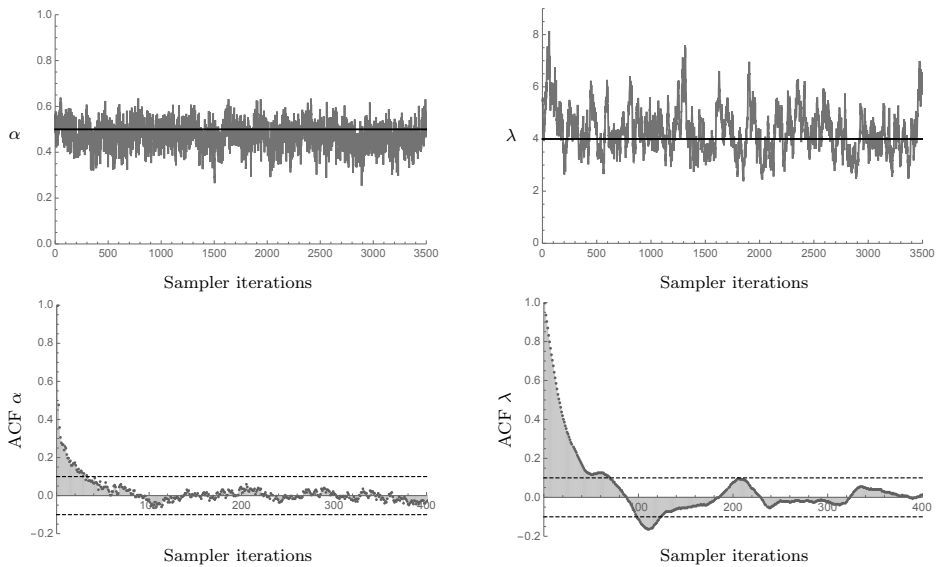
\begin{figure}[bt] 
  \makebox[\textwidth][c]{
    \resizebox{\textwidth}{!}{
    \begin{tikzpicture}
      \begin{scope}
      \node at (0,0) {\includegraphics[width=6.5cm]{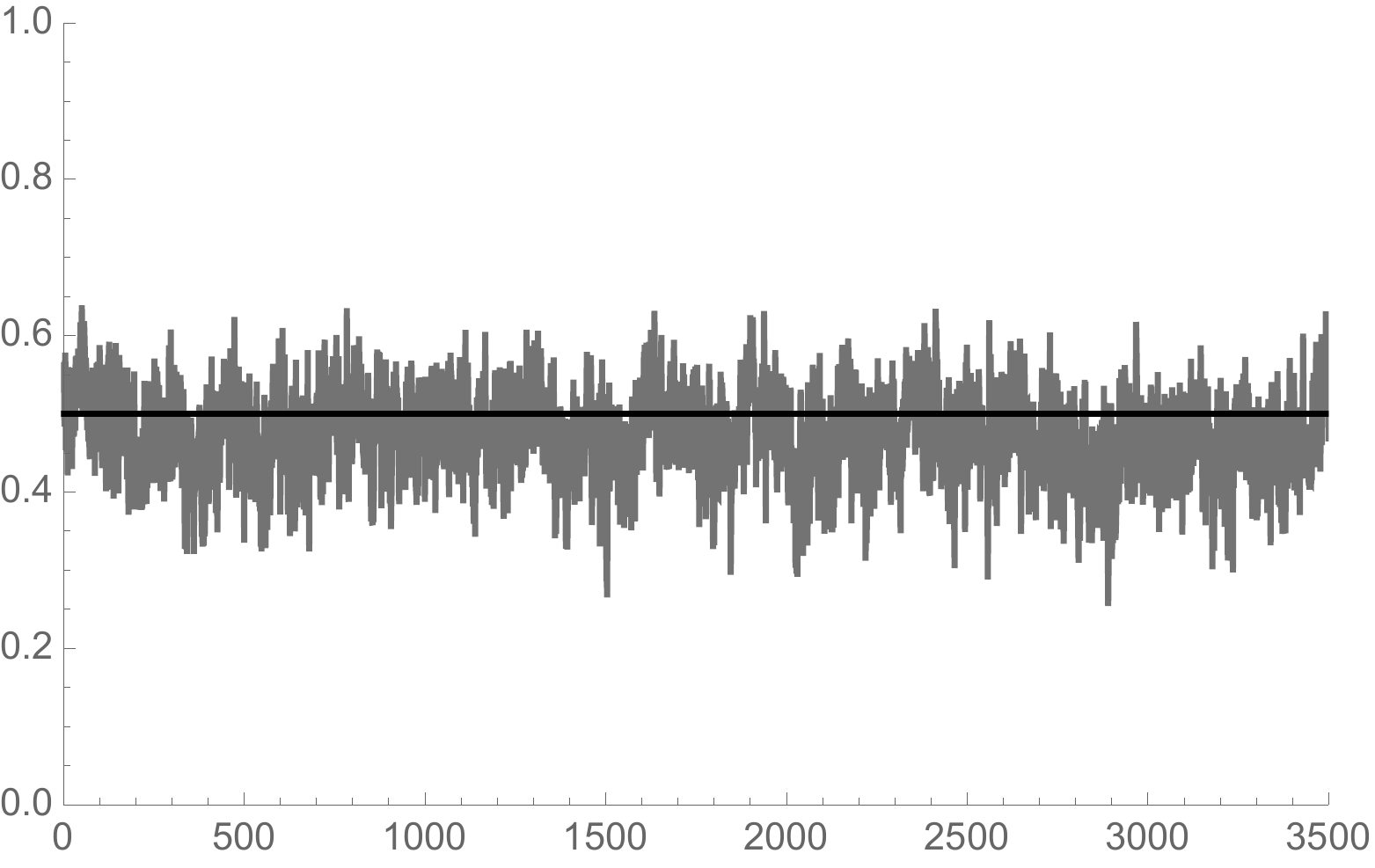}};
      \node at (-3.6,0) {\footnotesize $\alpha$};
      \node at (0,-2.4) {\footnotesize Sampler iterations};

      \node at (-.07,-4.6) {\includegraphics[width=6.55cm]{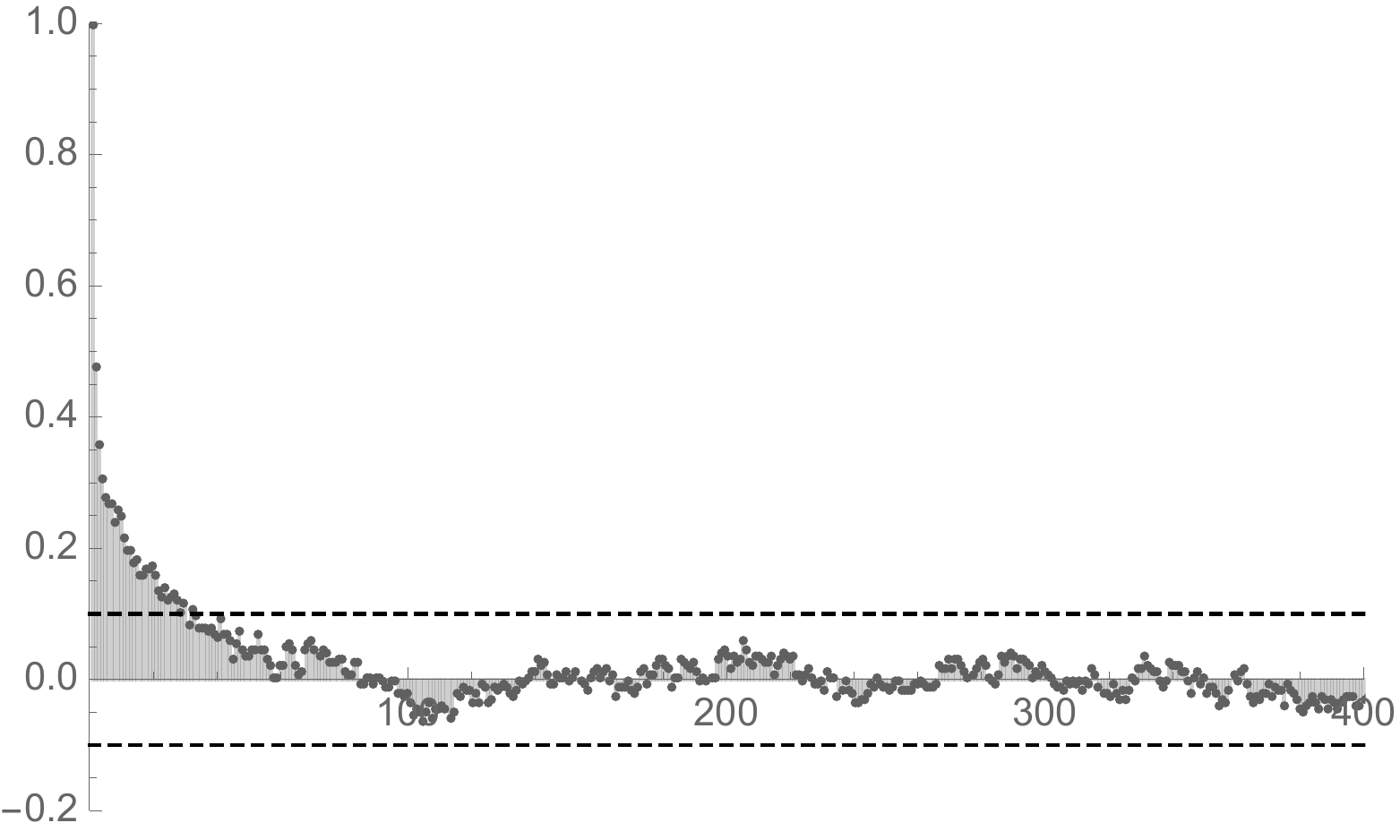}};
      \node[rotate=90] at (-3.6,-4.6) {\footnotesize ACF $\alpha$};
      \node at (0,-6.8) {\footnotesize Sampler iterations};
      \end{scope}
      \begin{scope}[xshift=8cm]
      \node at (0,0) {\includegraphics[width=6.5cm]{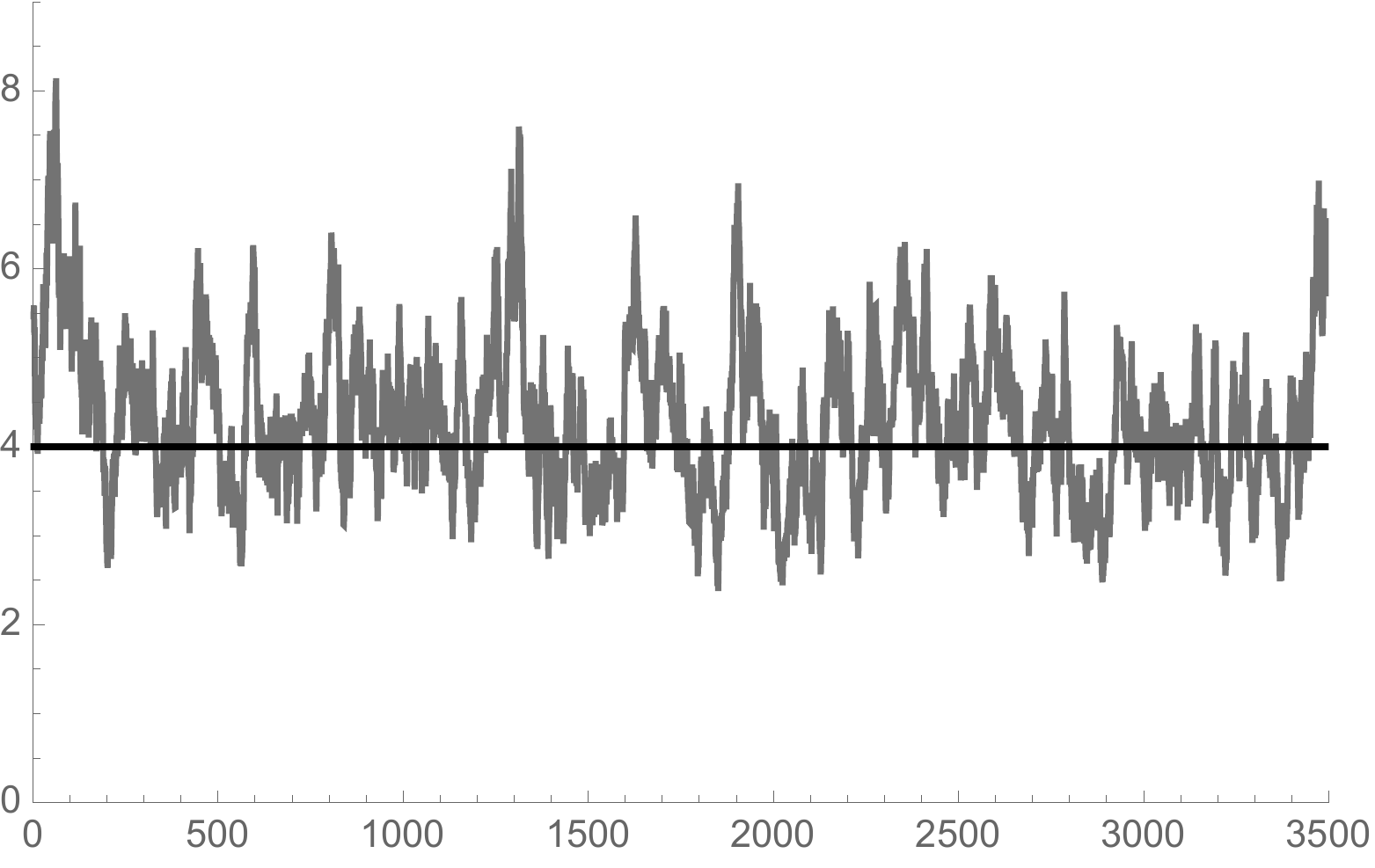}};
      \node at (-3.6,0) {\footnotesize $\lambda$};
      \node at (0,-2.4) {\footnotesize Sampler iterations};

      \node at (-.15,-4.6) {\includegraphics[width=6.7cm]{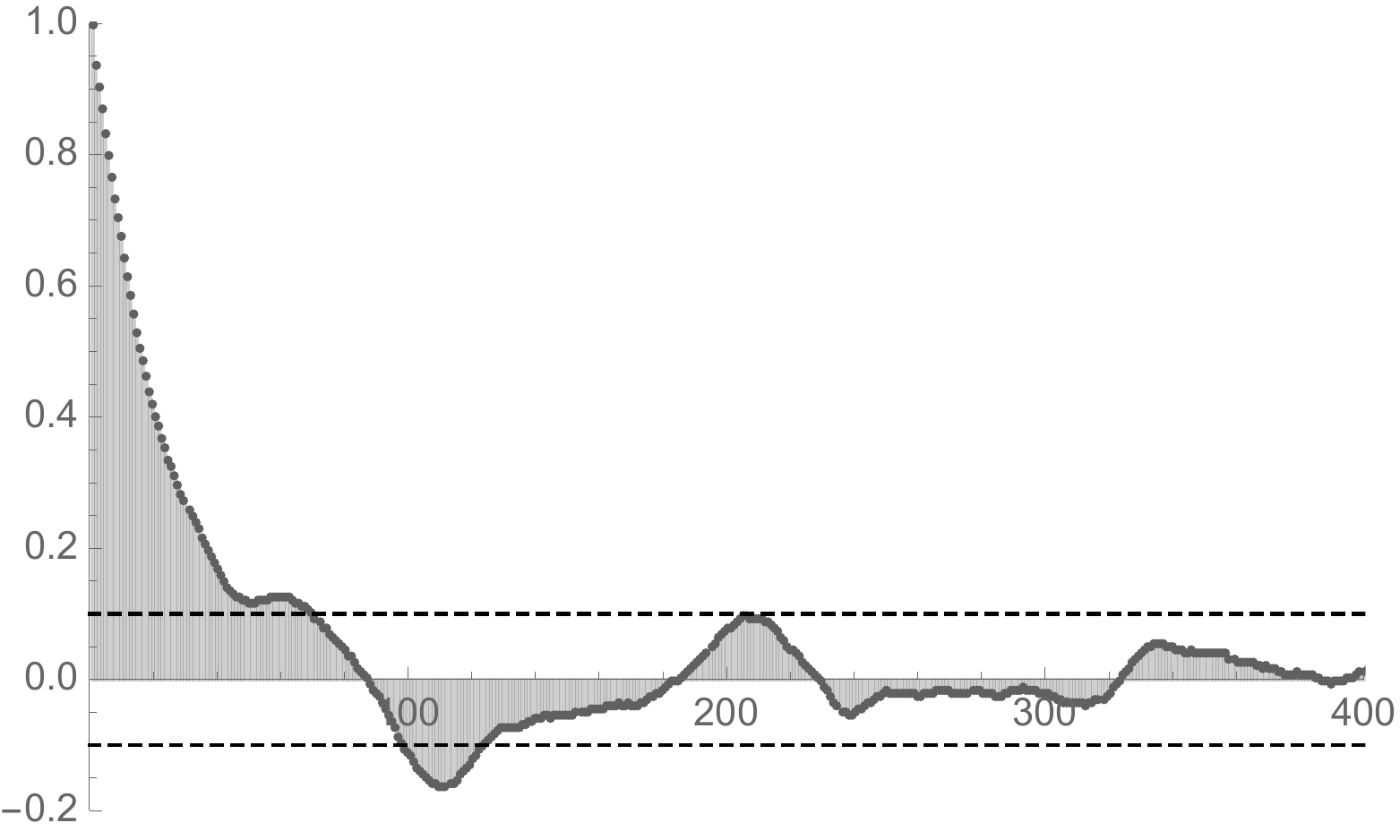}};
      \node[rotate=90] at (-3.7,-4.6) {\footnotesize ACF $\lambda$};
      \node at (0,-6.8) {\footnotesize Sampler iterations};
      \end{scope}
    \end{tikzpicture}
  }}
  \caption{Top: Posterior sample traces for a PG sampler fit to a synthetic graph $G_{T}$ generated with ${\alpha=0.5}$ and ${\lambda=4}$ (solid black lines). Bottom: The corresponding autocorrelation functions.}
  \label{fig:posterior:sampling}
\end{figure}

To assess sampler performance, we test the sampler's ability to recover parameter values for graphs generated by
the model: A single graph $G_T$ with ${T=250}$ edges is generated for fixed values of $\alpha$ and $\lambda$. The joint posterior distribution of $\alpha$ and $\lambda$ given $G_T$ is then estimated using \cref{alg:PG} (with $N=100$ 
particles). 
Regardless of the input parameter value, a uniform prior is chosen for $\alpha$, and a 
$\text{Gamma}(1,1/4)$ prior for $\lambda$. For the sake of brevity, we report results only for 
the $\RWU(\alpha,\Poisson(\lambda))$ model, on simple graphs (which pose a harder challenge for inference
than multigraphs, cf. \cref{sec:mle}).
\cref{fig:posterior:sampling} shows sampler traces and 
autocorrelation functions. The samplers mix and converge quickly, and
as one would expect, only $\lambda$ displays any noticeable autocorrelation. 
Posteriors are shown in \cref{fig:param:sweep}. Clearly, recovery of model
parameters from a single graph is possible.
The effect of constraining to simple graphs---the use of 
step (3') rather than (3) in \cref{scheme:multigraph}---is apparent for ${\lambda=2}$: 
For small values of $\lambda$, a significant proportion of random walks ends at their starting
point, resulting in the insertion of a new vertex. Since new vertices are otherwise inserted
as an effect of $\alpha$, parameters correlate in the posterior, as visible
in \cref{fig:param:sweep} for intermediate values of $\alpha$.

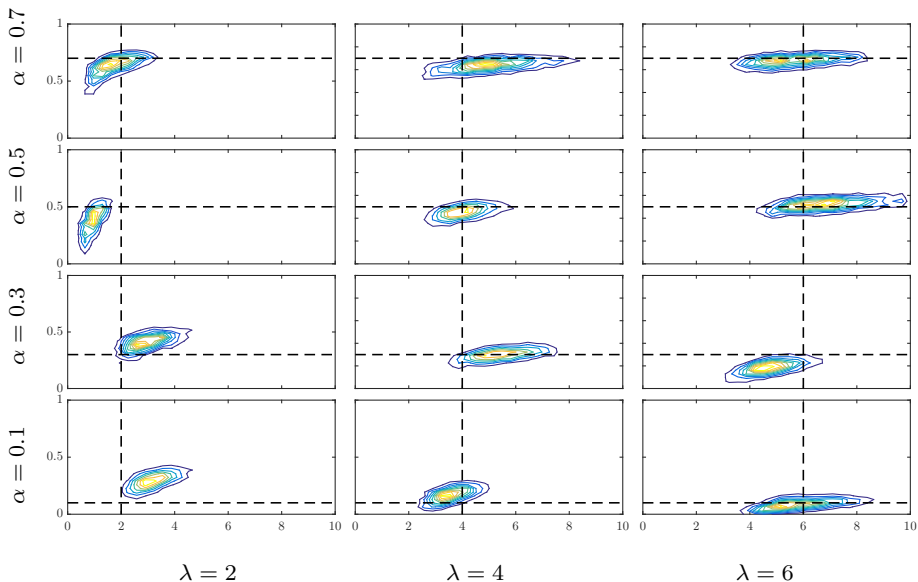
\begin{figure}[tb]
  \centering
  \makebox[\textwidth][c]{
  \resizebox{1.1\textwidth}{!}{
    \begin{tikzpicture}
      \node (a) {
        \includegraphics[width=\textwidth]{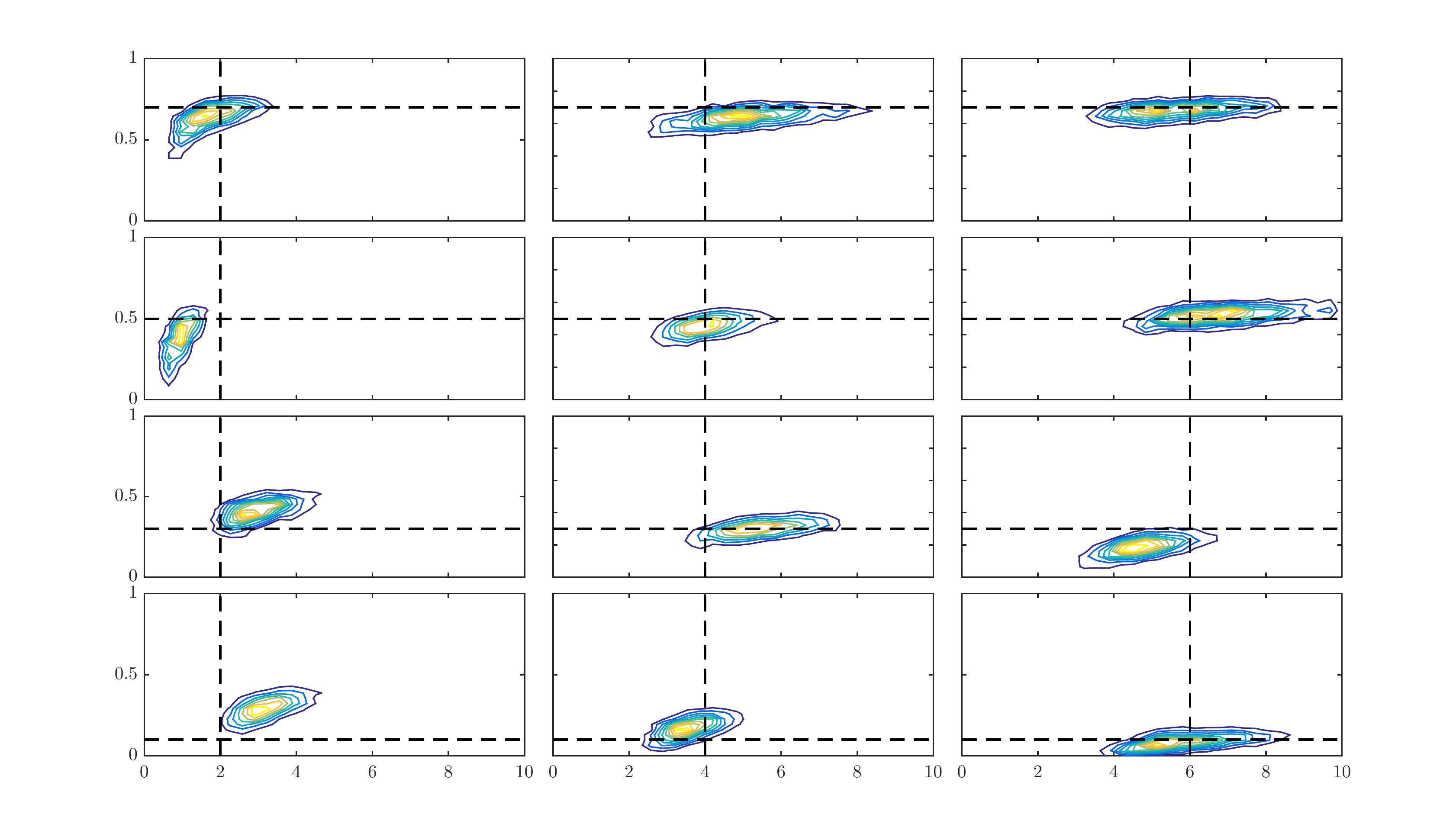}
      };
      % lambda labels
      \node [below=0.19\textheight, left=0.225\linewidth] at (a) {
        \footnotesize{$\lambda=2$}
      };
      \node [below=0.19\textheight, right=-0.45cm] at (a) {
        \footnotesize{$\lambda=4$}
      };
      \node [below=0.19\textheight, right=0.25\linewidth] at (a) {
        \footnotesize{$\lambda=6$}
      };
      % alpha labels
      \node [below=0.09\textheight, left=0.45\linewidth] at (a) [rotate=90] {
        \footnotesize{$\alpha=0.1$}
      };
      \node [below=0.0\textheight, left=0.45\linewidth] at (a) [rotate=90] {
        \footnotesize{$\alpha=0.3$}
      };
      \node [above=0.09\textheight, left=0.45\linewidth] at (a) [rotate=90] {
        \footnotesize{$\alpha=0.5$}
      };
      \node [above=0.18\textheight, left=0.45\linewidth] at (a) [rotate=90] {
        \footnotesize{$\alpha=0.7$}
      };
    \end{tikzpicture}
  }
  }
  \caption{Joint posterior distributions, given graphs generated from an $\RWU$ model.
    $\alpha$ is on the vertical, $\lambda$ on the horizontal axis. 
    Generating parameter values are shown by dashed lines.}
  \label{fig:param:sweep}
  \vspace{-.5cm}
\end{figure}

\subsection{Length scale}
\label{sec:experiments:scale}

The distance of vertices that can form connections is governed by the parameter
$\lambda$ of the random walk, which can hence be interpreted as a form of length-scale. 
This scale can be compared to the mean (over starting vertices) mixing time of a random walk on the observed graph,
as listed in the table above for each data set: If $\lambda$
is significantly smaller, the placement of edges inserted by the random walk strongly depends on the
graph structure; if $\lambda$ is large relative to mixing time, dependence between edges is weak.
\begin{itemize}
\item In \cref{fig:data_posterior}, 
  concentration of the $\lambda$-posterior on small values in $[1,2]$ for the NIPS data
  thus indicates predominantly short-range dependence in the data and, since the mixing time 
  is much larger, strong dependence between edges.
  Concentration of the $\alpha$-posterior near the origin means connections are mainly formed 
  through existing connections.
\item In the SJS network, the posterior peaks near ${\lambda=6}$, and hence near the mean 
  mixing time, with $\alpha$ again small. Thus, 
  the principal mechanism for inserting edges under the model is again the random walk, but a 
  random walk of typical length has almost mixed. Hence, the local connections it creates do not strongly 
  influence each other. 
\item The PPI network exhibits an intermediate scale of dependence, with most posterior mass in the 
  range $\lambda \in [3,5]$. Although structures like pendants and chains indicate strong local 
  dependence, there are also denser, more homogeneously connected regions that indicate weaker
  dependence with longer range.
\end{itemize}

\begin{figure}[tb] 
    \makebox[\textwidth][c]{
    \resizebox{\textwidth}{!}{
  \begin{tikzpicture}
    \node at (0,0) {
      \includegraphics[width=4cm,height=2.5cm]{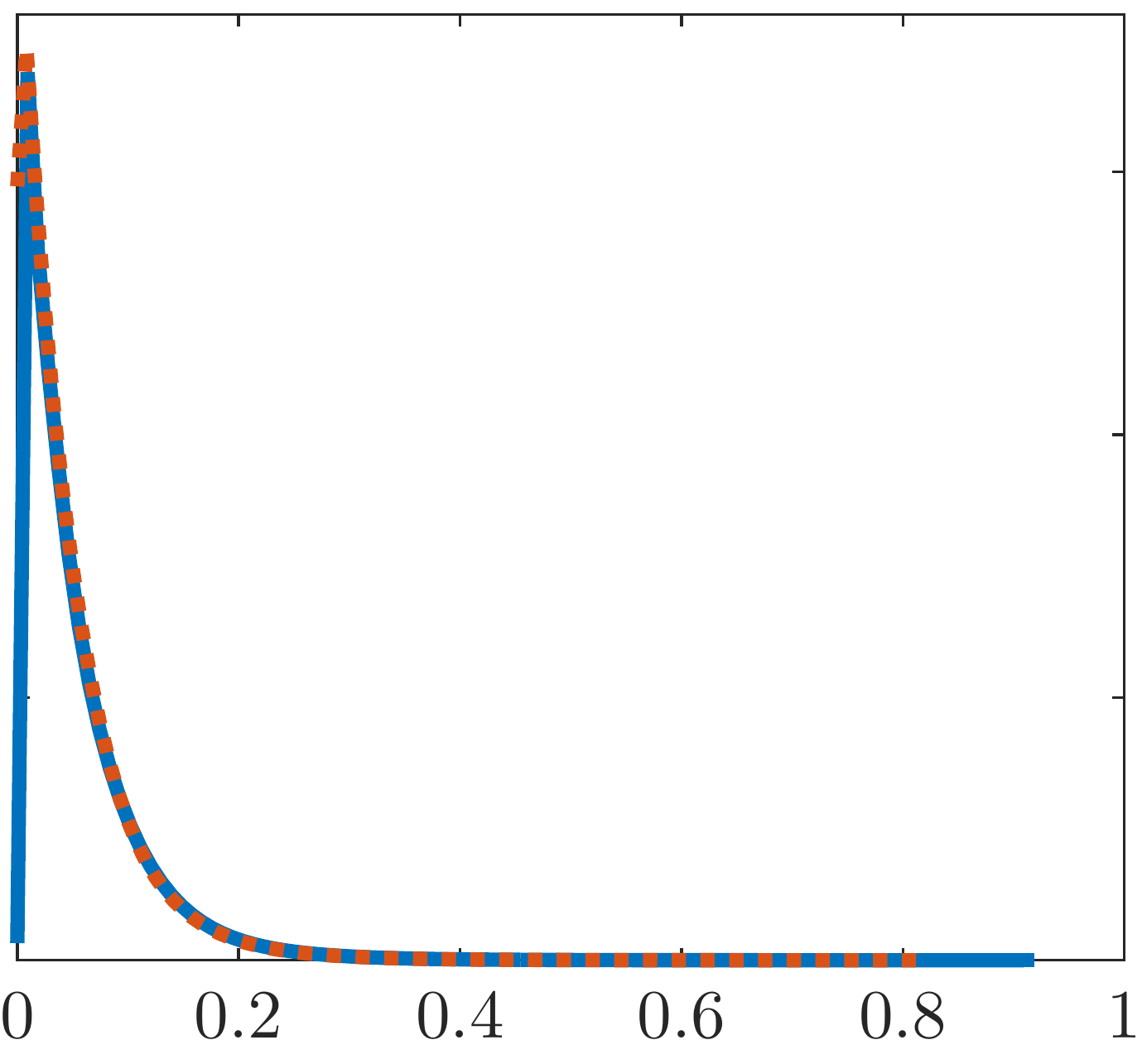}
    };
    \node at (4.6,0) {
      \includegraphics[width=4cm,height=2.5cm]{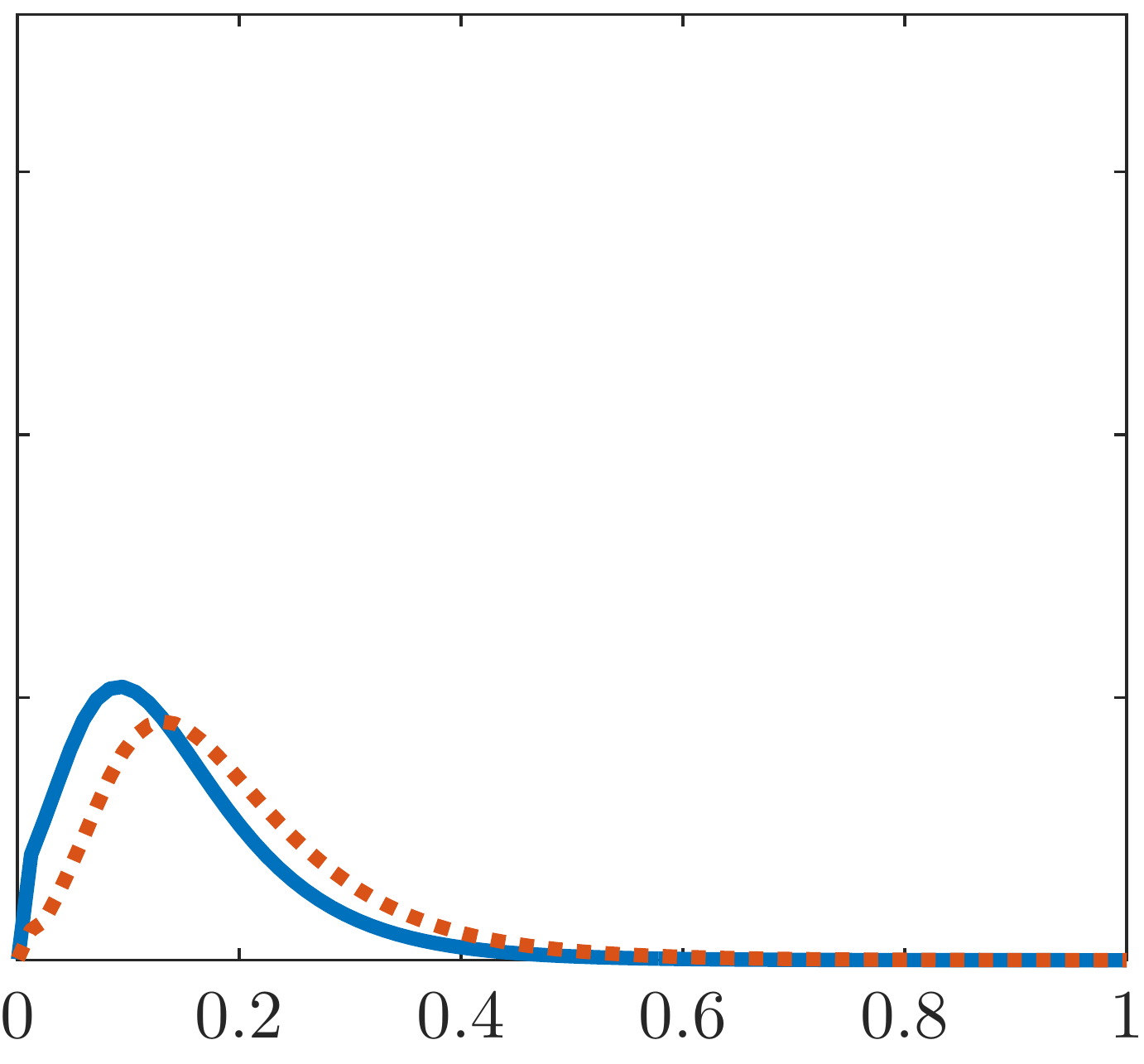}
    };
    \node at (9.2,0) {
      \includegraphics[width=4cm,height=2.5cm]{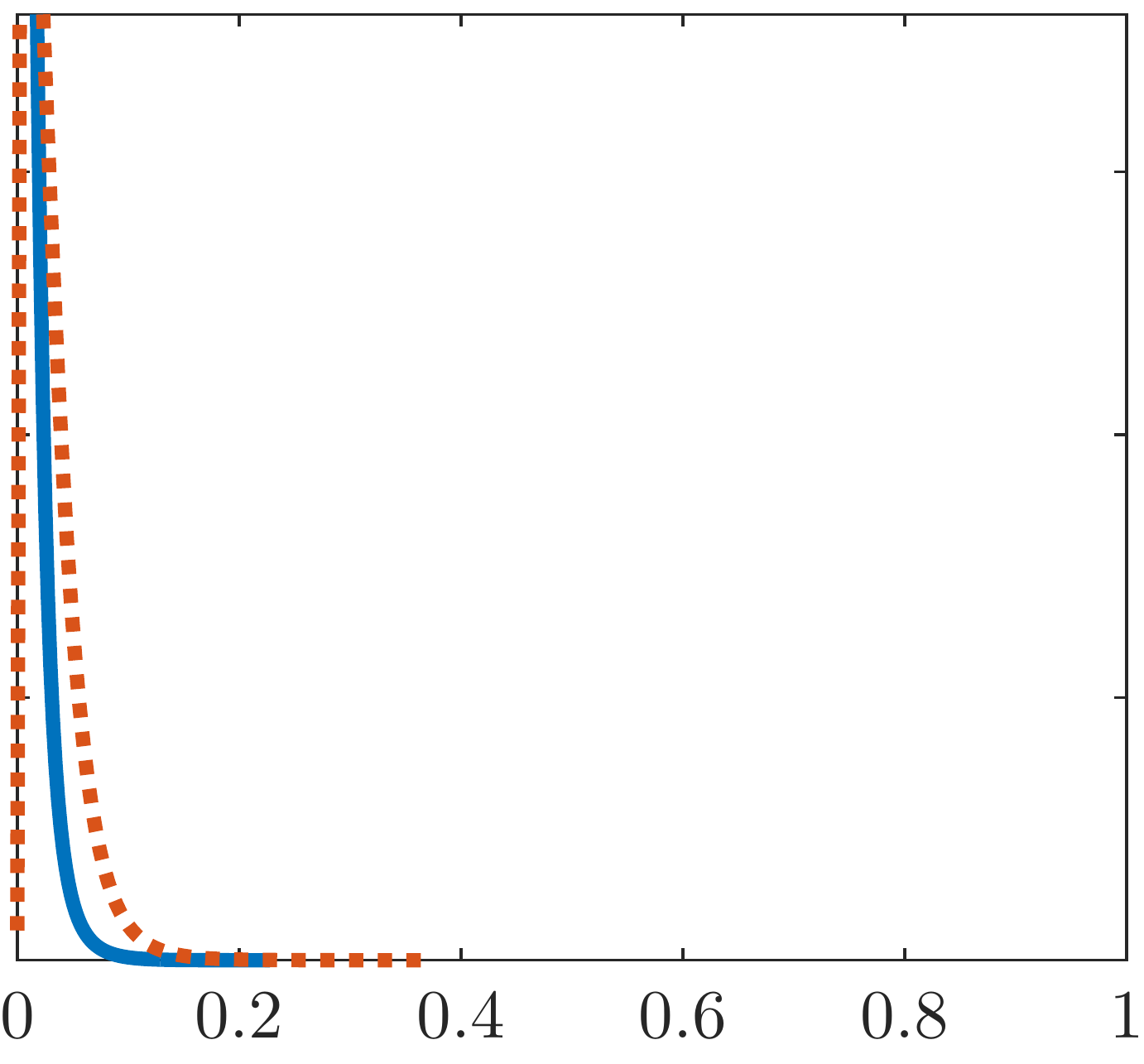}
    };
    \node at (0,-3.2) {
      \includegraphics[width=4cm,height=2.5cm]{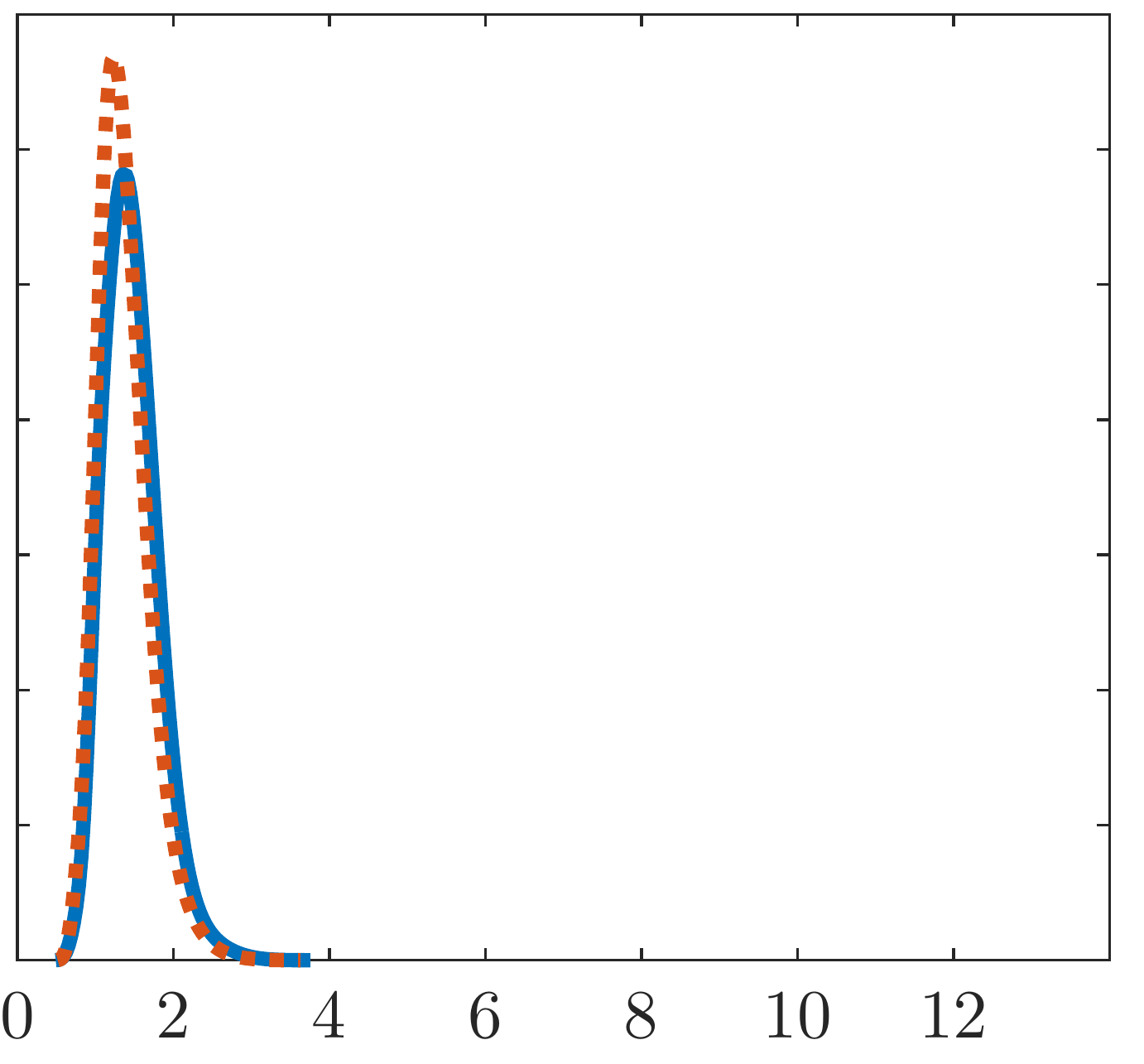}
    };
    \node at (4.6,-3.2) {
      \includegraphics[width=4cm,height=2.5cm]{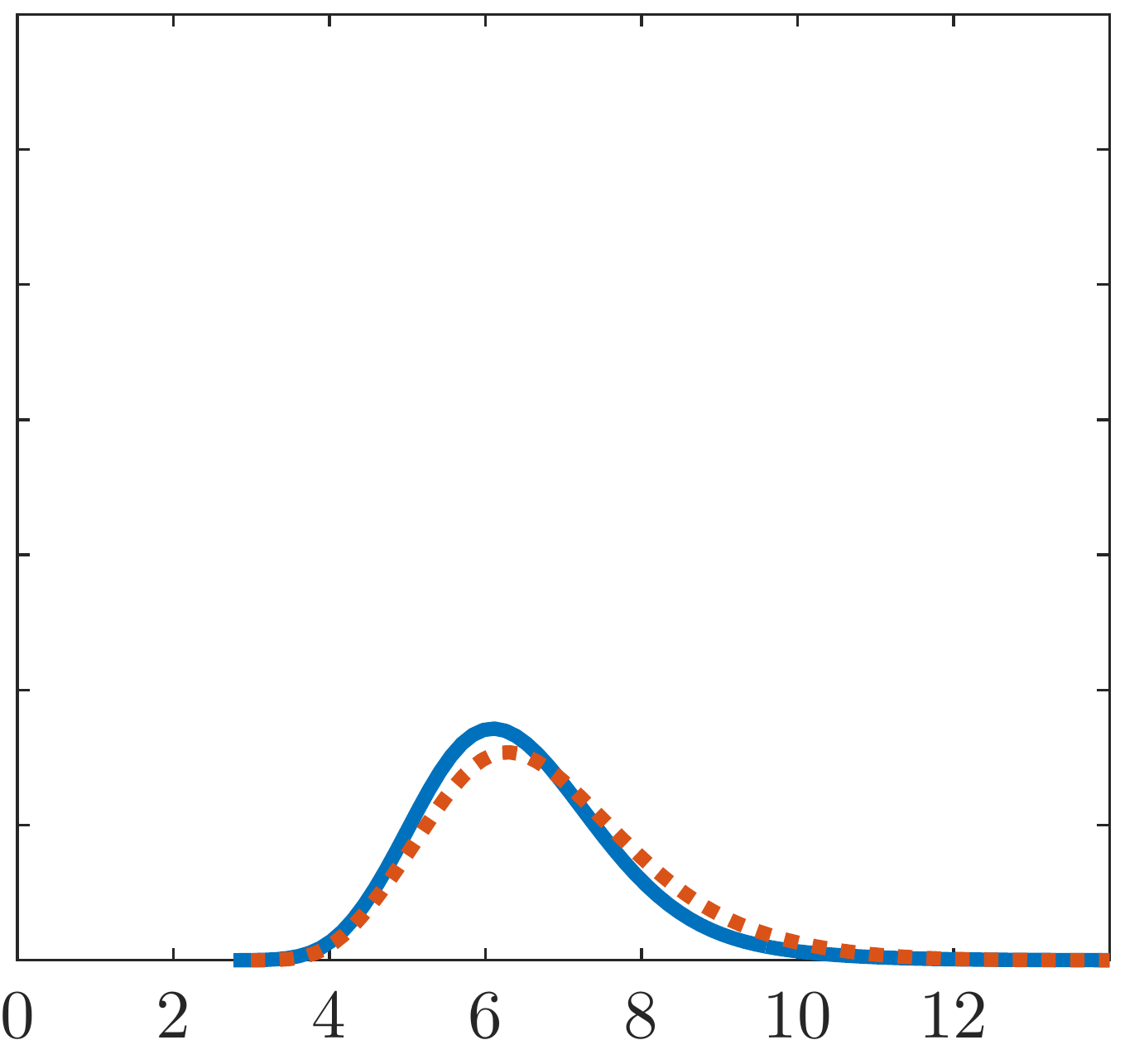}
    };
    \node at (9.2,-3.2) {
      \includegraphics[width=4cm,height=2.5cm]{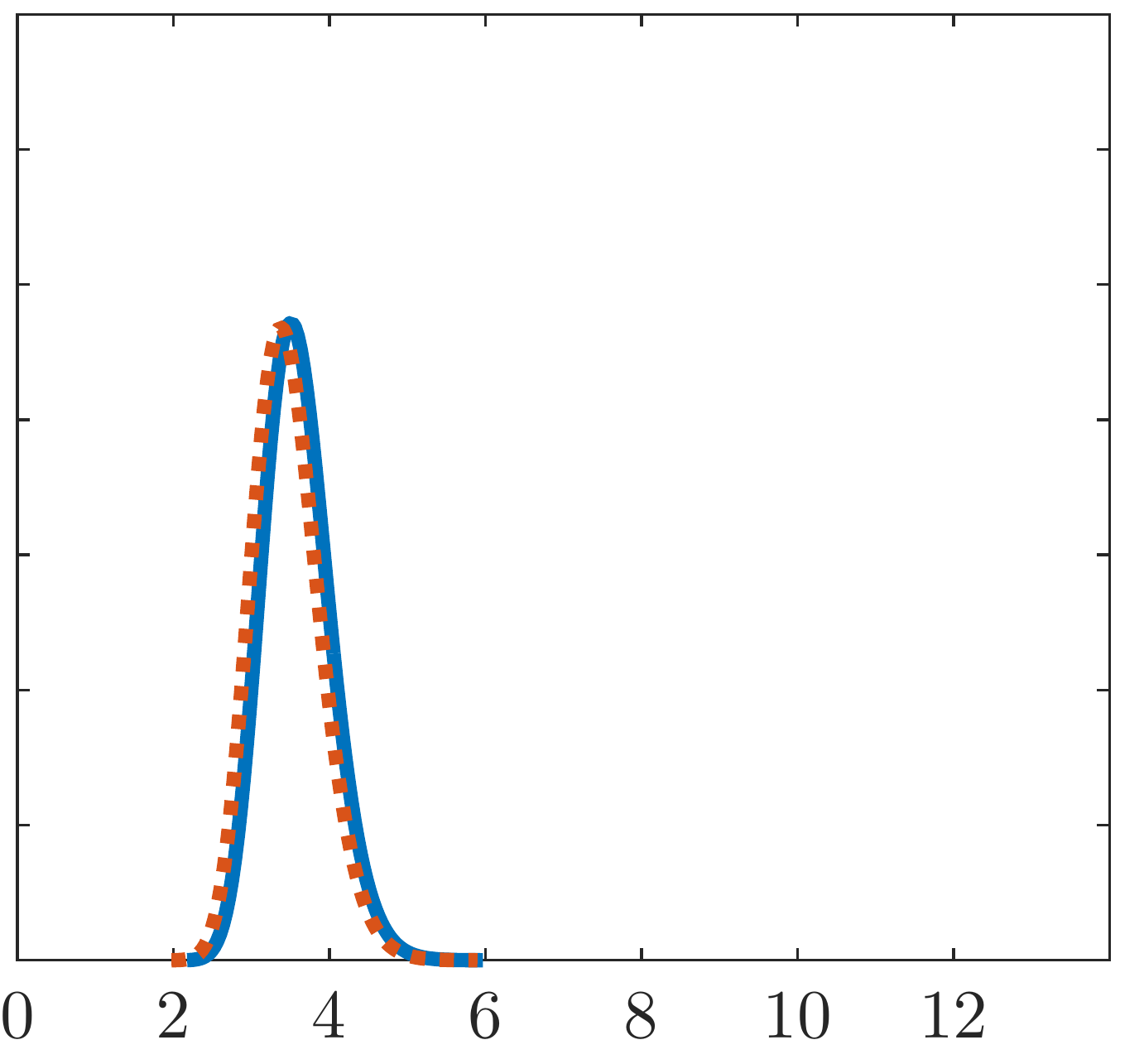}
    };

    \node at (0,-1.5) {\footnotesize $\alpha$};
    \node at (4.6,-1.5) {\footnotesize $\alpha$};
    \node at (9.2,-1.5) {\footnotesize $\alpha$};

    \node at (0,-4.7) {\footnotesize $\lambda$};
    \node at (4.6,-4.7) {\footnotesize $\lambda$};
    \node at (9.2,-4.7) {\footnotesize $\lambda$};

    \node[rotate=90] at (-2.5,0) {\footnotesize $\pi(\alpha\mid G_T)$};
    \node[rotate=90] at (-2.5,-3.2) {\footnotesize $\pi(\lambda\mid G_T)$};
  \end{tikzpicture}
}}
  \caption{Kernel-smoothed estimates of the model posteriors of 
    $\RWU$ (blue/solid) and $\RWSB$ (orange/dotted).
    \textit{Left column}: NIPS data. \textit{Middle}: SJS. \textit{Right}: PPI. 
    (1000 samples each, at lag 40, after 1000 burn-in iterations; 100 samples each are drawn from 10 chains).
  }
  \label{fig:data_posterior}
  \vspace{-.5cm}
\end{figure}
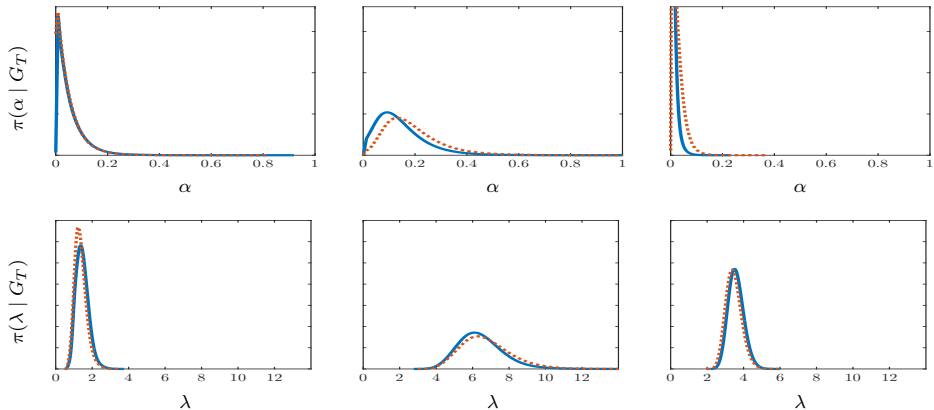

\subsection{Model fitness} \label{sec:model:fitness}

\begin{figure}[t!]
  \resizebox{\textwidth}{!}{
    \begin{tikzpicture}[mybraces]
  \path[use as bounding box]
  (-2.5,1.8) rectangle (13.7,-11);
  \begin{scope}[xshift=0cm]
    \node at (0.2,1.6) {\footnotesize degree};
    \node (a) at (0,0) {
      \includegraphics[width=4.5cm]{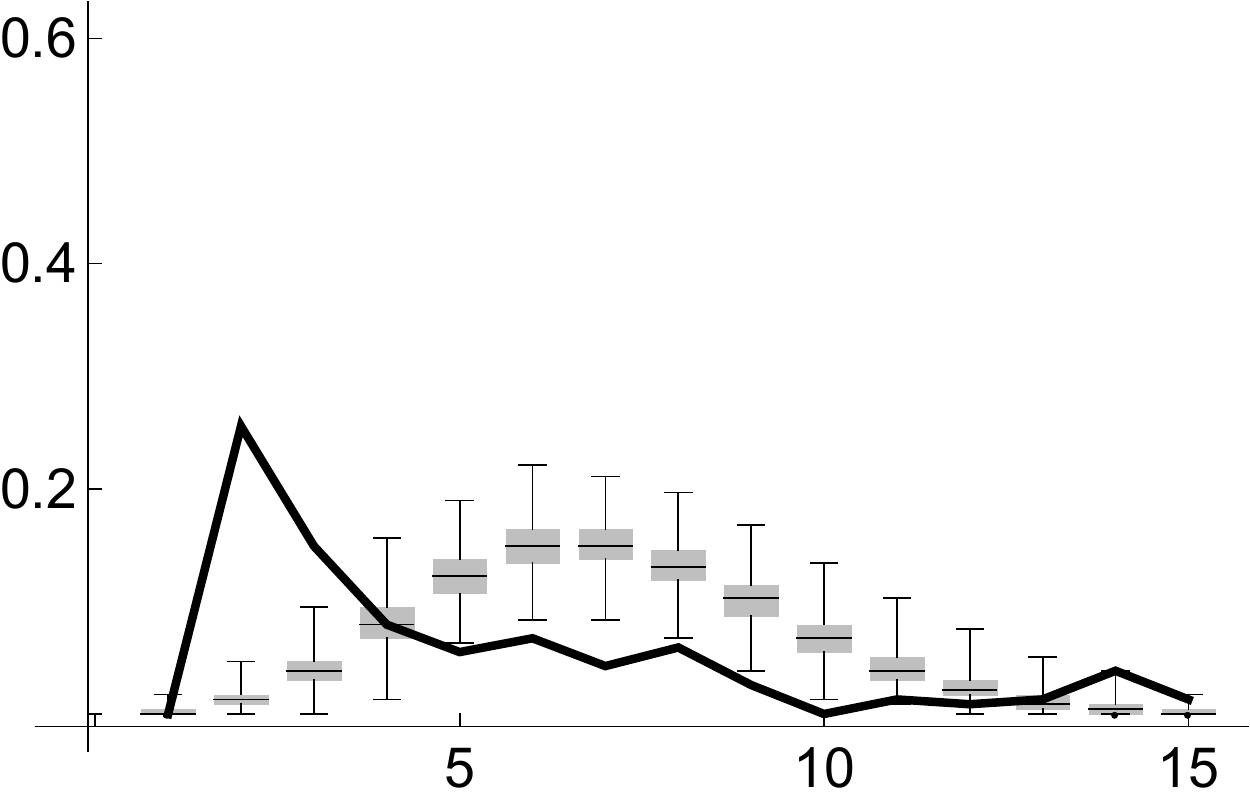}   
    };                                                                                 
    \node (b) at (0,-3) {                                                              
      \includegraphics[width=4.5cm]{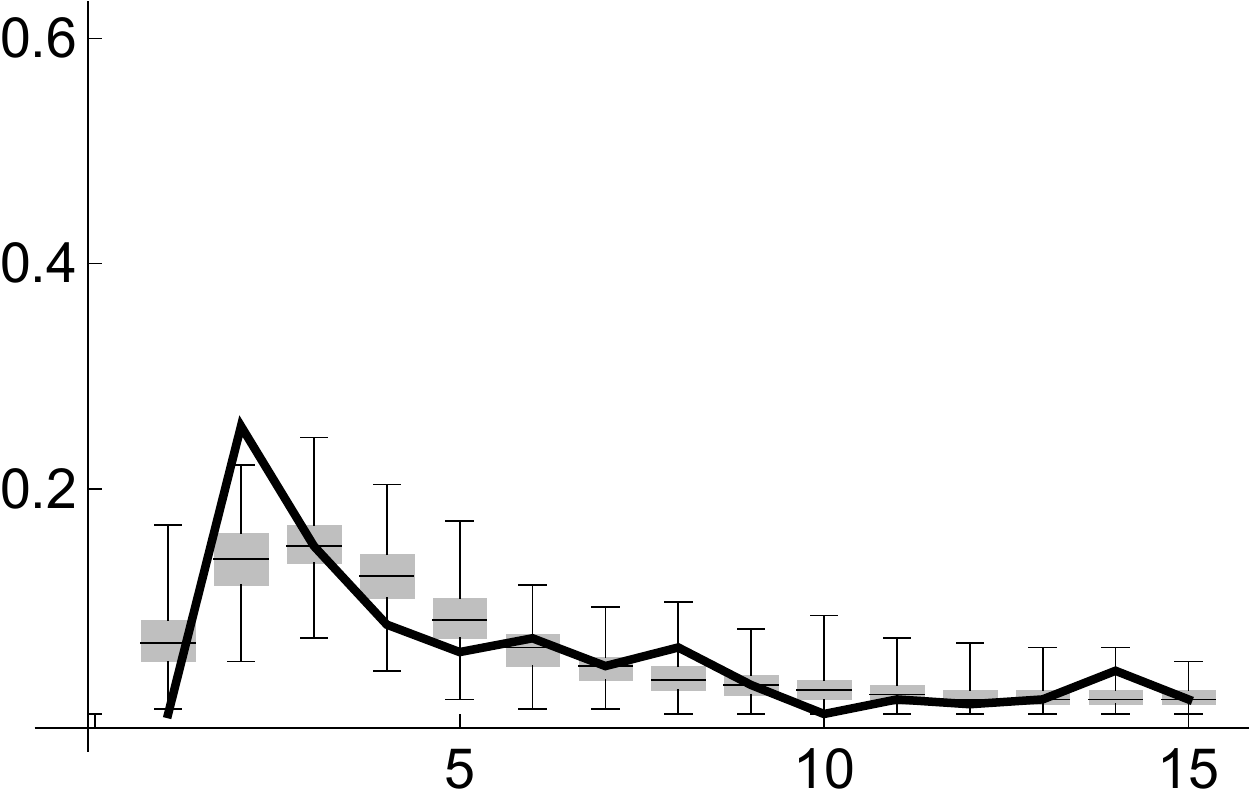}  
    };                                                                                 
    \node (c) at (0,-6) {                                                              
      \includegraphics[width=4.5cm]{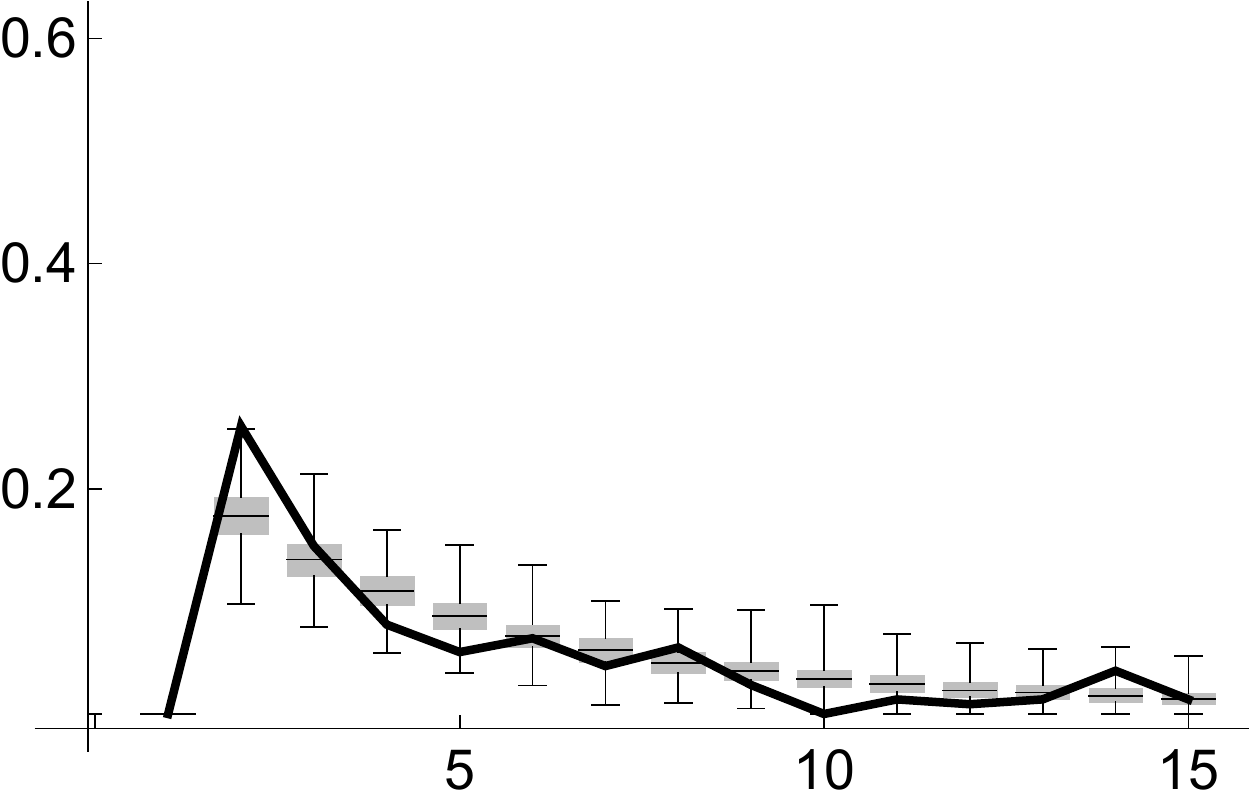} 
    };                                                                                 
    \node (d) at (0,-9) {                                                              
      \includegraphics[width=4.5cm]{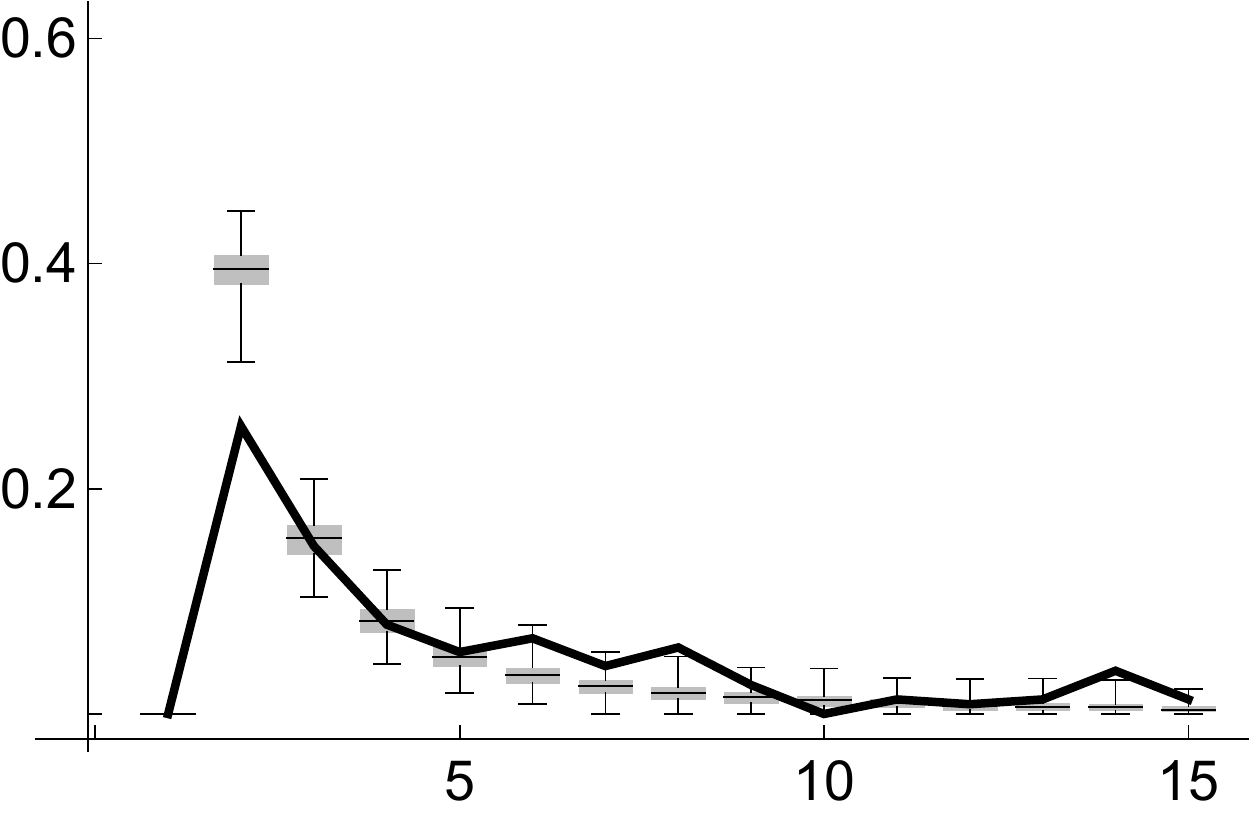}
    };
    \node at (0,-10.8) {\footnotesize{$k$}};
    \node [rotate=90] at (-2.6,-4.5) {\footnotesize{$\log(1 + d_k)$}};
  \end{scope}
  \begin{scope}[xshift=5.3cm]
    \node at (0.2,1.6) {\footnotesize edgewise shared partner};
    \node (a) at (0,0) {
      \includegraphics[width=4.5cm]{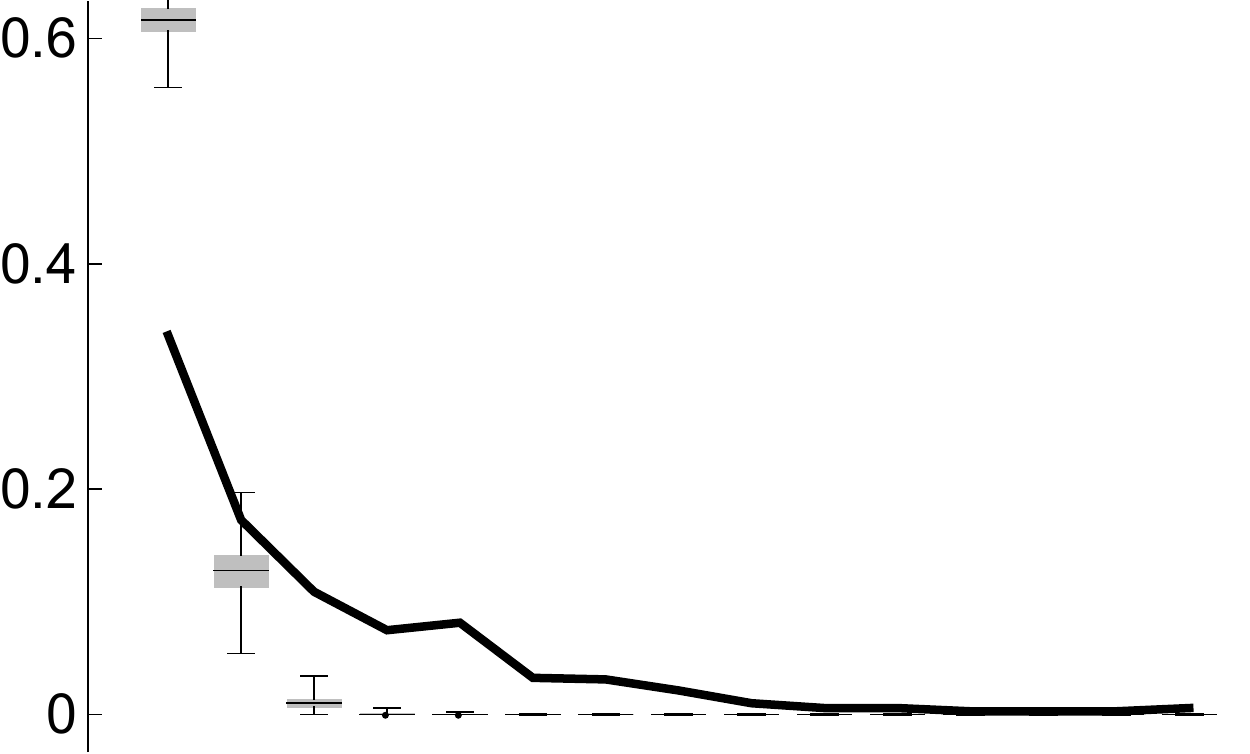}   
    };                                                                                 
    \node (b) at (0,-3) {                                                              
      \includegraphics[width=4.5cm]{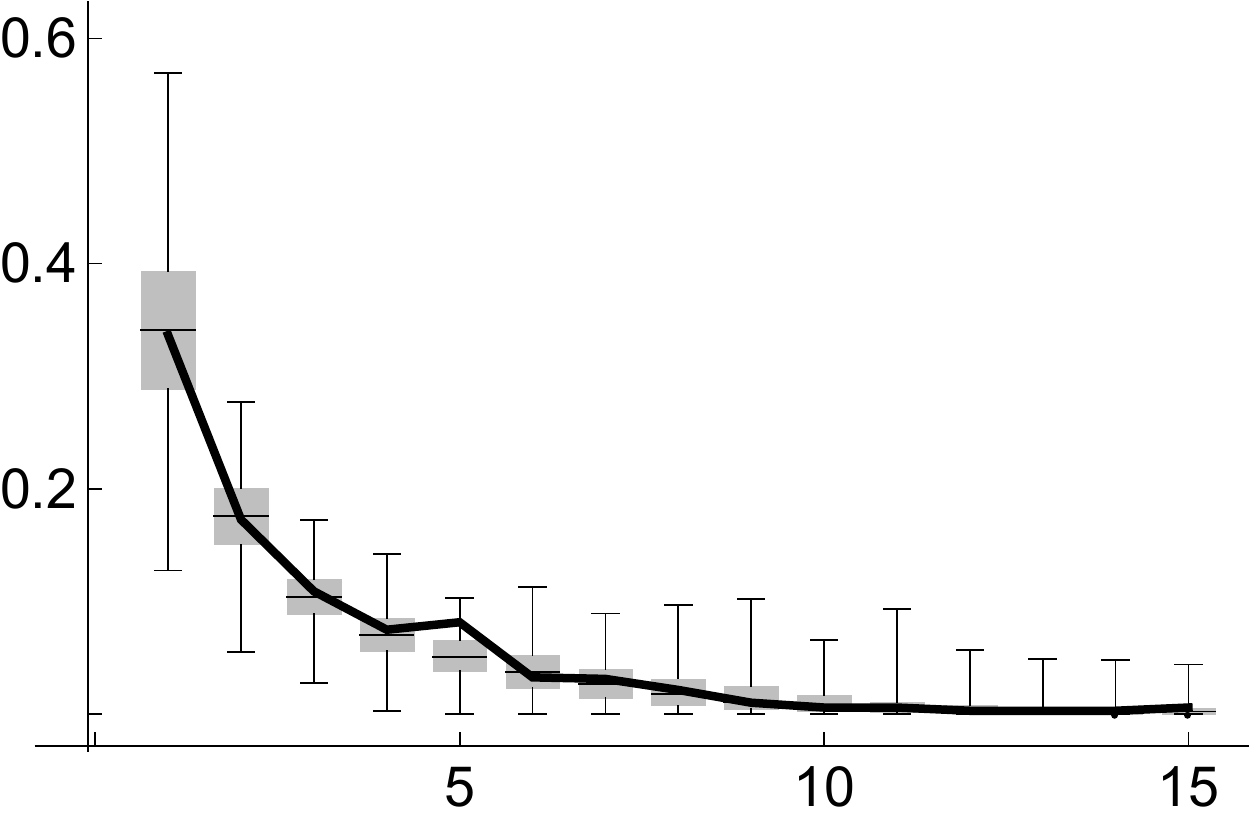}  
    };                                                                                 
    \node (c) at (0,-6) {                                                              
      \includegraphics[width=4.5cm]{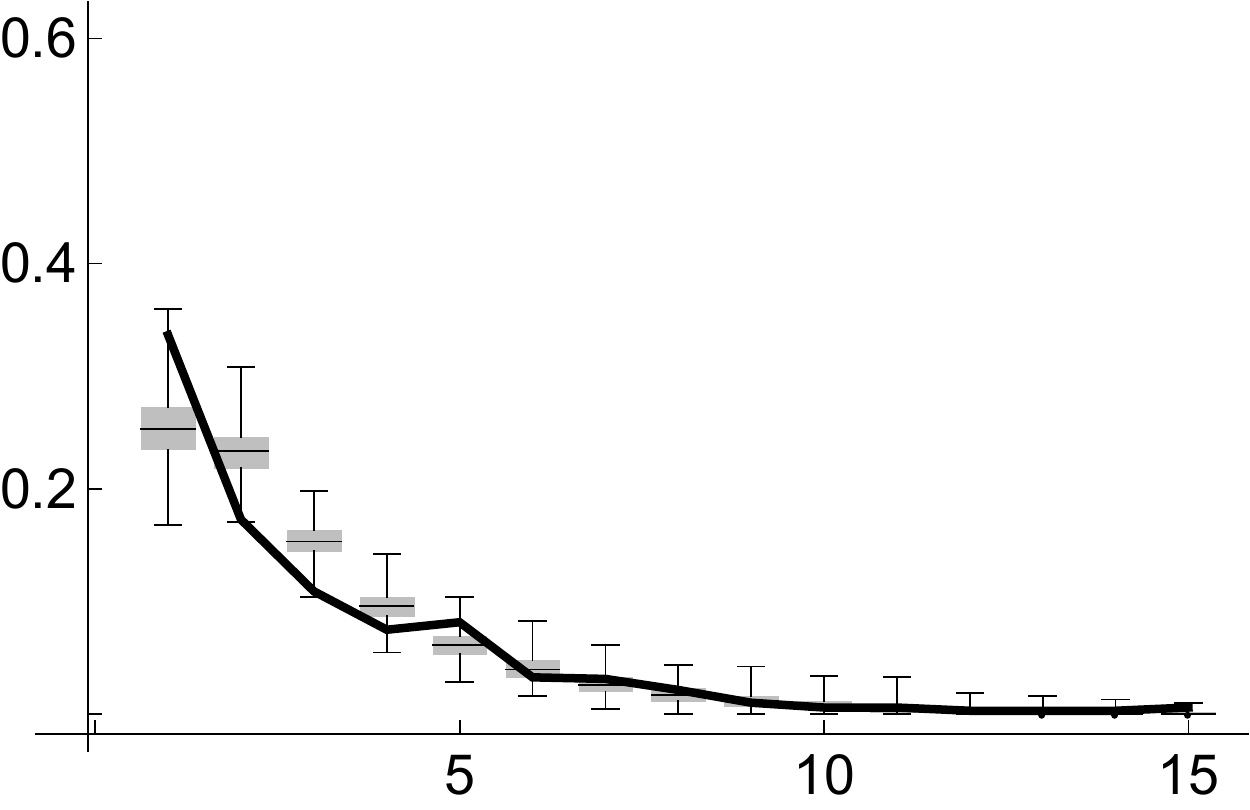} 
    };                                                                                 
    \node (d) at (0,-9) {                                                              
      \includegraphics[width=4.5cm]{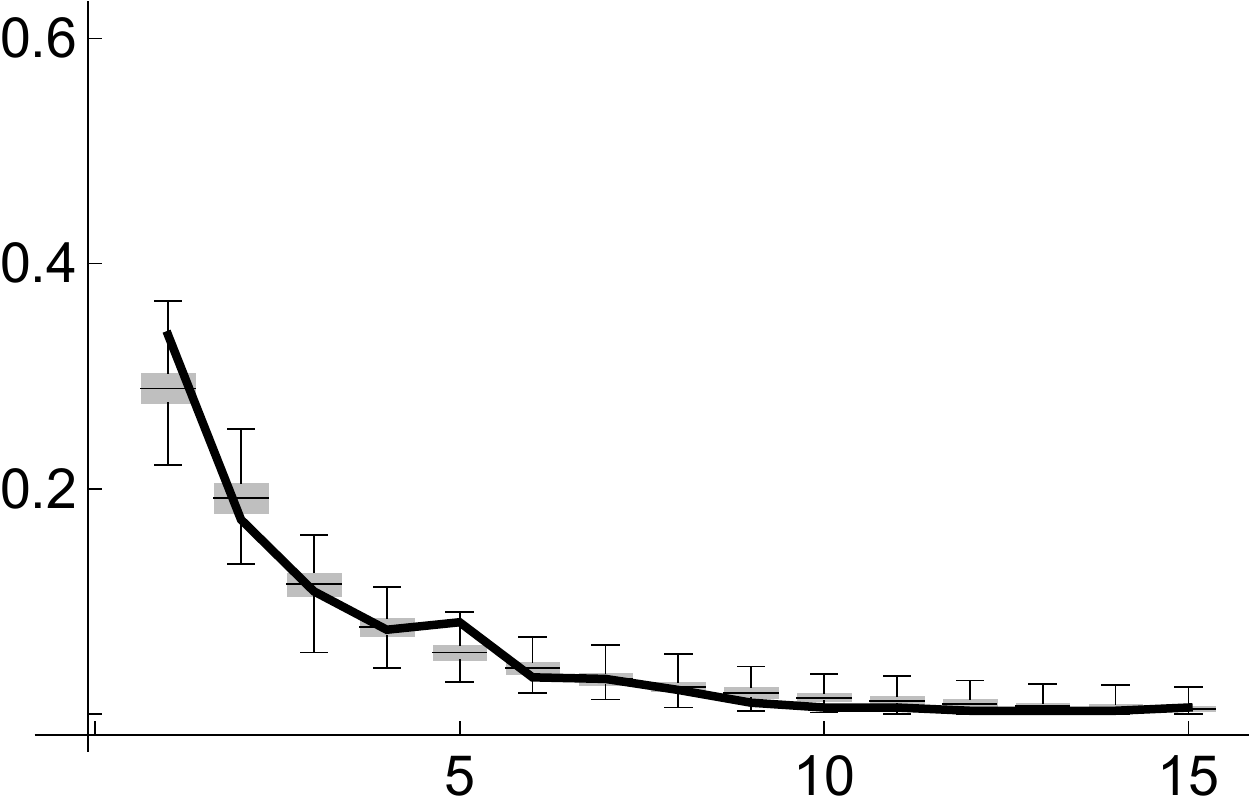}
    };
    \node at (0,-10.8) {\footnotesize{$k$}};
    \node [rotate=90] at (-2.6,-4.5) {\footnotesize{$\log(1 + \chi_k)$}};
  \end{scope}
  \begin{scope}[xshift=10.6cm]
    \node at (0.2,1.6) {\footnotesize pairwise geodesic};
    \node (a) at (0,0) {
      \includegraphics[width=4.5cm]{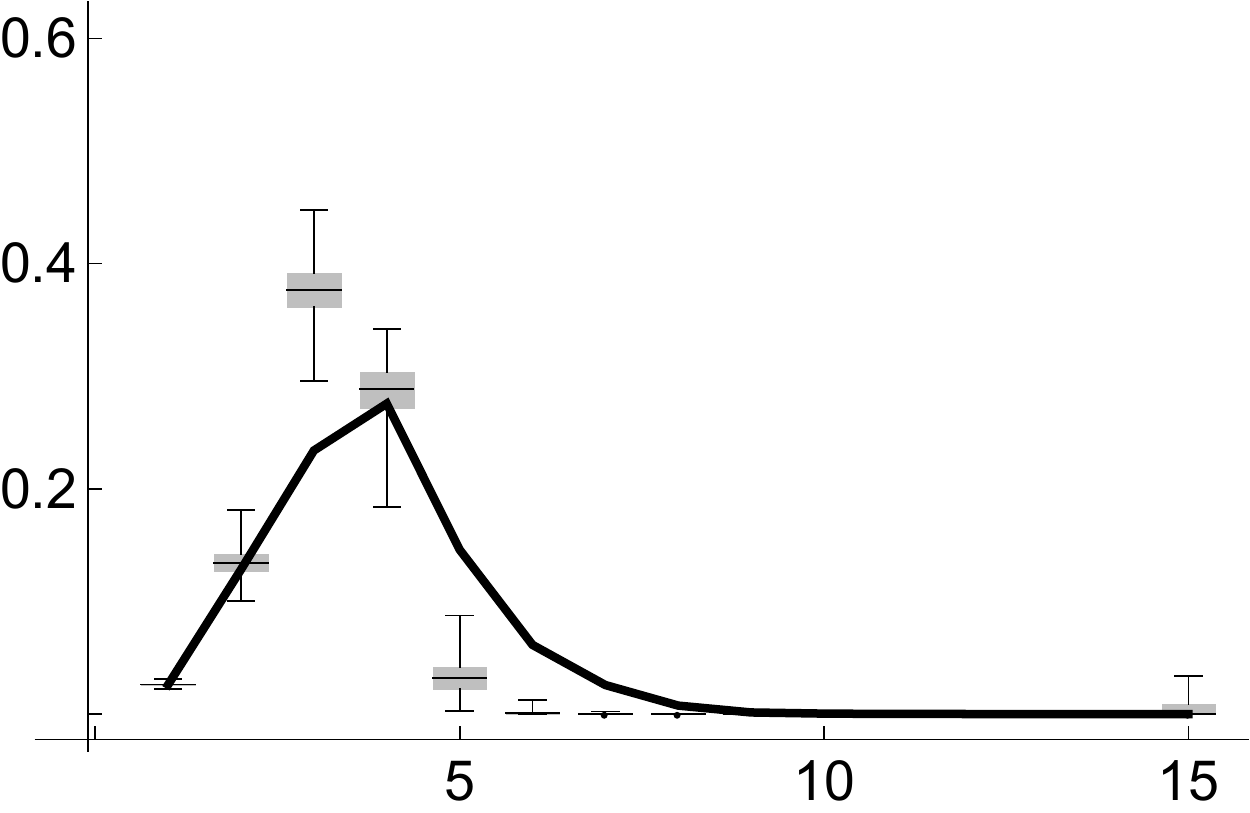}   
    };                                                                                 
    \node (b) at (0,-3) {                                                              
      \includegraphics[width=4.5cm]{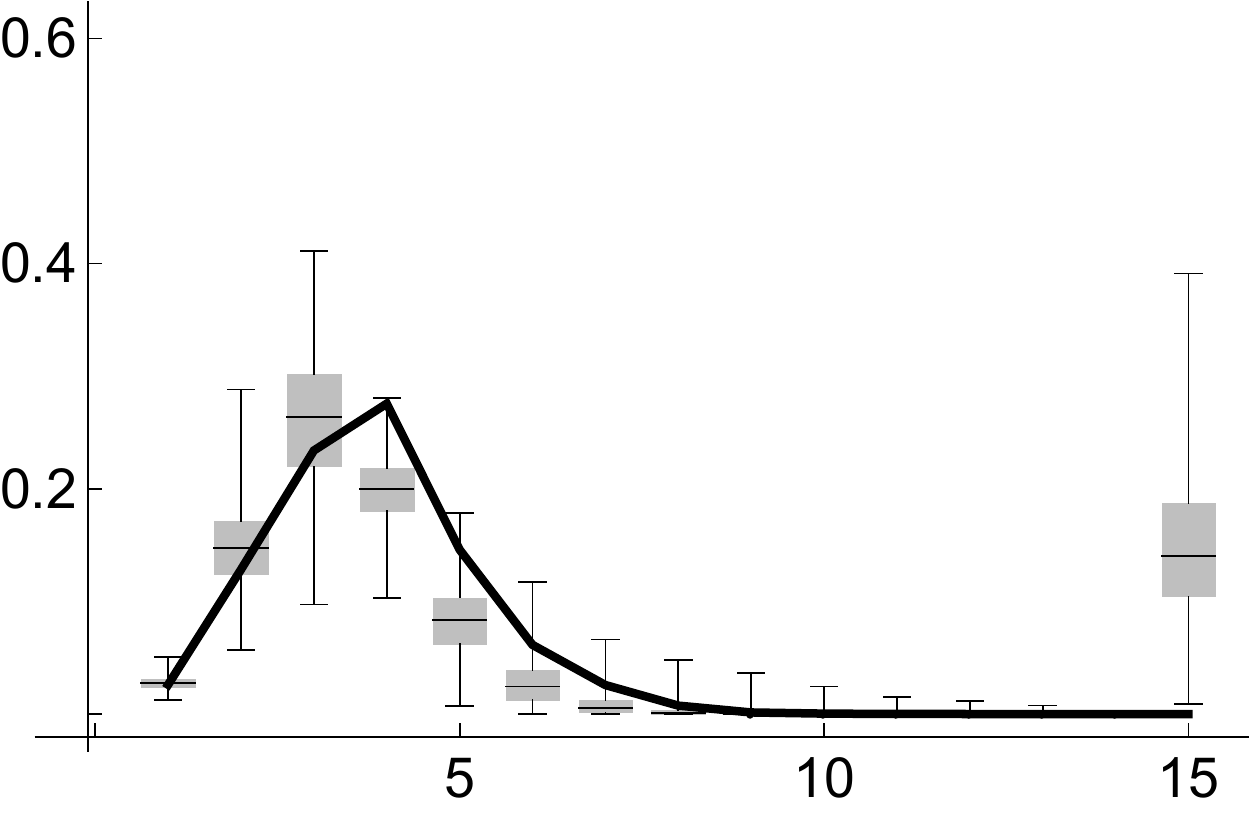}  
  		};                                                                                 
    \node (c) at (0,-6) {                                                              
  			\includegraphics[width=4.5cm]{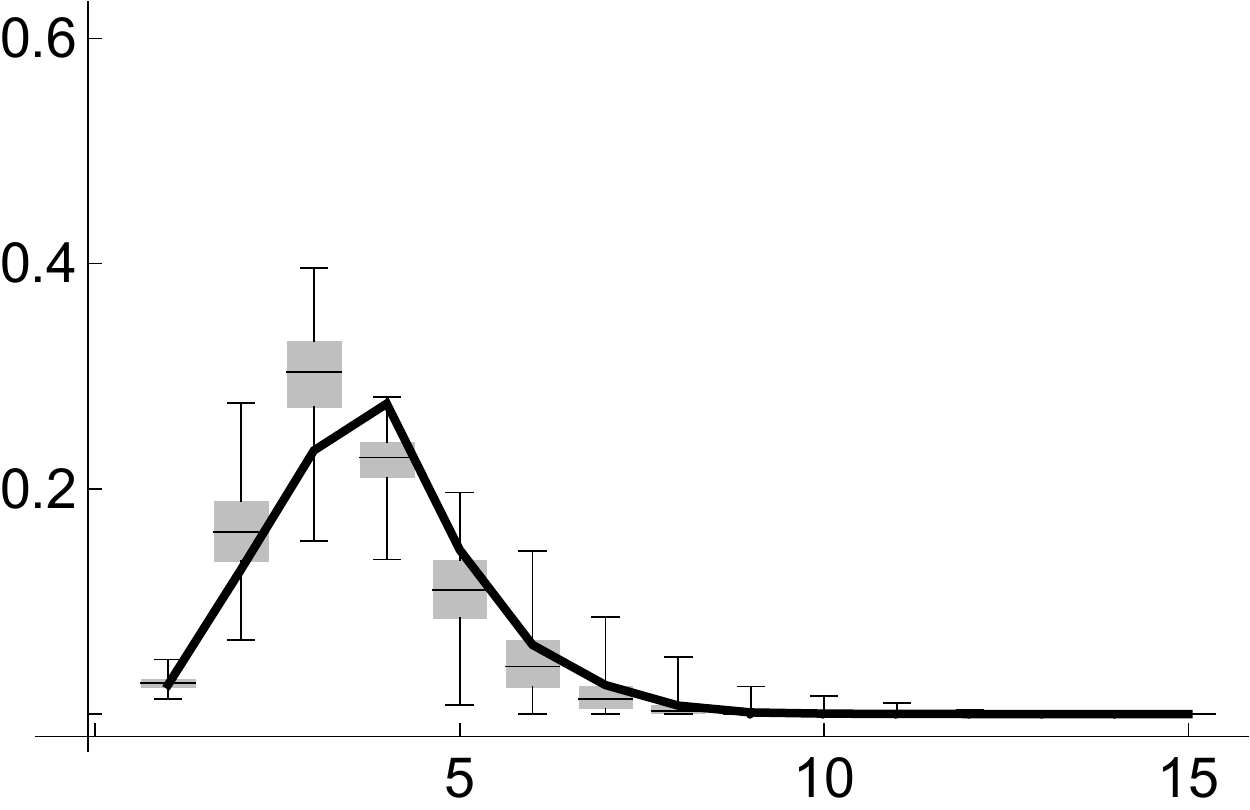} 
    };                                                                                 
    \node (d) at (0,-9) {                                                              
      \includegraphics[width=4.5cm]{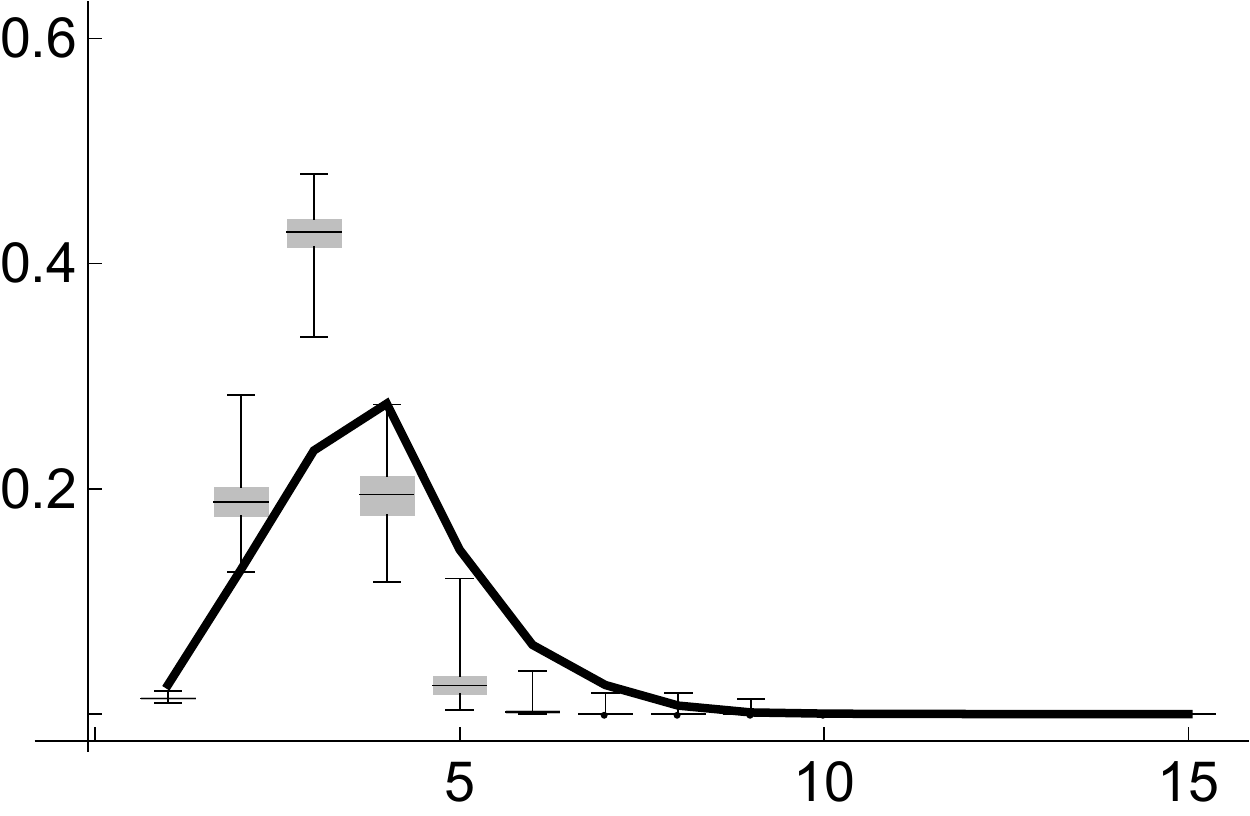}
    };
    \node at (0,-10.8) {\footnotesize{$k$}};
    \node [rotate=90] at (-2.6,-4.5) {\footnotesize{$\log(1 + \gamma_k)$}};
  	  \end{scope}
  \begin{scope}[xshift=11cm]
    \draw[brace] ($(a.north east)+(0,-.1)$)--($(a.south east)+(0,.4)$);
    \node[rotate=270] at ($(a.east)+(.5,0)$) {\footnotesize Erd\H{o}s-R\'{e}nyi};
    
    \draw[brace] ($(b.north east)+(0,-.1)$)--($(b.south east)+(0,.4)$);
    \node[rotate=270] at ($(b.east)+(.5,0)$) {\footnotesize IRM};
    
    \draw[brace] ($(c.north east)+(0,-.1)$)--($(c.south east)+(0,.4)$);
    \node[rotate=270] at ($(c.east)+(.5,0)$) {\footnotesize RW (uniform)};
    
    \draw[brace,gray!80!black] ($(d.north east)+(0,-.1)$)--($(d.south east)+(0,.4)$);
    \node[rotate=270] at ($(d.east)+(.5,0)$) {\footnotesize RW (size-biased)};
  \end{scope}
\end{tikzpicture}

  }
  \caption{Estimated PPDs for the PPI data set of three statistics: Normalized degree $d_k$,
    normalized edgewise shared partner statistic $\chi_k$, and normalized pairwise geodesic 
    $\gamma_k$.
    Results are shown for four models: Erd\H{o}s-R\'{e}nyi (top row), IRM (second row), $\RWU$ (third row),
    and $\RWSB$ (bottom row). The black line represents the distribution from the PPI data.}
  \label{fig:resampling:degree}
  \vspace{-.5cm}
\end{figure}
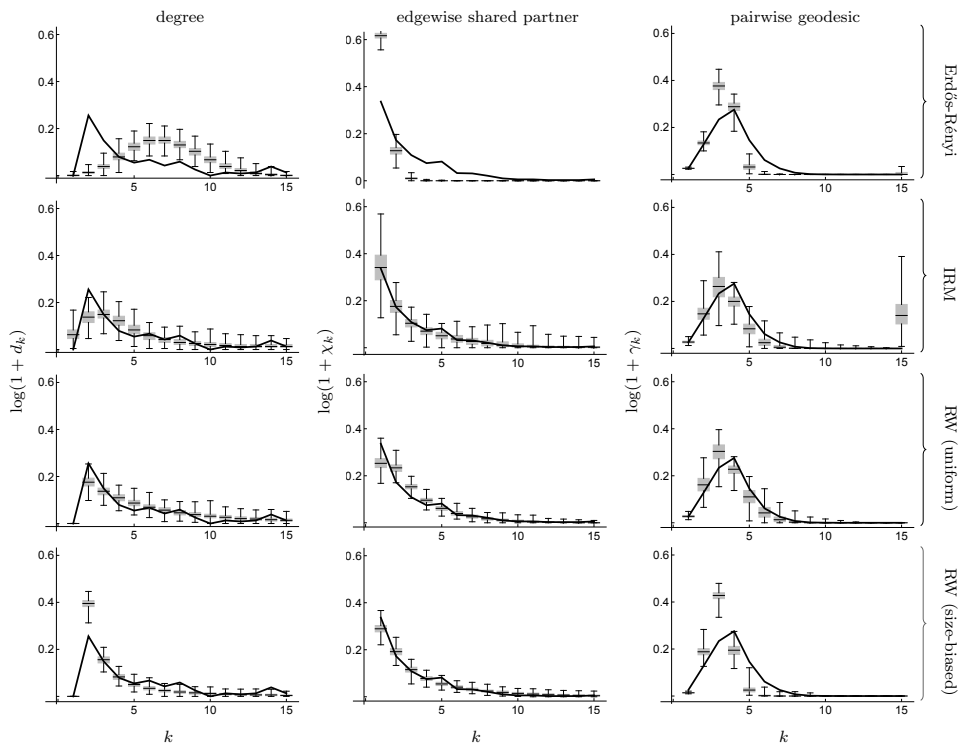

Assessing model fitness on network data is difficult: For other types of data,
cross validation is often the tool of choice, and indeed, cross-validated link prediction (where links and non-links 
are deleted uniformly and independently at random) is widely used for evaluation in the
machine learning literature. Even if link prediction is considered the relevant statistic of interest,
however, cross-validating network data requires subsampling
the observed network. Any choice of such a subsampling algorithm implies strong assumptions on how
the data was generated; it may also favor one model over another. 

We therefore use a protocol developed by \citet{Hunter:Goodreau:Handcock:2008}, 
who compare models to data by fixing a set of network statistics. The fitted model is evaluated 
by comparing the chosen statistics of the data to those of networks simulated from the model.
The selected statistics specify which properties are considered important in assessing fit.
To capture a range of structures, the following count statistics are proposed in 
\citep{Hunter:Goodreau:Handcock:2008}, which we also adopt here:
\begin{itemize}
  \item Normalized degree statistics (ND), $d_k$, the number of vertices of degree $k$, divided by the total number of vertices.
  \item Normalized edgewise shared partner statistics (NESP), $\chi_k$, the number of unordered, connected pairs 
    $\{ i,j \}$ with exactly $k$ common neighbors, divided by the total number of edges.
  \item Normalized pairwise geodesic statistics (NPG), $\gamma_k$, the number of unordered pairs $\{ i,j \}$ with distance $k$ in the graph, divided by the number of dyads.
\end{itemize}
In a Bayesian setting, executing the protocol amounts to performing posterior predictive checks~\citep{Box:1980,Gelman:Meng:Stern:1996} via the following procedure:
\begin{enumerate}[label=(\arabic*)]
	\item Sample the model parameters from the posterior, $\seqParams \sim \pi(\seqParams | G_T)$.
	\item Simulate a graph of the same size as the data, $G^{(s)} \sim P(G | \seqParams)$.
	\item Calculate the statistic(s) of interest for the simulated graph, $f(G^{(s)})$.
\end{enumerate}
The value ${f(G^{(s)})}$ is then a sample from the posterior predictive distribution (PPD) of the statistic $f$ under the model. Small PPD mass around the observed value of $f$ indicates
the model does not explain those properties of the data that $f$ measures.

PPDs are estimated for the following models:
The Erd\H{o}s-R\'{e}nyi model (ER); the infinite relational model (IRM) with a Chinese Restaurant Process prior on the number of blocks~\citep{Kemp:etal:2006}; the infinite edge partition model with a hierarchical gamma process (EPM)~\citep{Zhou:2015}; the ACL model (\cref{sec:pa}); and the
$\RWU$ and $\RWSB$ models. The table lists total variation distance between the empirical
distribution of each statistic on the observed graph and on graphs generated from the respective models. 
Standard errors are computed over 1000 samples. Smaller values indicate a better fit:

\begin{center}
  \makebox[\textwidth][c]{
    \resizebox{.9\textwidth}{!}{
\begin{tabular}{rcccccc}
  & \multicolumn{3}{c}{\textit{PPI data}}
  & \multicolumn{3}{c}{\textit{NIPS data}}\\
  \midrule
  \textit{Model} & \textit{Degree} & \textit{ESP} & \textit{Geodesic} 
  & \textit{Degree} & \textit{ESP} & \textit{Geodesic} \\                    
\midrule                                         
EPM & ${0.49\pm .03}$ & ${0.31\pm .07}$ & ${0.65 \pm .04}$   
& ${0.57 \pm .06}$ & ${0.43 \pm .15}$ & ${0.72 \pm .05}$ \\  

IRM & ${0.30 \pm .04}$ & ${0.15 \pm .07}$ & ${0.25 \pm .07}$   
 & ${ 0.29 \pm .08 }$ & ${ 0.46 \pm .10 }$ & ${0.36 \pm .13}$ \\   

ER & ${ 0.57 \pm .02 }$ & ${ 0.45 \pm .03 }$ & ${ 0.23 \pm .02}$ 
& ${ \bf 0.26 \pm .06 }$ & ${ 0.69 \pm .06 }$ & ${ 0.41 \pm .06}$ \\    

ACL & ${ 0.28 \pm .02 }$ & ${ \bf 0.09 \pm .02 }$ & ${ 0.34 \pm .03}$ 
& ${ 0.42 \pm .05 }$ & ${ 0.51 \pm .06 }$ & ${ 0.50 \pm .04}$ \\

RW-U & ${ \bf 0.23 \pm .03 }$ & ${ 0.17 \pm .02 }$ & ${ \bf 0.16 \pm .08}$ 
& ${ \bf 0.26 \pm .04 }$ & ${ \bf 0.33 \pm .05 }$ & ${ \bf 0.22 \pm .08 }$ \\  

RW-SB & ${ 0.27 \pm .02 }$ & ${ 0.11 \pm .02 }$ & ${ 0.34 \pm .04}$ 
& ${ 0.42 \pm .05 }$ & ${ 0.39 \pm .06 }$ & ${ 0.45 \pm .07}$ 
\end{tabular}
}}
\end{center}

For the ER model, the IRM, and the two random walk models, PPD estimates 
are shown in more detail in 
\cref{fig:resampling:degree}, on a logarithmic scale. 
In terms of the protocol of \citep{Hunter:Goodreau:Handcock:2008}, 
the uniform random walk model provides the best fit on both data sets.
The IRM does similarly well, although
it does place significant posterior mass on $d_0$ and $\gamma_{\infty}$ (since it tends 
to generate isolated vertices). The good fit comes at the price of model complexity:
The IRM is a prior on stochastic blockmodels with an infinite number of classes,
a finite number of which are invoked to explain a finite graph.
For the PPI data set, for example, the posterior sharply concentrates at 9 classes,
which amounts to 53 scalar parameters, compared to 2 for the RW model.

\cref{fig:resamples} compares reconstructions of a
data set from different fitted models.
Such visual comparisons should be treated with caution, but the plots
underscore that for some types of data, random walk models
provide an arguably better structural fit.

\begin{figure}[tb] 
	\makebox[\textwidth][c]{
    \resizebox{\textwidth}{!}{
        \begin{tikzpicture}
          \path[use as bounding box]
          (-1.5,-2.5) rectangle (13.5,2);
          \node at (-.4,0) {
            \includegraphics[width=3.5cm,angle=-110]{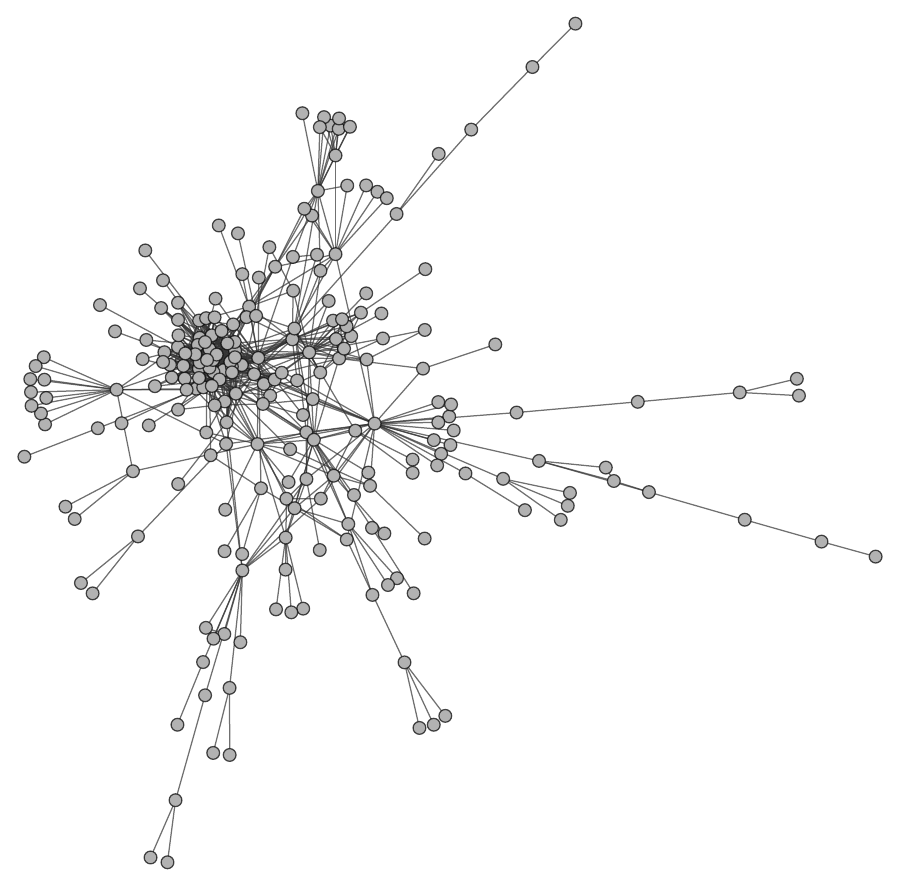}
          };
          \node at (2.7,0) {
            \includegraphics[width=3.5cm,angle=80]{interactomeAdj_resampled_unif2_graph.pdf}
          };
          \node at (5.4,0) {
            \reflectbox{\includegraphics[width=3.5cm,angle=90]{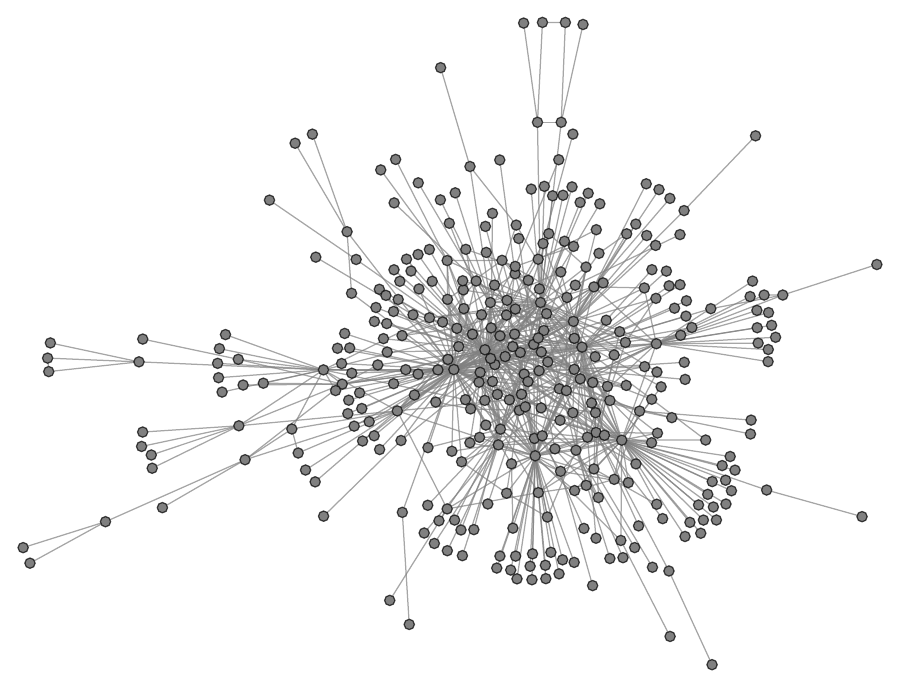}}
          };
	  \node at (8.6,0) {
            \includegraphics[width=3.5cm]{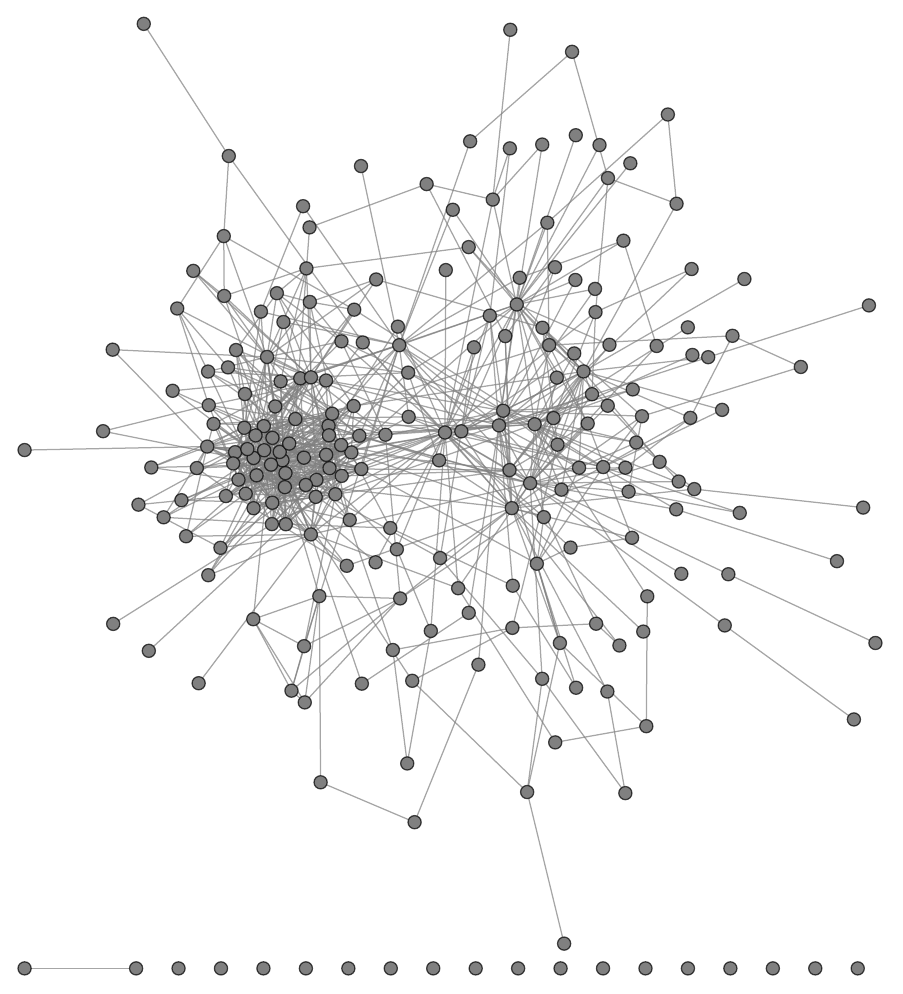}
          };
          \node at (12,0) {
            \includegraphics[width=3.5cm,angle=90]{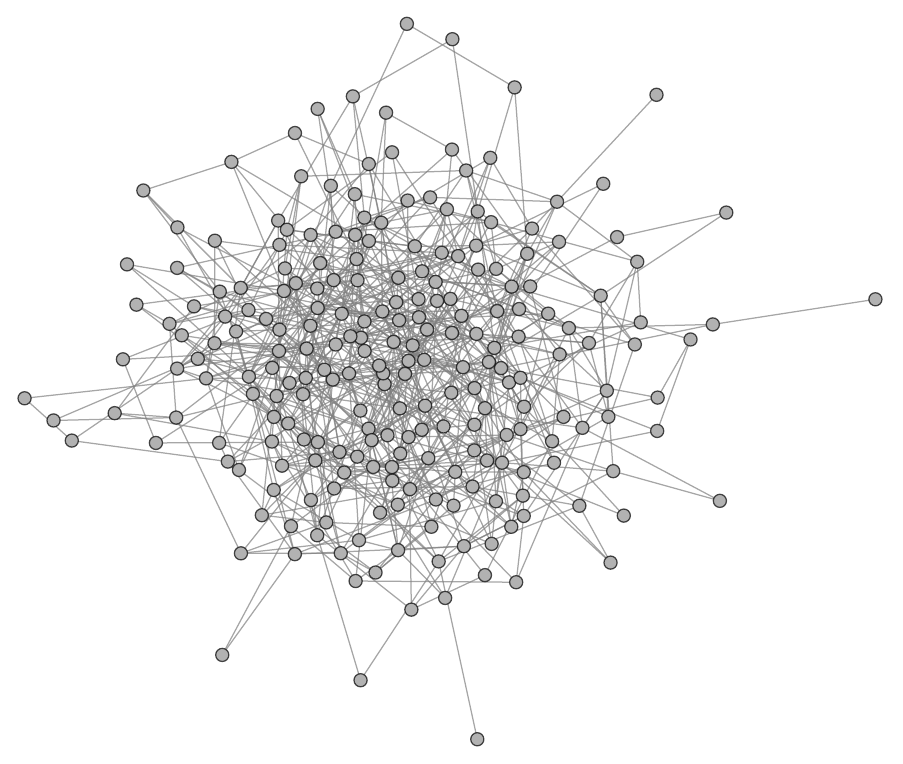}
          };
          \node at (-.4,-2.2) {\footnotesize (i)};
          \node at (2.7,-2.2) {\footnotesize (ii)};
          \node at (5.4,-2.2) {\footnotesize (iii)};
          \node at (8.6,-2.2) {\footnotesize (iv)};
          \node at (12,-2.2) {\footnotesize (v)};
	\end{tikzpicture}
        }}
	\caption{Reconstructions of the PPI network (i), sampled from the posterior mean of 
          (ii) the $\RWU$ model, 
          (iii) the $\RWSB$ model, (iv) the IRM, and 
          (v) the Erd\H{o}s-R\'{e}nyi model.}
	\label{fig:resamples}
        \vspace{-.6cm}
\end{figure}

\subsection{Latent order and vertex centrality} \label{sec:vertex:centrality}

\kword{Centrality} refers to the importance of a vertex 
for interactions in a network. A simple measure of centrality is vertex degree;
others include Eigenvalue, Katz and Information Centrality \citep[e.g.][]{Newman:2009}.
Under a $\RW$ model, the number of random walks passing through a vertex 
is an obvious measure of its importance to the formation process.
A simpler
proxy is the vertex insertion order $S_{1:T}$.
If only the final graph $G_T$ is observed, the order is given by the 
posterior of the bridge $\Gall$.
The median posterior order under the $\RWU$ model is listed below for the NIPS dataset,
for those authors with earliest median appearance, compared to other centrality
measures. Listed are ranks of vertices according to each measure, with numerical values in parentheses.
\\

\makebox[.9\textwidth][c]{
  \resizebox{.8\textwidth}{!}{
    \begin{tabular}{c c l c c c c c}% c c}
      \multicolumn{4}{c}{\emph{Sampling Order}} & \emph{Degree} & 
      \emph{Btwn. Cent.} & \emph{Info. Cent.} & \emph{Katz Cent.}\\
      \midrule
      \qquad & {\input{fig_red_circle.tex}} & (r) & 1 (6) & 14 (4) & 2 (2376) & 1 (0.4478) & 6 (0.2380) \\
      \qquad & {\input{fig_red_circle.tex}} & (l) & 1 (6) & 4 (8) & 3 (2340) & 2 (0.4443) & 2 (0.3740) \\
      \qquad & {\input{fig_blue_circle.tex}} & (r) & 2 (7) & 1 (12) & 1 (2797) & 3 (0.4374) & 3 (0.2571) \\
      \qquad & {\input{fig_blue_circle.tex}} & (l)  & 2 (7) & 8 (6) & 5 (2116) & 5 (0.4178) & 4 (0.2539) \\
      \qquad & {\input{fig_green_circle.tex}} & & 3 (8) & 2 (10) & 11 (728) & 4 (0.4259) & 1 (0.3789) \\
      \qquad & {\input{fig_orange_circle.tex}} & & 4 (10) & 2 (10) & 4 (2196) & 11 (0.3877) & 5 (0.2458) \\
      \qquad & {\input{fig_cyan_circle.tex}} & & 5 (12) & 14 (4) & 6 (1870) & 17 (0.3645) & 35 (0.0667) \\
      \qquad & {\input{fig_pink_circle.tex}} & & 6 (13) & 14 (4) & 16 (173) & 6 (0.4040) & 19 (0.1291) \\
    \end{tabular}
}}

\vspace{.1cm}

Although one would expect vertices appearing earlier to have larger degree,
$S$ seems to correlate more closely with Betweenness Centrality 
and Information Centrality than with vertex degree. The vertices above correspond
to the following vertices in the graph:\\
\makebox[\textwidth][c]{
\begin{tikzpicture}
    \node[scale=.84] at (0,0) {\includegraphics[width=.8\textwidth]{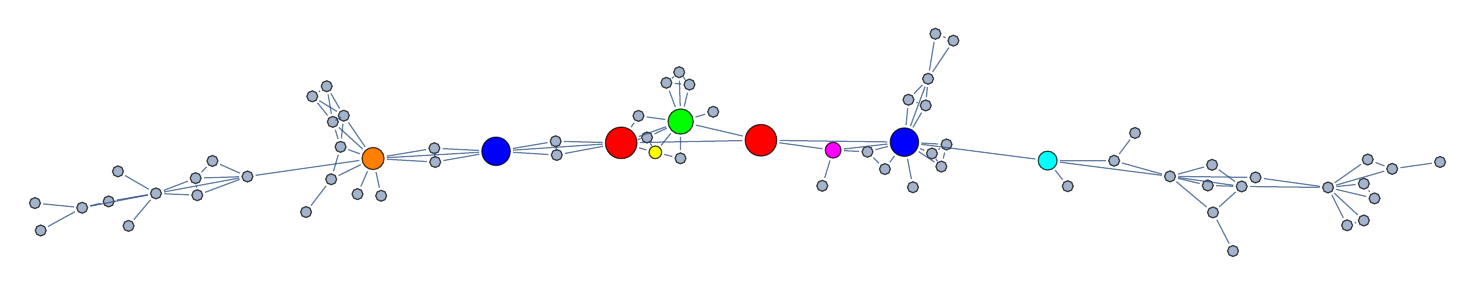}};
\end{tikzpicture}
}

\section{Discussion and open questions} \label{sec:conclusion}

The models introduced here account for dependence
of new links on existing links, but are nonetheless tractable. 
In light of our results, there are two reasons for tractability:
\begin{itemize}
\item The effects of the parameters $\alpha$ and $\lambda$
  are sufficiently distinct that inference is possible---a single
  graph generated by the model carries sufficient signal for parameter
  values to be recovered with high accuracy, as evident in \cref{fig:param:sweep}.
\item There exists a latent variable---the edge insertion 
  order $\Pi$---conditioning on which greatly simplifies the distribution of ${(\alpha,\lambda) \mid G_T}$.
\end{itemize}
Conditioning on history information is instrumental both for theoretical results and for inference.
As noted in the introduction, a qualitatively similar effect---the joint distribution of the
network graph simplifies conditionally on a suitable latent variable---is observed 
in other popular random graph models.

The inference algorithms in \cref{sec:inference} work well for moderately sized graphs (of up to a few hundred
vertices). The dominant source of computational complexity is sampling the edge order---similar 
to inference algorithms for graphon models (which have comparable input size limitations), although the 
latent order plays a very different role here than in an exchangeable model.
Although improvements in efficiency or approximations are certainly possible,
we do not expect the methods to scale to very large graphs. Such graphs are not the subject of 
this work. It is crucial that, under the random walk model, graphs of feasible size do provide
sufficient data for reliable inference. For more complicated models, that might not be the case.

To conclude, we note there are open questions beyond the scope of this paper.
Some of these are obvious, such as 
applications to dynamic network problems; generalizations to
disconnected graphs (e.g.\ using a disconnected seed graph); or
which choices of the distribution $P$ of the random walk length
in \cref{thm:degree:distribution} satisfy condition \eqref{condition:self:loops}.
We also mention two open questions that seem intriguing, but non-trivial:

\emph{1)} 
The generative process (\cref{scheme:multigraph}) starts with an initial ``seed graph'' $G_1$. In our experiments,
$G_1$ simply consists of a single edge, but one might choose any finite graph instead. 
For example, if a graph $g$ of a given, finite size is observed, and interpreted
as the outcome of a random walk model after $t$ steps, one could perform predictive
analysis by considering the random graph ${G_T|G_1=g}$, now with ${t+T-1}$ edges. 
For preferential attachment models, it has been shown that the choice of such
a seed graph $G_1$ can have significant influence on the resulting graph $G_T$ that remains detectable even 
in the limit ${T\rightarrow\infty}$ \citep{Bubeck:Mossel:Racz:2015}. 
The structure of our model seems to suggest a similar property:
The results \cref{sec:asymptotic:degrees} show, roughly speaking,
that properties similar to those of preferential attachment models hold
for the limiting degree sequence if one conditions on the vertex arrival time sequence $S$.
It is not clear, however, whether and how the techniques used by
\citet{Bubeck:Mossel:Racz:2015} are applicable to the random walk model.

\emph{2)} Another open question is whether the random walk itself results in power law behavior:
The results in \cref{sec:properties} stem from size-biased sampling, not the random
walk. The $\RWU$ model on multigraphs, however, does not use such reinforcement, but empirically
exhibits heavy-tailed degrees nonetheless. Intuitively, the probability to reach a vertex
by random walk increases with degree: 
For any $\RW(\alpha,P)$ model, the probability
\eqref{eq:mixed:rw:probability} can be written as
\begin{equation*}
  \P\braces{V_{end}=v|V_0=u,\, G}
  = \Bigl[ {\textstyle\frac{d_v}{\vol(G)}} + {\textstyle\sqrt{\frac{d_v}{d_u}}} \sum_{i=2}^n H_P(1-\LeigVal_i) \LeigVec_i(u) \LeigVec'_i(v) \Bigr] \;,
\end{equation*}
where $(\LeigVec_i)$ are the eigenvectors, and $(\LeigVal_i)$ the eigenvalues, of the Laplacian $\L$. Thus, the probability to terminate at $v$ is a mixture of a preferential attachment term (proportional to $d_v$),
and a term depending on $P$ and on the structure of $G$ at various scales. That
seems to suggest the random walk induces a power law, but at present, we have no proof.

\clearpage

%%%%%%%% appendix %%%%%%%

\appendix

\section{Proofs} \label{sec:proof_appx}

\subsection{Proof of \texorpdfstring{\cref{prop:mixed_random_walks}}{Proposition 3.1}} \label{subsec:proof:mixed_random_walks}

If $K$ has law $\Poisson(\lambda)$,~\eqref{eq:mixed:rw:probability} becomes
\begin{align} \label{eq:poisson:rwprob}
  \RWprob^{\lambda}
   & = \sum_{k=0}^{\infty} \frac{e^{-\lambda} \lambda^k}{k!} (D^{-1}A)^{k+1} = D^{-1/2} \left[ \left( \iden_n - \L \right) e^{-\lambda} \sum_{k=0}^{\infty} \frac{ \lambda^k \left(\iden_n - \L \right)^{k} }{k!}  \right] D^{1/2} \nonumber \\
   & = D^{-1/2} \left[ \left( \iden_n - \L \right) e^{-\lambda \L} \right] D^{1/2} 
    = D^{-1/2} \left[ \left( \iden_n - \L \right) \mathbf{K}^{\lambda} \right] D^{1/2}\, .\nonumber
\end{align}
Since the spectral norm of $\iden - \L$ is 1, series is absolutely convergent. For ${K\sim\NB(r,p)}$,
\begin{equation} \label{eq:nb:kernel}
   \sum_{k=0}^{\infty} \frac{ \Gamma(k + r) }{ \Gamma(k + 1) \Gamma(r) } p^k (1-p)^r (D^{-1}A)^{k+1}
   = D^{-1/2} \left[ \left( \iden_n - \L \right) \left( \iden_n + \frac{p}{1-p} \L \right)^{-r} \right] D^{1/2}
\end{equation}
similarly yields ${\RWprob^{r,p}=D^{-1/2} \left[ \left( \iden_n - \L \right) \mathbf{K}^{r,p} \right] D^{1/2}}$.

\subsection{Proof of \texorpdfstring{\cref{thm:degree:distribution}}{Theorem 3.3}} \label{subsec:proof:degree:distribution}

We begin with a lemma used in the proofs of \cref{thm:degree:distribution} and~\cref{thm:degree:sequence}.
Let $\stub[1][j][t]$ be an indicator variable that is equal to 1 if the first end of the $t$-th edge is attached to vertex $v_j$, and 0 otherwise; likewise define $\stub[2][j][t]$ for the second end of the $t$-th edge.
\begin{lemma} \label{lemma:invariance}
  Suppose a sequence of graphs is generated with law $\RWSB(\alpha,P)$, such that $P$ is a probability distribution on $\N$. Then the size-biased distribution $\SB_t$ is left-invariant under the mixed random walk probability $\RWprob_t$ induced by $P$, that is
  \begin{align} \label{eq:sb:invariant:again}
    \SB_t'\RWprob_t= \SB_t' 
    \qquad\text{ where }
    [\RWprob_t]_{uv}:=\P\lbrace V_{end}=v\on V_0=u,G_t\rbrace
  \end{align}
  holds for every ${t\in\mathbb{N_+}}$. Furthermore, for each $v_j \in \vertexset(G_t)$, the following holds:
  \begin{align}
    \E[ \stub[1][j][t+1] \on G_t ] = \frac{1}{1-\alpha}\E[ \stub[2][j][t+1] \on G_t ] = \frac{ \deg_t(v_j) }{ 2t } \;.
  \end{align}
\end{lemma}
\begin{proof}
  Substituting \eqref{eq:mixed:rw:probability} into ${\SB_t'\RWprob_t}$ yields
  \begin{align}
    \SB_t'\RWprob_t 
      & = \sum_{k} \SB_t' (D^{-1}A)^{k} P\lbrace K=k\rbrace
         = \sum_{k} \SB_t' P\lbrace K=k\rbrace = \SB_t'\; .
  \end{align}
  The $\RWSB$ model samples the first end of each edge from the size-biased distribution:
  \begin{align}
    \E[ \stub[1][j][t+1] \on G_t ] = \P[ V_{t+1} = v_j \on G_t ] = \frac{ \deg_t(v_j) }{ 2t }\; .
  \end{align}
  For the second end, denote by $\indicator_{j,t}$ the indicator vector for vertex $v_j$. Then
  \begin{equation*}
    \E[ \stub[2][j][t+1] \on  G_t ] 
      = \sum_{u} \P[V_{t+1}=u \on G_t] \bigl[ \RWprob_t \bigr]_{uv}
      = (1-\alpha)\SB'_t \RWprob_t \indicator_{j,t}
      = (1-\alpha)\frac{ \deg_t(v_j) }{ 2t }\; .
  \end{equation*}

\end{proof}

In the setup of \cref{thm:degree:distribution}, asymptotic expected degree counts are as follows:
\begin{lemma} \label{lemma:deg:dist:conv:expectation}
  Let a sequence of multigraphs ${(G_1,G_2,\dots)}$ have law $\RWSB(\alpha,P)$, for a probability distribution 
  $P$ on $\N$ satisfying \cref{condition:self:loops}. Then 
  \begin{align} \label{eq:deg:dist:conv:expecation}
    \frac{ \E[m_{d,t}] }{ \alpha t } \almostSurely p_d 
    \qquad\text{ for each } d\in\mathbb{N} \text{ as } t\to\infty \;,
  \end{align}
  where $p_d$ is given in~\eqref{eq:deg:distn}.
\end{lemma}
\begin{proof}
  \cref{condition:self:loops} yields the recurrence
  \begin{equation*} \label{eq:deg:recurrence}
    \E (m_{d,t+1}) = \alpha\cdot\delta_1(d) + \E(m_{d,t}) \Bigl( 1 - \frac{(2 - \alpha) d}{2t} \Bigr) + \E(m_{d-1,t})\frac{(2-\alpha)(d-1)}{2t} + o(1) \;,
  \end{equation*}
  with $m_{0,t}=0$ for all $t$. This can be written more generally as
  \begin{align} \label{eq:general:recurrence}
    M(d,t+1) = (1 - b(t)/t)M(d,t) + g(t) \;.
  \end{align}
  If ${b(t)\to b}$ and ${g(t)\to g}$, then ${M(d,t)/t \to g/(b+1)}$, by
  \citep[][Lemmas 4.1.1, 4.1.2]{Durrett:2006}.
  For $d=1$, we have ${b(t) = b = (2-\alpha)/2}$ and
  ${g(t) = g = \alpha}$, so ${\E(m_{1,t})/t \to \frac{2\alpha}{4-\alpha}=\alpha p_1}$. 
  Proceeding by induction, ${b = (2-\alpha)d/2}$ and
  ${g = \alpha p_{d-1} (2-\alpha)(d-1)/2}$ yield
  \begin{equation*}
    \frac{\E[m_{d,t}]}{\alpha t} \quad\to\quad
      p_{d-1}{\textstyle \frac{(2-\alpha)(d-1)}{2 + (2-\alpha)d}}
       = \frac{2}{4-\alpha} \prod_{j=1}^{d-1} \frac{j}{\frac{2}{2-\alpha} + j + 1}
      = \frac{2}{2-\alpha}
    \frac{\Gamma(d) \Gamma(2 + \frac{\alpha}{2-\alpha})}{\Gamma(d + 2 + \frac{\alpha}{2-\alpha})} \;,
  \end{equation*}
  for $d>1$, which is just $p_d$ as defined in \eqref{eq:deg:distn}.
\end{proof}

The following result, the proof of which can be found in \citet[][Sec. 3.6]{Chung:Lu:2006} %[Chung, F. and L.~Lu,
% {\em Complex Graphs and Networks}, AMS 2006, Section 3.6], 
shows that the random variable $m_{d,t}$ concentrates about its mean:
\begin{lemma} \label{lem:concentration}
  Let a sequence of multigraphs ${(G_1,G_2,\dots)}$ have law $\RWSB(\alpha,P)$, for a probability distribution $P$ on $\N$ satisfying \cref{condition:self:loops}. Then 
  \begin{align}
    \P\Big[ \; \Big\lvert \frac{m_{d,t}}{\alpha t} - p_d \; \Big\rvert \leq 2\sqrt{\frac{ d^3 \log t }{\alpha^2 t}} \; \Big] \geq 1 - 2(t+1)^{d-1} t^{-d} = 1 - o(1) \;.
  \end{align}
\end{lemma}

\begin{proof}[Proof of \cref{thm:degree:distribution}]
  \cref{lem:concentration} shows that $\frac{m_{d,t}}{\alpha t}$ concentrates around its mean, and with \cref{lemma:deg:dist:conv:expectation} showing that the mean is $p_d$, the proof of the theorem is complete.
\end{proof}

\subsection{Proof of \texorpdfstring{\cref{thm:degree:sequence}}{Theorem 3.4}} \label{subsec:proof:degree:sequence}

By definition of $\RW$ models (see~\cref{scheme:multigraph}), at each step $t$ a new vertex appears with probability $\alpha$, which is represented as a collection of $\Bernoulli(\alpha)$ random variables $(B_t)_{t\geq 2}$. Alternatively, there is the collection of steps ${S_1,S_2,\dots,S_j,\dots}$ at which new vertices appear, such that $S_j$ is the step at which the $j$-th vertex appears. By the fundamental relationship between Bernoulli and Geometric random variables, the sequence ${S_1,S_2,\dots,S_j,\dots}$ can be sampled independently of the graph sequence as $S_1=S_2=1$ and
\begin{align}
  S_j=S_{j-1} + \Delta_j \mbox{, where } \Delta_j\simiid\Geom(\alpha) \mbox{, for } j>2 .
\end{align}

In what follows, we condition on the sequence $(S_j)_{j\geq 1}$ unless explicitly stated otherwise. We begin by calculating the expected degree of the $j$-th vertex.

\begin{lemma} \label{lemma:degrees:fixed:vertices}
  Let a sequence of multigraphs ${(G_1,G_2,\dots)}$ have law $\RWSB(\alpha,P)$. Let $d_j(t):=\deg_t(v_j)$ be the degree of $v_j$ in graph $G_t$, where $v_j$ is the $j$-th vertex to appear in the graph sequence, and let $\rho = 1 + \frac{\alpha}{2-\alpha}$. 
  Then conditional on $S_j$, $d_j(t) t^{-1/\rho}$ converges almost surely to a random variable $\xi_j$ as $t \to \infty$, and
  \begin{align} \label{eq:expected:degree}
    \E[d_j(t) \on S_j] = \frac{ \Gamma(S_j) \Gamma(t + \frac{1}{\rho}) }{ \Gamma(S_j + \frac{1}{\rho}) \Gamma(t) }
  \end{align}
\end{lemma}
\begin{proof}
  Let $\sigalg[t]$ denote the $\sigma$-algebra generated after $t$ steps. Then for $t\geq S_j$,
  \begin{equation} \label{eq:cond:prob}
    \E[d_j(t+1) \on  \sigalg[t]]
      = d_j(t) + \E[\stub[1][j][t+1] \on  \sigalg[t]] + \E[\stub[2][j][t+1] \on  \sigalg[t]] 
      = d_j(t) (1 + {\textstyle\frac{1}{\rho t}}) \; ,
  \end{equation}
  where the final identity follows from~\cref{lemma:invariance}, and ${d_j(S_j)=1}$. The sequence
  \begin{equation} \label{eq:degree:martingale}
    M_j(t) 
      := d_j(t)\frac{ \Gamma(S_j + \frac{1}{\rho}) }{ \Gamma(S_j) }\cdot\frac{ \Gamma(t) }{ \Gamma(t + \frac{1}{\rho}) }
    \qquad\text{ for } t\geq S_j
  \end{equation}
  is a non-negative martingale with mean 1, by~\eqref{eq:cond:prob}. Therefore, $M_j(t)$ converges almost surely to a random variable $M_j$ as ${t\rightarrow\infty}$. Taking expectations on both sides and rearranging~\eqref{eq:degree:martingale} yields~\eqref{eq:expected:degree}. With Stirling's formula,
  \begin{align} \label{eq:degree:convergence}
    \frac{ d_j(t) }{ t^{1/\rho} }
      & \almostSurely M_j \frac{ \Gamma(S_j) }{ \Gamma(S_j + \frac{1}{\rho}) } %\nonumber \\
      := \xi_j\; .
  \end{align}
\end{proof}

We now consider the joint distribution of the limiting random variables $(\xi_j)_{j\geq 1}$. The approach, which is adapted from \cite{Mori:2005} (see also~\cite{Durrett:2006}), is to analyze a martingale that yields the moments of $(\xi_j)_{j\geq 1}$. First, define for $t\geq S_j$,
\begin{align} \label{eq:martingale:ratio}
  R_{j,k}(t) := \frac{ \Gamma(d_j(t) + k) }{ \Gamma(d_j(t)) \Gamma(k+1) } \; , 
\end{align}
At a high level, $R_{j,k}(t) \approx d_j(t)^k/k!$ for large $t$, so \eqref{eq:degree:convergence} shows that properly scaled $R_{j,k}(t)$ should converge to $\xi^k_j/k!$ for each $j$. We make this precise in what follows.

Let ${\jvec:=(j_i)}$ be an ordered collection of vertices such that ${1 \leq j_1 < j_2 < \dots < j_r}$, 
and let ${\kvec:=(k_i)}$ be a corresponding vector of moments. To define a suitable martingale, 
we abbreviate 
\begin{equation*}
  \mu_{\kvec}:=\sum_{i=1}^r k_i
  \quad\text{ and }\quad
  \Sigma_{\jvec,\kvec}(t):= \sum_{i,l=1}^r \left(\frac{1}{2}\right)^{\indicator_{i=l}} \frac{ k_i(k_{l} - \indicator_{i=l}) }{ d_{j_{l}}(t) + \indicator_{i=l} } [\RWprob_t]_{j_i j_{l}} \;,
\end{equation*}
and note $[\RWprob_t]_{j j'} = 0$ for all $t < \max\{ S_j, S_{j'} \}$.
Define
\begin{align} \label{eq:martingale:constants}
  c_{\jvec,\kvec}(t)
  =\Gamma(t) \Bigl[ \prod_{s=0}^{t-1}\bigl( s + {\textstyle\frac{\mu_{\kvec}}{\rho}} 
  + {\textstyle\frac{1-\alpha}{2}}\Sigma_{\jvec,\kvec}(s)  \bigr) \Bigr]^{-1} \;,
\end{align}
which is a random variable since $\RWprob_t$ is random, and is $\sigalg[t]$-measurable for each $t$. Since ${[\RWprob_t]_{j_i j_{l}} < 1}$, and since $d_j(t)\almostSurely\infty$ as $t\to\infty$ by~\cref{lemma:degrees:fixed:vertices}, 
${\Sigma_{\jvec,\kvec}(t)\almostSurely 0}$ as $t\to\infty$, which yields
\begin{equation} \label{eq:martingale:constants:asymptotics}
  c_{\jvec,\kvec}(t)
    = \frac{ \Gamma(t) }{ \Gamma(t + \frac{1}{\rho}\mu_{\kvec}) }(1 + o(1))
    = t^{-\mu_{\kvec}/\rho}(1 + o(1)) \quad \text{as } t\to\infty \;.
\end{equation}
Note that
\begin{align} \label{eq:mlimit}
  \lim_{t \to\infty} c_{\jvec,\kvec}(t) \prod_{i=1}^r R_{j_i,k_i} = \prod_{i=1}^r \frac{ \xi^{k_i}_{j_i} }{ \Gamma(k_i + 1) } \;,
\end{align}
if the limit exists.
Existence is based on the following result:
\begin{lemma} \label{prop:martingale}
  Let $\sigalg[t]$ denote the $\sigma$-algebra generated by a $\RWSB(\alpha,P)$ multigraph sequence up to step $t$. Let $r>0$ and ${1 \leq j_1 < j_2 < \dots < j_r}$ be integers, and real-valued ${k_1,k_2,\dots,k_r > -1}$. Then with $R_{j,k}(t)$ defined in~\eqref{eq:martingale:ratio} and ${c_{\jvec,\kvec}(t)}$ defined in~\eqref{eq:martingale:constants}, 
  \begin{align} \label{eq:martingale:full}
    Z_{\jvec,\kvec}(t) := c_{\jvec,\kvec}(t) \prod_{i=1}^r R_{j_i,k_i}(t)
  \end{align}
  is a nonnegative martingale for $t\geq\max\{S_{j_r},1\}$. If ${k_1,k_2,\dots,k_r > -\frac{1}{2}}$, then $Z_{\jvec,\kvec}(t)$ converges in $L_2$.
\end{lemma}
\begin{proof}
  It can be shown that
  \begin{align} \label{eq:martingale:expectation}
    \E\Bigl[ \prod_{i=1}^r R_{j_i,k_i}(t+1) \;\Big|\; \sigalg[t] \Bigr]
      = \Bigl(\prod_{i=1}^r R_{j_i,k_i}(t) \Bigr)
        \Bigl( 1 + \frac{\mu_{\kvec}}{\rho t}  + \frac{1-\alpha}{2t}\Sigma_{\jvec,\kvec}(t)\Bigr) 
        \; ,
  \end{align}
  from which it follows that
  \begin{equation} \label{eq:martingale:proof}
    \E[ Z_{\jvec,\kvec}(t+1) \on  \sigalg[t] ]
    = c_{\jvec,\kvec}(t+1) R_{\jvec,\kvec}(t)
    \Bigl( 1 + \frac{\mu_{\kvec}}{\rho t}  + \frac{1-\alpha}{2t}\Sigma_{\jvec,\kvec}(t)\Bigr) 
    = c_{\jvec,\kvec}(t) R_{\jvec,\kvec}(t) = Z_{\jvec,\kvec}(t)\;.
  \end{equation}
  Furthermore, ${Z_{\jvec,\kvec}(\max\{S_{j_r},1\})>0}$. By~\eqref{eq:martingale:constants:asymptotics} and properties of the gamma function,
  \begin{align}
    Z_{\jvec,\kvec}(t)^2 \leq Z_{\jvec,2\kvec}(t) \prod_{i=1}^r \binom{2k_i}{k_i} \;.
  \end{align}
  By~\eqref{eq:martingale:proof}, $Z_{\jvec,2\kvec}(t)$ is a martingale with finite expectation for ${2k_1,\dots,2k_r>-1}$. Therefore, $Z_{\jvec,\kvec}(t)$ is an $L_2$-bounded martingale and hence converges in $L_2$ (and also in $L_1$) for ${k_1,\dots,k_r>-\frac{1}{2}}$.
\end{proof}

Combining the auxiliary results above, we can now give proof of the result.

\begin{proof}[Proof of \cref{thm:degree:sequence}]
The limit of $Z_{\jvec,\kvec}(t)$ is \eqref{eq:mlimit}, which enables calculation of the moments of $\xi_j$. In particular, for any vertex $v_j$, $j\in\N_+$, $k\in\R$ such that $k > -\frac{1}{2}$,
\begin{align} \label{eq:deg:seq:moments}
 \E\Bigl[{\textstyle\frac{ \xi^{k}_{j} }{ \Gamma(k + 1) }}\, \;\Big|\; S_j \Bigr]
   = \lim_{t\to\infty} \E[Z_{j,k}(t) \on  S_j]
   = \E[Z_{j,k}(S_j) \on S_j]
   = {\textstyle \frac{ \Gamma(S_j) }{ \Gamma(S_j + \frac{k}{\rho}) }} \;.
\end{align}
Although the joint moments involve $\Sigma_{\jvec,\kvec}$, the moments \eqref{eq:deg:seq:moments} characterize the marginal distribution of $\xi_j \on S_j$ via the Laplace transform. Since
\begin{align} \label{eq:moment:char}
  \E[\xi^k_j \on S_j] 
    & = \frac{ \Gamma(S_j) \Gamma(k+1) }{ \Gamma\bigl(S_j + \frac{k}{\rho} \bigr) } 
    = \frac{ \Gamma\left(\rho(S_j-1) + 1 + k\right) \Gamma(S_j) }{ \Gamma\left(\rho(S_j-1) + 1 \right) \Gamma\bigl(S_j + \frac{k}{\rho} \bigr) }  \frac{ \Gamma\left(\rho(S_j-1) + 1 \right) \Gamma(k+1) }{ \Gamma\left(\rho(S_j-1) + 1 + k \right) } \;,\nonumber
\end{align}
${\E[\xi^k_j \on S_j]=\E[ M_j^{k} B^k_{j} \on S_j]}$, where
${M_j}$ is a generalized ${\MittagLeffler(\rho^{-1},S_j-1)}$ variable \citep{James:2015aa}, and $B_j$ is ${B_j\sim\BetaDist(1,\rho(S_j-1))}$.
It follows that ${M_j\perp\!\!\perp_{S_j} B_j}$. The second equality in distribution is verified by checking moments.

It remains to show~\eqref{eq:marginal:moments}. We begin by noting that the left-hand side of~\eqref{eq:deg:seq:moments} is an expectation that conditions on $S_j$. Define ${\tilde{S}_j:=S_j - 1 - (j-2)}$, which is marginally distributed as $\NBo(j-2,1-\alpha)$. Then
\begin{align}
  \E\bigl[ \E[ \xi_j^k \on S_j ] \bigr]
  & = \E\Bigl[  \frac{ \Gamma(\tilde{S}_j + 1 + (j-2)) \Gamma(k+1) }{ \Gamma(\tilde{S}_j + 1 + (j-2) + \frac{k}{\rho}) } \Bigr] \nonumber \\
    & = \sum_{t=0}^{\infty}
      \frac{ \Gamma(t + 1 + (j-2)) \Gamma(k+1) }{ \Gamma(t + 1 + (j-2) + \frac{k}{\rho}) } 
      \frac{ \Gamma(t + j - 2) }{ \Gamma(j-2)\Gamma(t+1) }(1-\alpha)^t \alpha^{j-2} \\ \nonumber
    & = \frac{ \Gamma(k+1) \Gamma(j-1) }{ \Gamma(j-1 + \frac{k}{\rho}) }
      \alpha^{\frac{k}{\rho}}
      \,_2F_1(1 + \frac{k}{\rho},\frac{k}{\rho}; j-1+\frac{k}{\rho};1-\alpha)\; ,
\end{align}
where $\,_2F_1(a,b;c;z)$ is the ordinary hypergeometric function. For $j\to\infty$,
\begin{equation*}
  \lim_{j\to\infty} \E\bigl[ \E[ \xi_j^k \on  S_j ] \bigr]
    = {\textstyle\lim_{j\to\infty} \frac{ \Gamma(k+1) \Gamma(j-1) }{ \Gamma(j-1 + \frac{k}{\rho}) }}
      \alpha^{\frac{k}{\rho}}
      (1 + O(j^{-1})) \nonumber \\
    = \Gamma(k+1)\alpha^{\frac{k}{\rho}} j^{-\frac{k}{\rho}}(1 + O(j^{-1})) \;.
\end{equation*}
follows using the series expansion of $\,_2F_1(a,b;c;z)$.
\end{proof}

\subsection{Proof of \texorpdfstring{\cref{prop:limiting_ACL}}{Proposition 3.5}} \label{subsec:proof:limiting_ACL}

\begin{proof}[Proof of \cref{prop:limiting_ACL}]
Existence of the limit as $\lambda\to\infty$ follows from the existence of the limiting (stationary) distribution for each $t\in\N_+$. For equivalence, it suffices to show that given any connected graph $G$, the conditional distribution over graphs $G'$ is the same for ${\RWSB(\alpha,\infty)}$ and $\ACL(\alpha)$.
The probability of attaching a new vertex to an existing vertex $v$ in both models is 
${\alpha\deg(v)/\vol(G)}$. With probability ${1-\alpha}$, a new edge is inserted between existing vertices $u$ and $v$, and so it remains to show that in this case the distribution over pairs of vertices, $\eDist(u,v)$, is the same. 
We have 
\begin{equation*}
  \eDist(u,v) = 2\frac{\deg(u)}{\vol(G)}\frac{\deg(v)}{\vol(G)}
  \quad\text{ and }\quad
  \eDist(u,v) = \frac{\deg(u)}{\vol(G)}[\RWprob^{\lambda}]_{uv} + \frac{\deg(v)}{\vol(G)}[\RWprob^{\lambda}]_{vu}
\end{equation*}
for the ACL and $\RWSB$ model respectively, with $\RWprob^{\lambda}$. First, consider 
${(\iden_n - \L)\mathbf{K}^{\lambda}}$ in terms of the spectrum of $\L$. Let 
${0=\LeigVal_1 \leq \dots \leq \LeigVal_n\leq 2}$ be the eigenvalues for a graph on $n$ vertices, 
with eigenvectors $\LeigVec_i$. Then
\begin{align}
(\iden_n - \L)\mathbf{K}^{\lambda} = (\iden_n - \L)e^{-\lambda\L} = \sum_{i=1}^{n} (1 - \LeigVal_i)e^{-\lambda \LeigVal_i} \LeigVec_i \LeigVec_i' \;.
\end{align}
When $\lambda \to \infty$, only the eigenvector $\LeigVec_1$ corresponding to the eigenvalue $\LeigVal_1=0$ contributes to the random walk probabilities, \ie ${\lim_{\lambda\to\infty} (\iden_n - \L)e^{-\lambda\L} = \LeigVec_1 \LeigVec_1'}$. 
The limit satisfies $\LeigVec_1 \propto D^{1/2}\indicator$, where $\indicator$ is the vector of all ones \citep[][Ch. 2]{Chung:1997}. Therefore, 
% [see Chung, F., {\em Spectral Graph Theory}, AMS 1997, Chapter 2]. Therefore,
\begin{equation*}
  \lim_{\lambda\to\infty} [\RWprob^{\lambda}]_{uv} 
  = \bigl[ \indicator \indicator' D{\textstyle\frac{1}{\vol(G)}} \bigr]_{uv} = \frac{\deg(v)}{\vol(G)} \;,
\end{equation*}
and the result follows.
\end{proof}

\subsection{Proofs for \texorpdfstring{\cref{sec:inference}}{Section 4}}

\begin{proof}[Proof of \cref{prop:smc}]
  Proposition 2.2 of \citet{delMoral:Murray:2015} shows that
  \begin{equation*}
    \E[ f(\Gall) \on G_1, G_T ] = L_{\seqParams}^1(G_1)^{-1}\E\bigl[ f(\Gall) L_{\seqParams}^{T-1}(G_{T-1}) 
      {\textstyle\frac{h_t(G_t)}{h_t(G_{T-1})} \prod_{s=2}^{T-1} \frac{h_s(G_s)}{h_{s-1}(G_{s-1})}} \on G_1 \bigr] \;,
  \end{equation*}
  for any approximation function $h_t$ that is positive outside a null set. For our choice of $h_t$, 
  the proposal density \eqref{eq:smc:bridge:proposal} only places probability mass on graphs with 
  non-zero bridge likelihood, making $\indicator\{h_t=0\}$ a probability zero event.
  For an SMC approximation to be consistent, the target density must be absolutely continuous with respect to the proposal density at each step $t$ \citep[see, e.g.][]{Robert:Casella:2004,Doucet:Johansen:2011}. In other words, the proposal density must be a valid importance sampling distribution at each step. The target, $\mathcal{L}_{\seqParams}(G_{1:t} \on G_T, G_1)$ is absolutely continuous with respect to $\prod_{s=2}^t r_{\seqParams}(G_s \on G_{s-1})$ from \eqref{eq:smc:bridge:proposal} by construction if properties (P1) and (P2) hold, and the result follows.
\end{proof}

\Cref{prop:unbiased:bridge:likelihood} is a special case of the following:
\begin{proposition}
  Let ${q_{\seqParams}^t}$, for ${t=1,\ldots,T}$, be the Markov kernels defining a sequential network model that satisfies conditions (P1) and (P2), and let ${r_{\seqParams}^t}$ be the corresponding proposal kernels. Furthermore, let $h_t(G_T \mid G_t)$ for ${t = 1,\ldots,T-2}$ be a sequence of fixed functions that are strictly positive if $G_t \subseteq G_T$. Given an observation $G_T$ and a fixed $G_1$, define the weights
  \begin{align}
    \tilde{w}^i_t := 
      \begin{cases}
        1, & \text{for } \quad t = 1 \\
        \displaystyle \frac{1}{ h_t(G_T \mid G^i_{t-1}) } \frac{ q_{\seqParams}^t(G^i_{t} \mid G^i_{t-1}) }{ r_{\seqParams}^t(G^i_{t} \mid G^i_{t-1})  }, & \text{for } \quad 2 \leq t < T - 1 \\
        \displaystyle\frac{q_{\seqParams}^t(G_T \mid G_{t})}{ h_t(G_T \mid G^i_{t-1}) } \frac{ q_{\seqParams}^t(G^i_{t} \mid G^i_{t-1}) }{ r_{\seqParams}^t(G^i_{t} \mid G^i_{t-1})  }, & \text{for } \quad t = T-1
      \end{cases}
      \;,
  \end{align}
  and $w^i_t$ the corresponding weights normalized across the $N$ particles. Define the estimator
  \begin{align}
  \label{eq:SMC:marginal:general}
  \hat{L}_{\seqParams}^1 
    & := \prod_{t=2}^{T-1} \left[ \left( \sum_{i=1}^N \frac{\tilde{w}^i_t}{N} \right)
      \left( \frac{ \sum_{i=1}^N h_{t-1}(G_T \mid G^i_{t-1}) \tilde{w}^i_{t-1} }{ \sum_{i=1}^N \tilde{w}^i_{t-1} }  \right)   \right] \\
    & := \prod_{t=2}^{T-1} \hat{L}_{\seqParams}(G_t \mid G_{t-1}) \nonumber \;.
  \end{align}
  Then $\hat{L}_{\seqParams}^1$ is positive and unbiased: $\hat{L}^1_{\seqParams}>0$ and ${\E[\hat{L}^1_{\seqParams}] = L_{\seqParams}^1(G_1) = p_{\seqParams}(G_T \mid G_1)}$, for any $N\geq 1$.
\end{proposition}
\begin{proof}
  Unbiasedness will be established by iterating expectations, following the approach in \citet{Pitt:etal:2010}. Let $\mathcal{S}_{t}$ be the set of particles and weights $\{G^i_t;\tilde{w}^i_t\}$ at step $t$. Then we have that
  \begin{align*}
    & \conditional[]{\left( \sum_{i=1}^N \frac{\tilde{w}^i_{T-1}}{N} \right) }{\mathcal{S}_{T-2}} \\
    % & = \conditional[]{\left( \sum_{i=1}^N \frac{\tilde{w}^i_{T-1}}{N} \right)}{\mathcal{S}_{T-2}} \\%\times \left( \sum_{i=1}^N \frac{ h_{T-2}(G_T \mid G^i_{T-2}) \tilde{w}^i_{T-2} }{ \tilde{w}^i_{T-2} }  \right) \\
    & = \sum_{i=1}^N  \displaystyle\int \frac{q_{\seqParams}^t(G_T \mid G_{T-1})}{ h_{T-2}(G_T \mid G^i_{T-2}) } \frac{ q_{\seqParams}^t(G_{T-1} \mid G^i_{T-2}) }{ r_{\seqParams}^t(G_{T-1} \mid G^i_{T-2})  } \frac{ r_{\seqParams}^t(G_{T-1} \mid G^i_{T-2})  h_{T-2}(G_T \mid G^i_{T-2}) \tilde{w}^i_{T-2}}{ \sum_{j=1}^N h_{T-2}(G_T \mid G^j_{T-2}) \tilde{w}^j_{T-2}} dG_{T-1} \\
    & = \sum_{i=1}^N \frac{ p_{\seqParams}^{T,T-2}(G_T \mid G^i_{T-2}) \tilde{w}^i_{T-2} }{ \sum_{j=1}^N h_{T-2}(G_T \mid G^j_{T-2}) \tilde{w}^j_{T-2} } \;.
  \end{align*}
  Therefore, ${\E[ \hat{L}_{\seqParams}( G_{T-1} \mid G_{T-2}) \mid \mathcal{S}_{T-2}]
          = \sum_{i=1}^N  p_{\seqParams}^{T,T-2}(G_T \mid G^i_{T-2}) w^i_{T-2}}$. 
  Likewise, it is straightforward to show that
  \begin{align*}
    & \E[ \hat{L}_{\seqParams}( G_{T-1} \mid G_{T-2}) \hat{L}_{\seqParams}( G_{T-2} \mid G_{T-3}) \mid \mathcal{S}_{T-3}] \\
      & = \E[ \; \E[ \hat{L}_{\seqParams}( G_{T-1} \mid G_{T-2}) \mid \mathcal{S}_{T-2} ] \; \hat{L}_{\seqParams}( G_{T-2} \mid G_{T-3}) \mid \mathcal{S}_{T-3} ] 
       = \sum_{i=1}^N p_{\seqParams}^{T,T-3}(G_T \mid G^i_{T-3}) w^i_{T-3} \;.
  \end{align*}
  Iterating for $t = T-4,\ldots,1$, we have 
  \begin{align*}
    & \conditional[]{\prod_{t=2}^{T-1} \left\{ \left( \sum_{i=1}^N \frac{\tilde{w}^i_t}{N} \right) \left( \frac{ \sum_{i=1}^N h_{t-1}(G_T \mid G^i_{t-1}) \tilde{w}^i_{t-1} }{ \sum_{i=1}^N \tilde{w}^i_{t-1} }  \right)   \right\} }{\mathcal{S}_1} \\
    & = \sum_{i=1}^N p^{T,1}_{\seqParams}(G_T \mid G^i_1) w^i_1 = p_{\seqParams}(G_T \mid G_1) \;,
  \end{align*}
  which proves the claim.
\end{proof}

\begin{proof}[Proof of \cref{prop:pmmh}]
  Theorem 4 of \citet{Andrieu:Doucet:Holenstein:2010}, of which \cref{prop:pmmh} is a special case, requires that:
  (a) The resampling scheme in the SMC algorithm is unbiased, \ie each particle is resampled with probability proportional to its weight. (b) $\mathcal{L}_{\seqParams}(G_{1:t} \on G_T, G_1)$ is absolutely continuous with respect to $\prod_{s=2}^t r_{\seqParams}(G_s \on G_{s-1})$ for any $\seqParams$.
  Condition (a) is satisfied by the multinomial, residual, and stratified resampling schemes. As discussed in the proof of \cref{prop:smc}, the condition (b) holds by construction if properties (P1) and (P2) hold.
\end{proof}

\section{Particle Gibbs updates} \label{sec:inference_appx}

The particle Gibbs sampler updates are performed in three blocks: $(\alpha,\lambda)$,
$(\Iall,\kAll)$, and $\Glatent$. As we describe in the following subsections, we make use of conditional conjugacy and tools from spectral graph theory to increase sampling efficiency. We discuss each block of updates in turn. \\

{\noindent\bf Updating the parameters $\alpha$ and $\lambda$}.
The model is specified in terms of distributions for the latent variables, and placing conjugate priors on their parameters we have,
\begin{flalign} \label{eq:latentvar_prior}
& B_t \on \alpha \simiid \Bernoulli(\alpha), \quad \alpha \sim \BetaDist(a_{\alpha},b_{\alpha}) \\
& K_t \on \lambda \simiid \Poisson(\lambda), \quad \lambda \sim \GammaDist(a_{\lambda},b_{\lambda}).
\end{flalign}
The Gibbs updates are hence ${\alpha \on \Iall \sim \BetaDist\left(\chi, \omega \right)}$ and
${\lambda \on \kAll \sim \GammaDist\left(\kappa, \tau \right)}$, with
\begin{flalign}
& \chi := a_{\alpha} + \sum_{t=2}^T B_t, \ \quad \omega := b_{\alpha} + (T-1) - \sum_{t=2}^T B_t \label{eq:post_params_I} \\
& \kappa := a_{\lambda} + \sum_{t=2}^T K_t, \quad \tau := b_{\lambda} + (T-1) \;. \label{eq:post_params_K}
\end{flalign}

{\noindent\bf Updating the latent variables $\Iall$ and $\kAll$}.
The sequential nature of the model allows each of the latent variables $B_t$ and $K_t$ to be updated individually. 
We use conditional independence to marginalize out sampling steps where possible, since such marginalization
increase the sampling efficiency.
For both $\Iall$ and $\kAll$, we marginalize twice: The first transforms the conditional distributions 
$P(B_t \on \alpha)$ and $P(K_t \on \lambda)$ into the predictive distributions $P(B_t \on \INot)$, and $P(K_t \on \kNot)$, which yields 
\begin{equation*}
  B_t \on \INot \sim \Bernoulli\Bigl(\frac{\chi_{-t}}{\chi_{-t} + \omega_{-t}}\Bigr)
  \quad\text{ and }\quad
  K_t \on \kNot \sim \NB\Bigl(\kappa_{-t},\frac{1}{1 + \tau_{-t}}\Bigr)\;.
\end{equation*}
The second improvement is made by marginalizing the conditional dependence of $B_t$ on $K_t$, 
and of $K_t$ on $B_t$. That requires some notation: $\delta_t (V_t,U_t)$ indicates that a new edge is added to the graph between vertices $V_t$ and $U_t$. When a new vertex is attached to $V_t$, we write $\delta_t (V_t,u^*)$. 
We abbreviate ${\bar{\tau} := \frac{1}{1 + \tau}}$, and write $\nball(v)$ for the $\{0,1\}$-ball of vertex $v$, \ie $v$ and its neighbors. The updates for $\Iall$ are
\begin{align} 
& P(B_t=1 \on G_{(t-1) : t}, \kNot, \INot) \propto \frac{\delta_t(V_t,u^*) \vDist_{t-1}(V_t) \chi_{-t}}{\chi_{-t} + \omega_{-t}}  \nonumber\\
& P(B_t=0 \on G_{(t-1) : t}, \kNot, \INot) \propto 
	\delta(V_t,U_t) + \frac{\delta(V_t,u^*)\omega_{-t}\vDist_{t-1}(V_t)}{\chi_{-t} + \omega_{-t}}  \sum_{u \in \nball(V_t)} \bigl[\RWprob^{\kappa_{-t}, \bar{\tau}_{-t}}_{t-1} \bigr]_{V_t,u} \nonumber
\end{align}
with $\RWprob_{t-1}^{\kappa_{-t}, \bar{\tau}_{-t}}$ as in~\eqref{eq:nb:kernel} with $r=\kappa_{-t}$, $p=\bar{\tau}_{-t}$. The updates for $\kAll$ are
\begin{align} \label{eq:update_k}
& P(K_t=k \on G_{(t-1) : t}, \kNot, \INot) \propto 
  \frac{\Gamma(\kappa_{-t} + k)}
       {\Gamma(k+1)\Gamma(\kappa_{-t})} 
       (1 - \bar{\tau}_{-t})^{\kappa_{-t}}\bar{\tau}_{-t}^k \ \dots \nonumber \\
       & ~ \times\ 
       \Bigl[\ \delta_t(V_t, u^*) \vDist_{t-1}(V_t) \Bigl(\frac{\chi_{-t}}{\chi_{-t} + \omega_{-t}} + \frac{\omega_{-t}}{\chi_{-t} + \omega_{-t}} \sum_{u \in \nball(V_t)} \bigl[\RWprob^{k+1}_{t-1}\bigr]_{V_t, u} \Bigr) \dots  \nonumber \\
& ~~  +\ \delta_t(V_t, U_t)\frac{\omega_{-t}}{\chi_{-t} + \omega_{-t}} \Bigl( \vDist_{t-1}(V_t) \bigl[\RWprob^{k+1}_{t-1}\bigr]_{V_t, U_t} + \vDist_{t-1}(U_t) \bigl[\RWprob^{k+1}_{t-1}\bigr]_{U_t, V_t} \Bigr) \Bigr]\nonumber
\end{align}
with $\RWprob^{k+1}_{t-1}=\DegMat_{t-1}^{-1/2}(\iden_{t-1} - \L_{t-1})^{k+1} \DegMat_{t-1}^{1/2}$, the probability of a random walk of length $k+1$ from $u$ to $v$. For implementation, the distribution for $K_t$ must be truncated at some finite $k$,
which can safely be done at three or four times the diameter of $G_T$: The total remaining probability mass can be calculated analytically, and the mass is placed on larger $k$ is negligible.\\

{\noindent\bf Sampling $\Glatent$}.
Implementation of \cref{alg:SMC} is straight-forward, with one exception: at each SMC step $G_{t-1} \to G_t$, we collapse the dependence on the particular values of $B_t, K_t$ in that sampling iteration so that edges are proposed from the collapsed transition kernel $q^t_{\alpha,\lambda}(G_t \on G_{t-1}, \INot, \kNot)$. Thus, $G_t$ is composed of $G_{t-1}$ plus a random edge $e_t$ sampled as
\begin{align}
  P(e_t & = (v,u)) \propto \nonumber \\
  	 & \indicator\{(v,u)=(V_t,u^*)\}\vDist_{t-1}(v) \Bigl( {\textstyle\frac{\chi_{-t}}{\chi_{-t} + \omega_{-t}} + \frac{\omega_{-t}}{\chi_{-t} + \omega_{-t}}} \Bigl( \sum_{u \in \nball(v)} \left[\RWprob^{\kappa_{-t}, \bar{\tau}_{-t}}_{t-1} \right]_{v,u} \Bigr) \Bigr) + \dots \nonumber \\
  	 & \indicator\{(v,u)=(V_t,U_t)\} \frac{\omega_{-t}}{\chi_{-t} + \omega_{-t}} \left( \vDist_{t-1}(v) \left[\RWprob^{\kappa_{-t}, \bar{\tau}_{-t}}_{t-1} \right]_{v,u} + \vDist_{t-1}(u) \left[\RWprob^{\kappa_{-t}, \bar{\tau}_{-t}}_{t-1} \right]_{u,v} \right) \;.
\nonumber
\end{align}

\clearpage

\section{Sensitivity to prior specification}
\label{sec:prior_sensitivity}

% \subsection{Influence of priors} \label{sec:prior:sensitivity}

Since the MCMC approach requires prior distributions on the model parameters $\alpha$ and $\lambda$, we have
to consider how different choices of priors imply assumptions on the graph structure, and how they
affect inference. Interpreting $\alpha$ is straightforward, as an expected ratio of vertices to edges (see also
\cref{fig:examples:uniform}). Since
it also has compact domain, the experiments reported in this section draw it from a uniform prior.
Of interest here is the influence of $\lambda$. 
The following tables report the empirical mean of several pertinent statistics,
computed on graphs generated from the model using a gamma prior with parameters $(a_{\lambda},b_{\lambda})$. 
Standard deviations
are computed over 100 realizations of a simple graph with 500 edges.
\begin{center}
  \makebox[\textwidth][c]{
    \resizebox{\textwidth}{!}{
  \begin{tabular}{ccc|cccccc}
    $\RWU$ $(a_{\lambda},b_{\lambda})$ & $\E[\lambda]$ & $\text{Var}[\lambda]$ & \emph{Diameter} & \emph{Avg. SP} & \emph{Clustering coeff.} & \emph{Vertices} & \emph{$\max(\deg)$} & \emph{$\overline{\deg}$} \\
    \midrule 
    (1.000, 0.250) & 4 & 16 & 16.53 (4.7) & 6.98 (2.1) & 0.1379 (0.088) & 375.2 (107) & 17.1 (7.4) & 3.02 (1.3) \\ 
    (4.000, 1.000) & 4 & 4 & 16.31 (4.4) & 6.81 (2.0) & 0.1486 (0.084) & 365.6 (99) & 17.6 (7.0) & 3.01 (1.1) \\ 
    (1.000, 1.000) & 1 & 1 & 19.01 (3.3) & 7.98 (1.4) & 0.1665 (0.097) & 417.9 (67) & 14.5 (4.2) & 2.47 (0.5) \\ 
    (0.001, 0.001) & 1 & 1000 & 21.87 (1.7) & 9.62 (0.5) & 0.0000 (0.000) & 501.0 (0) & 9.6 (1.3) & 2.00 (0.0) \\ 
    \\
  \end{tabular}
}}
\end{center}
\begin{center}
  \makebox[\textwidth][c]{
    \resizebox{\textwidth}{!}{
  \begin{tabular}{ccc|cccccc}
    $\RWSB$ $(a_{\lambda},b_{\lambda})$ & $\E[\lambda]$ & $\text{Var}[\lambda]$ & \emph{Diameter} & \emph{Avg. SP} & \emph{Clustering coeff.} & \emph{Vertices} & \emph{$\max(\deg)$} & \emph{$\overline{\deg}$} \\ 
    \midrule 
    (1.000, 0.250) & 4 & 16 & 11.28 (3.0) & 4.71 (1.1) & 0.1938 (0.117) & 395.9 (84) & 58.1 (18.4) & 2.67 (0.7) \\ 
    (4.000, 1.000) & 4 & 4 & 11.13 (3.2) & 4.62 (1.1) & 0.2054 (0.119) & 384.8 (82) & 55.2 (17.6) & 2.74 (0.7) \\ 
    (1.000, 1.000) & 1 & 1 & 12.67 (2.6) & 5.17 (0.9) & 0.1990 (0.117) & 429.4 (58) & 52.5 (14.5) & 2.38 (0.4) \\ 
    (0.001, 0.001) & 1 & 1000 & 15.75 (1.7) & 6.37 (0.5) & 0.0000 (0.000) & 501.0 (0) & 43.8 (13.8) & 2.00 (0.0) \\ 
    \\
  \end{tabular}
}}
\end{center}
The differences between the $\RWU$ and $\RWSB$ model are clearly visible: Regardless of prior choice, the latter
tends to have smaller diameter and average shortest path length, and a higher clustering coefficient. 
\cref{fig:prior:stats} shows averages for three of these statistics generated with fixed parameter values.

The last row of either table, with parameters ${(0.001,0.001)}$, corresponds to a gamma distribution commonly used as a ``non-informative'' prior in Bayesian analysis, and exhibits a phenomenon often observed for such priors:
On the one hand, the prior places most of its mass near 0, with the result that simulation from the prior
tends to produce very small parameters values. In our model, such small values of $\lambda$ generate
a tree-like graph (note the number of vertices in all experiments is exactly 501, for 500 edges). 
On the other hand, the tail of the
prior decays sufficiently slowly that even a small sample can pull the posterior away from 0, and indeed
parameter updates
in the Gibbs sampler are nearly equal to the MLE of $\lambda$ based on the latent walk lengths $\mathbf{K}$. 

More generally, the definition of the model implies that the posterior should peak somewhere between zero to the mixing time of the random walk. 
More important for inference than the prior's shape is hence that it does not exclude this region from its support, and as a rule of thumb, any prior that spreads its mass well over the interval between 0 and 2-3 times the diameter of the network can be expected to yield reasonable results in most cases. %See \cref{sec:prior_sensitivity} for more detailed experiments studying the effect of prior specification on posterior inference.

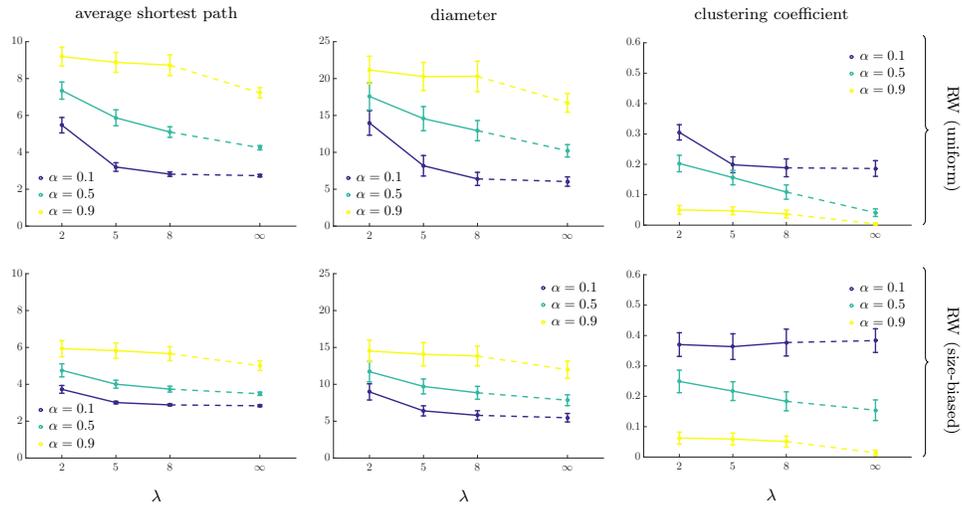
\begin{figure}[t!]
\makebox[\textwidth][c]{
  \resizebox{\textwidth}{!}{
    \begin{tikzpicture}[mybraces]
  \path[use as bounding box]
  (-2.5,2.5) rectangle (13.7,-6.25);
  \begin{scope}[xshift=0cm]
    \node (a) at (0,0) {
      \includegraphics[width=5cm]{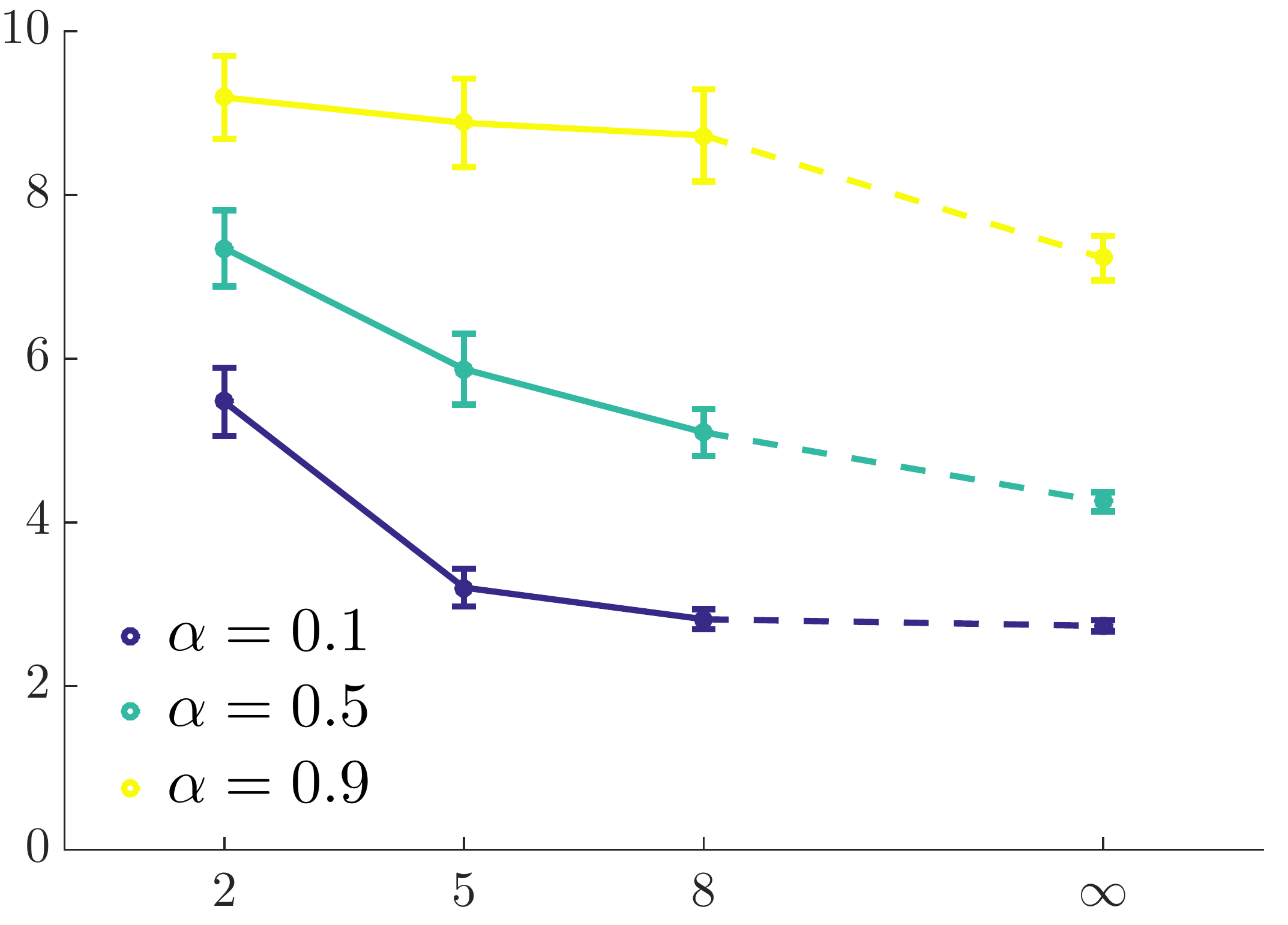}   
    };                                                                                 
    \node (b) at (0,-4) {                                                              
      \includegraphics[width=5cm]{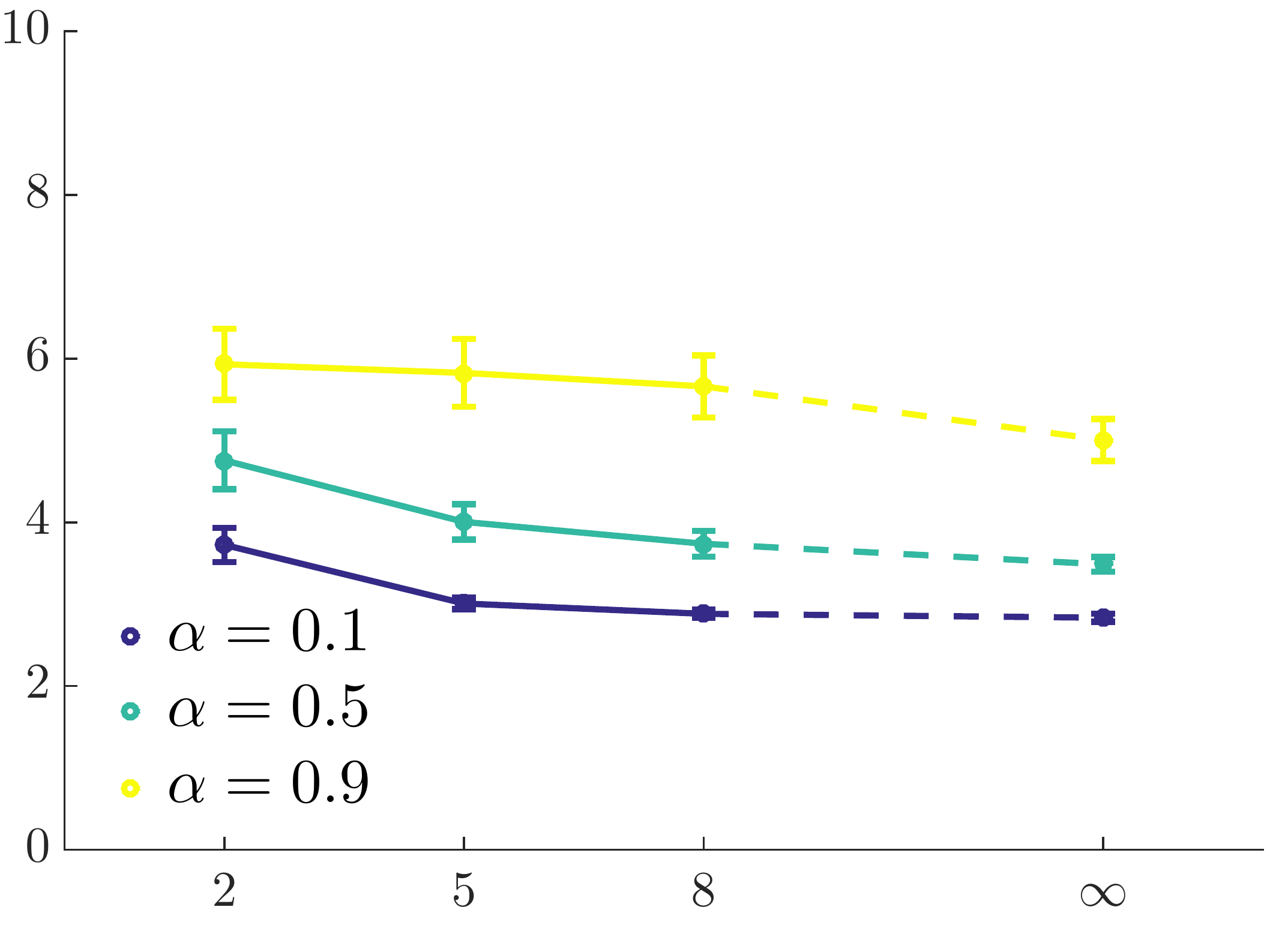}  
    };                              
    \node at ($(a.north)+(0,0.25)$) {\footnotesize average shortest path};                                                   
    \node at ($(b.south)-(0,0.25)$) {\footnotesize{$\lambda$}};
  \end{scope}
  \begin{scope}[xshift=5.3cm]
    \node (a) at (0,0) {
      \includegraphics[width=5cm]{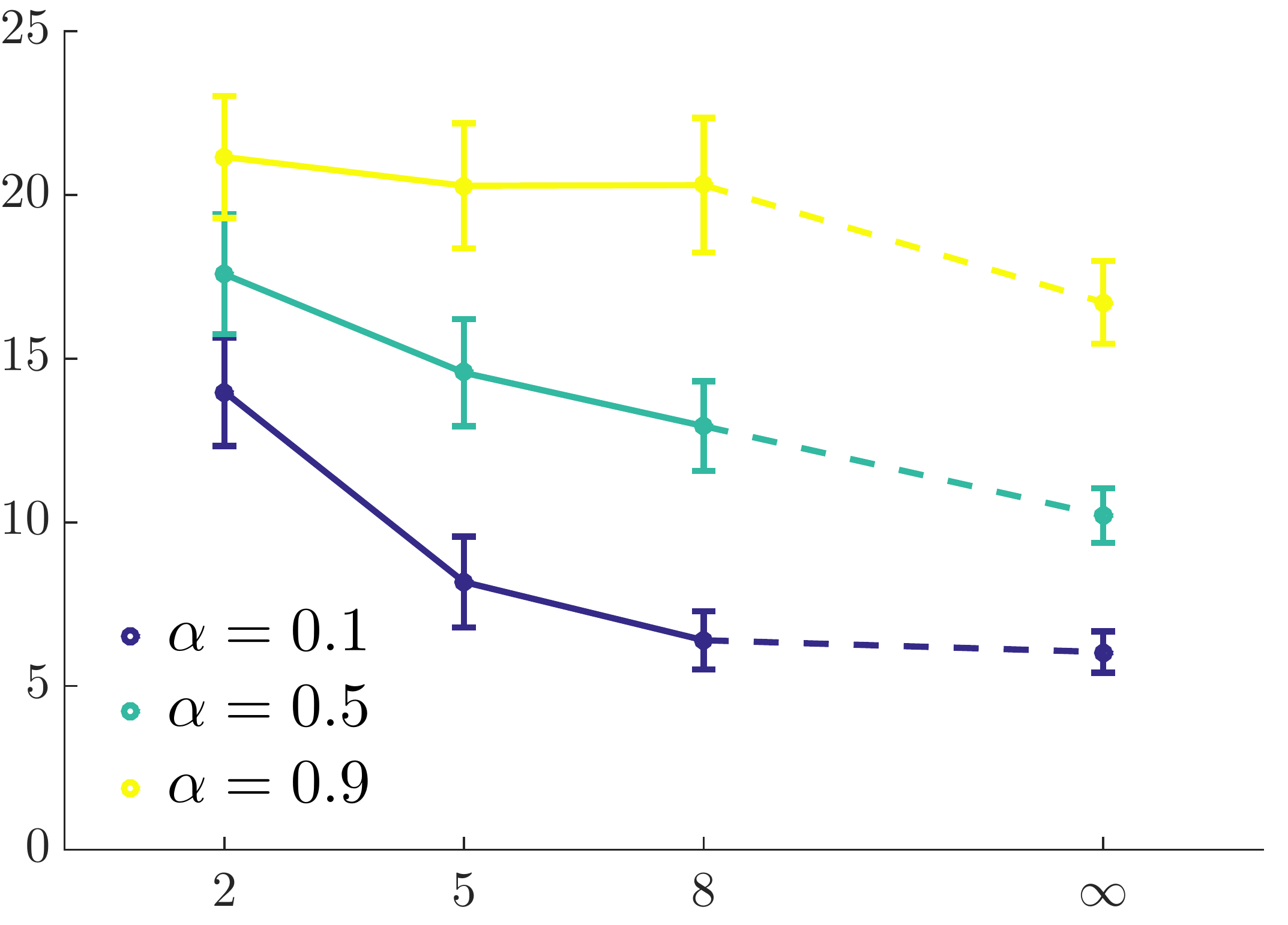}   
    };                                                                                 
    \node (b) at (0,-4) {                                                              
      \includegraphics[width=5cm]{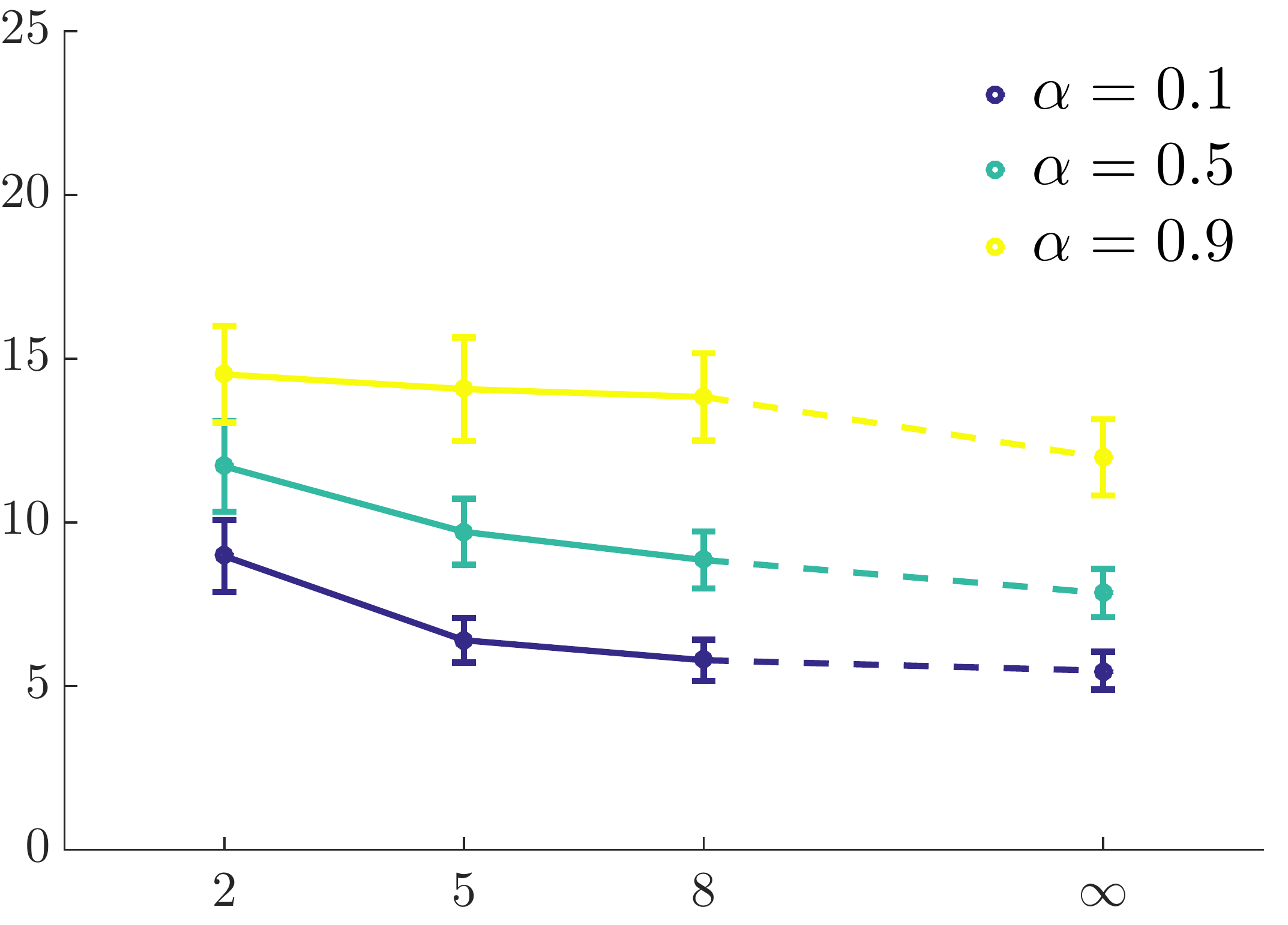}  
    };                   
    \node at ($(a.north)+(0,0.25)$) {\footnotesize diameter};                                                                
    \node at ($(b.south)-(0,0.25)$) {\footnotesize{$\lambda$}};
  \end{scope}
  \begin{scope}[xshift=10.6cm]
    \node (a) at (0,0) {
      \includegraphics[width=5cm]{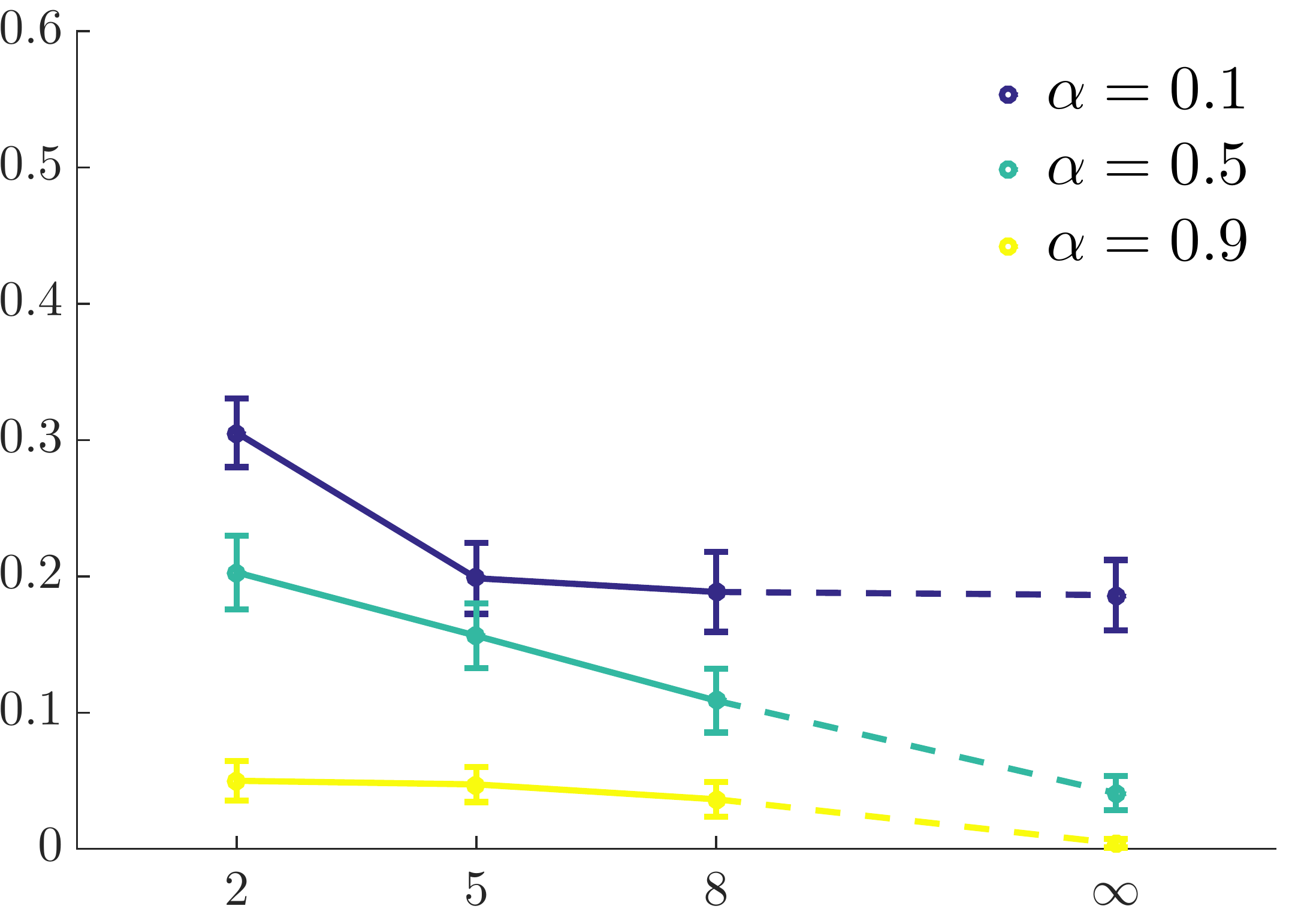}   
    };                                                                                 
    \node (b) at (0,-4) {                                                              
      \includegraphics[width=5cm]{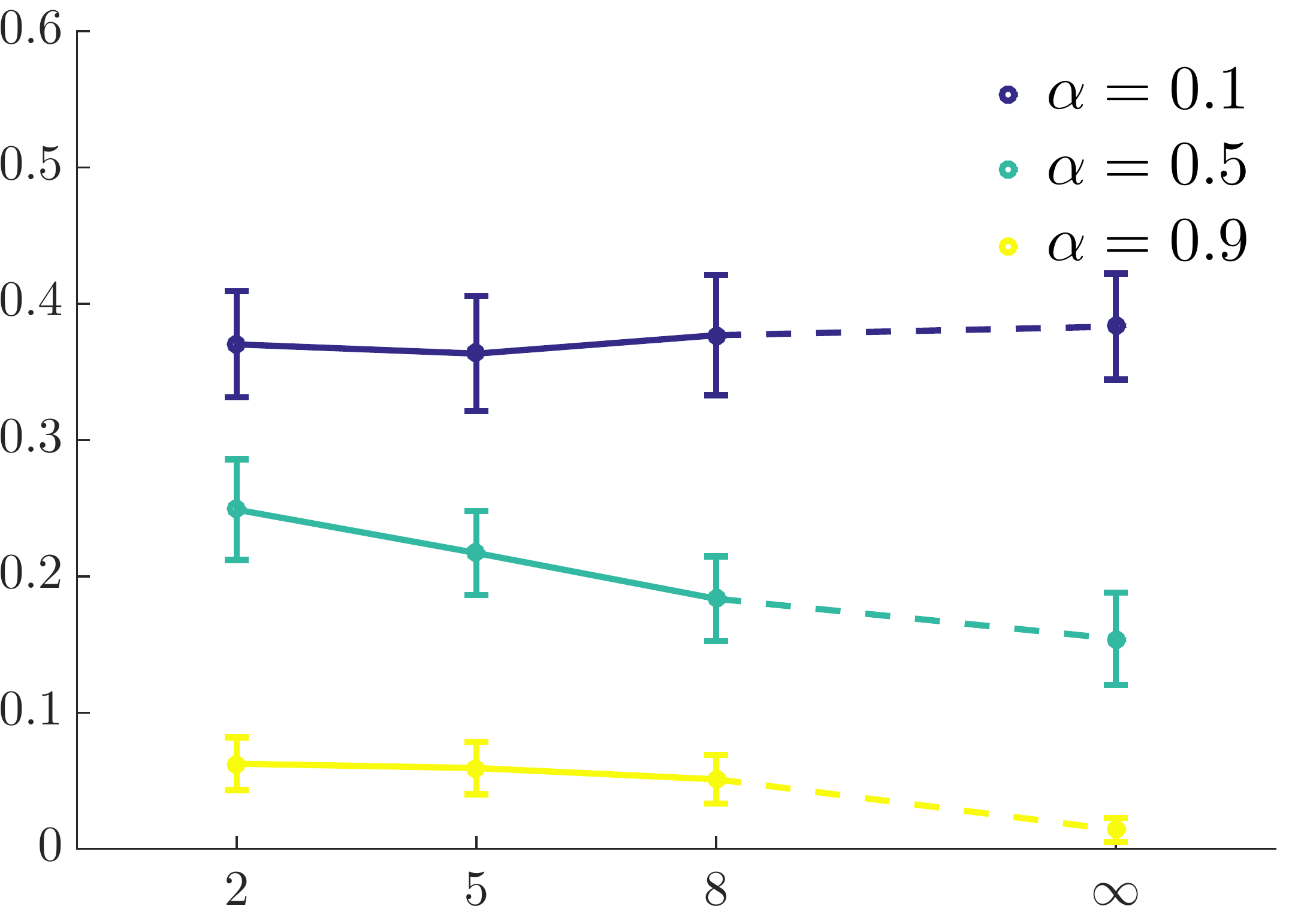}  
    };                                                  
    \node at ($(a.north)+(0,0.25)$) {\footnotesize clustering coefficient};                                 
    \node at ($(b.south)-(0,0.25)$) {\footnotesize{$\lambda$}};
  \end{scope}
  \begin{scope}[xshift=11cm]
    \draw[brace] ($(a.north east)+(0,-.1)$)--($(a.south east)+(0,.4)$);
    \node[rotate=270] at ($(a.east)+(.5,0)$) {\footnotesize RW (uniform)};
    
    \draw[brace] ($(b.north east)+(0,-.1)$)--($(b.south east)+(0,.4)$);
    \node[rotate=270] at ($(b.east)+(.5,0)$) {\footnotesize RW (size-biased)};
  \end{scope}
\end{tikzpicture}

  }}
  \caption{Statistics (average shortest path, diameter, clustering coefficient) of simple graphs with 500 edges generated by the random walk model. Averages are over 100 samples and error bars are plus/minus one standard error. $\lambda=\infty$ corresponds to the limiting $\ACL$ model (see \cref{sec:pa}).}
  \label{fig:prior:stats}
  \vspace{-.5cm}
\end{figure}

In order to study the sensitivity of posterior inference to prior specification, we vary the prior distribution parameters and use the particle Gibbs sampler to fit the $\RWSB(\alpha,\lambda)$ model to synthetic graphs of two difference sizes ($T=50$ and $T=100$). \Cref{fig:prior_sensitivity} show $1000$ samples collected after 500 burn-in iterations, each run for combinations of three different priors on $\alpha$ and $\lambda$. 
%% \begin{center}
%%   \begin{tabular}{cccc}
%%   Scenario & Prior on $\alpha$ & Prior on $\lambda$ & Plot \\
%%   \midrule
%%   1 & $\BetaDist(0.5,0.5)$ & $\GammaDist(1.0,0.25)$ & UL \\
%%   2 & $\BetaDist(1,1)$ & $\GammaDist(1.0,0.25)$ & UM \\
%%   3 & $\BetaDist(2,2)$ & $\GammaDist(1.0,0.25)$ & UR \\
%%   4 & $\BetaDist(0.5,0.5)$ & $\GammaDist(20,2)$ & ML \\
%%   5 & $\BetaDist(1,1)$ & $\GammaDist(20,2)$ & MM \\
%%   6 & $\BetaDist(2,2)$ & $\GammaDist(20,2)$ & MR \\
%%   7 & $\BetaDist(0.5,0.5)$ & $\GammaDist(0.01,0.01)$ & BL \\
%%   8 & $\BetaDist(1,1)$ & $\GammaDist(0.01,0.01)$ & BM \\
%%   9 & $\BetaDist(2,2)$ & $\GammaDist(0.01,0.01)$ & BR
%%   \end{tabular}
%% \end{center}
As one would expect, problems arise if the prior peaks sharply at a value far away from the true parameter value.
In the absence of strong prior knowledge, it is hence advisable to choose reasonably uninformative priors. In particular, the beta prior on $\alpha$ should usually be chosen as uniform. If the prior on $\alpha$ has a distinctive peak, it may distort posterior inference, particularly in small data settings. For example, the bottom row of \cref{fig:prior_sensitivity} (bottom) shows this effect. 
% in the upper half of the interval, it can distort the posterior of $\lambda$ in a similar manner as evident in 
% \cref{fig:param:sweep}.

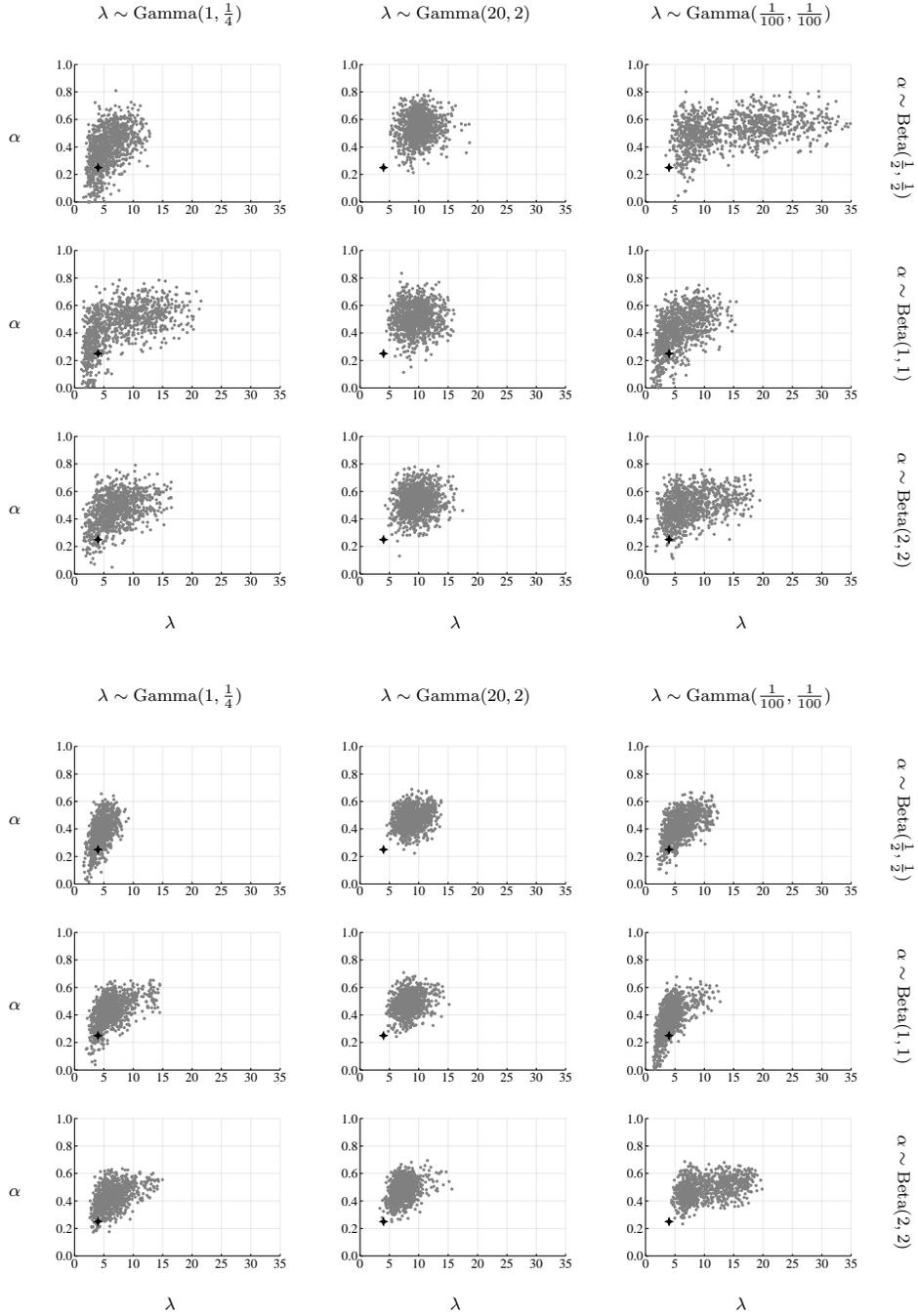
\begin{figure}[tb] 
    \makebox[\textwidth][c]{
    \resizebox{\textwidth}{!}{
  \begin{tikzpicture}
    \begin{scope}
    \node at (0,0) {
      \includegraphics[width=0.3\textwidth]{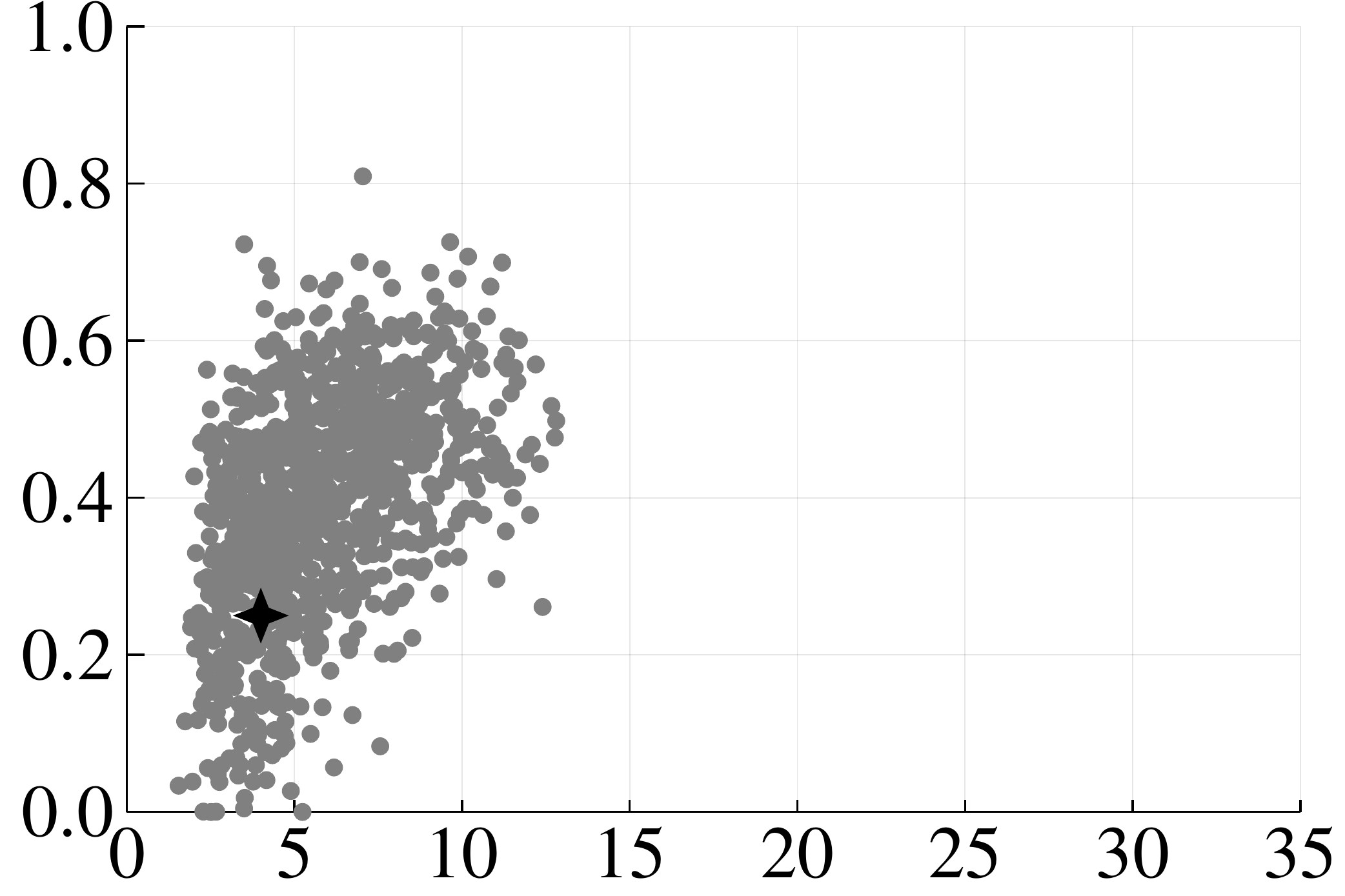}
    };
    \node at (4.6,0) {
      \includegraphics[width=0.3\textwidth]{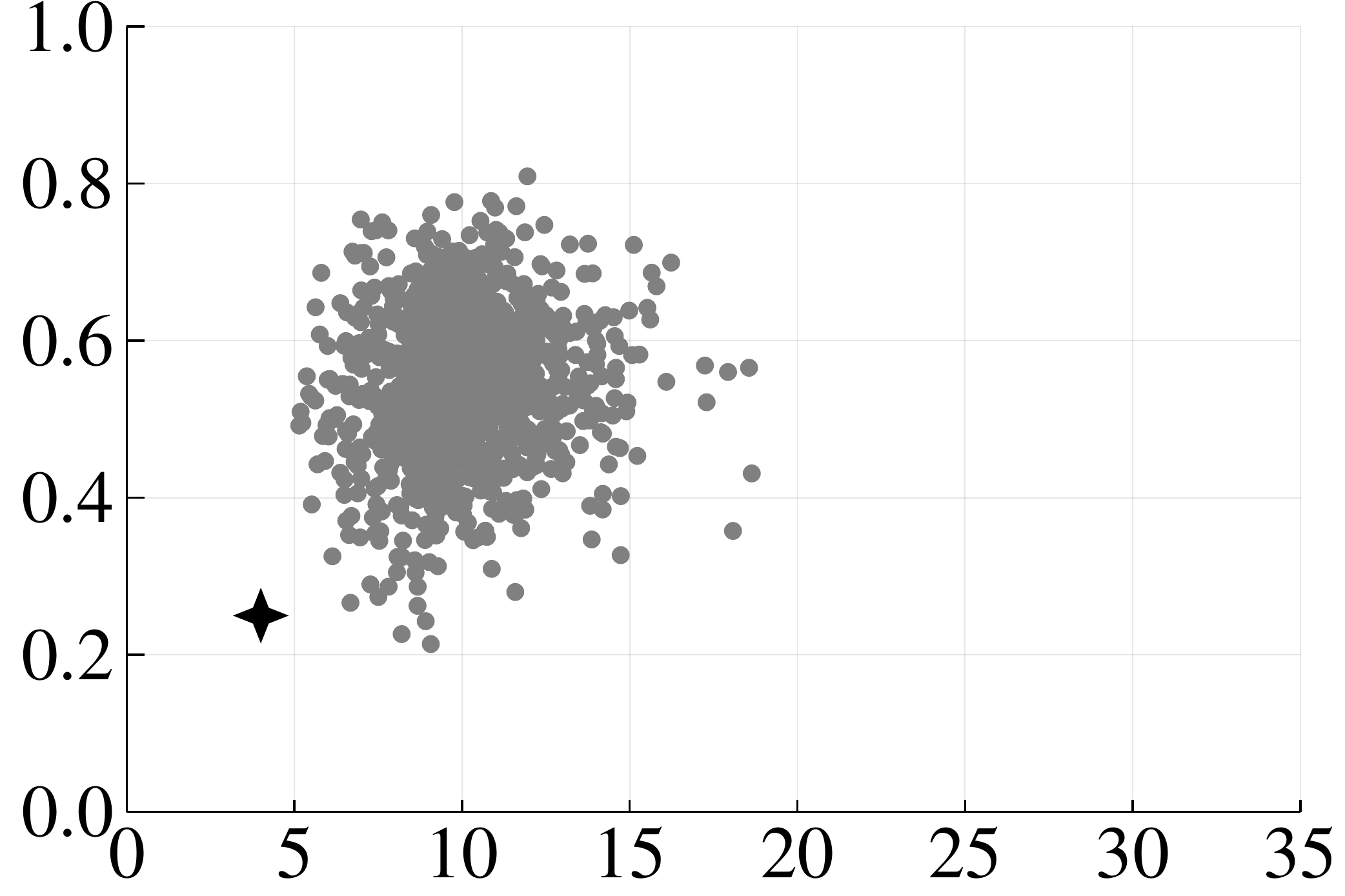}
    };
    \node at (9.2,0) {
      \includegraphics[width=0.3\textwidth]{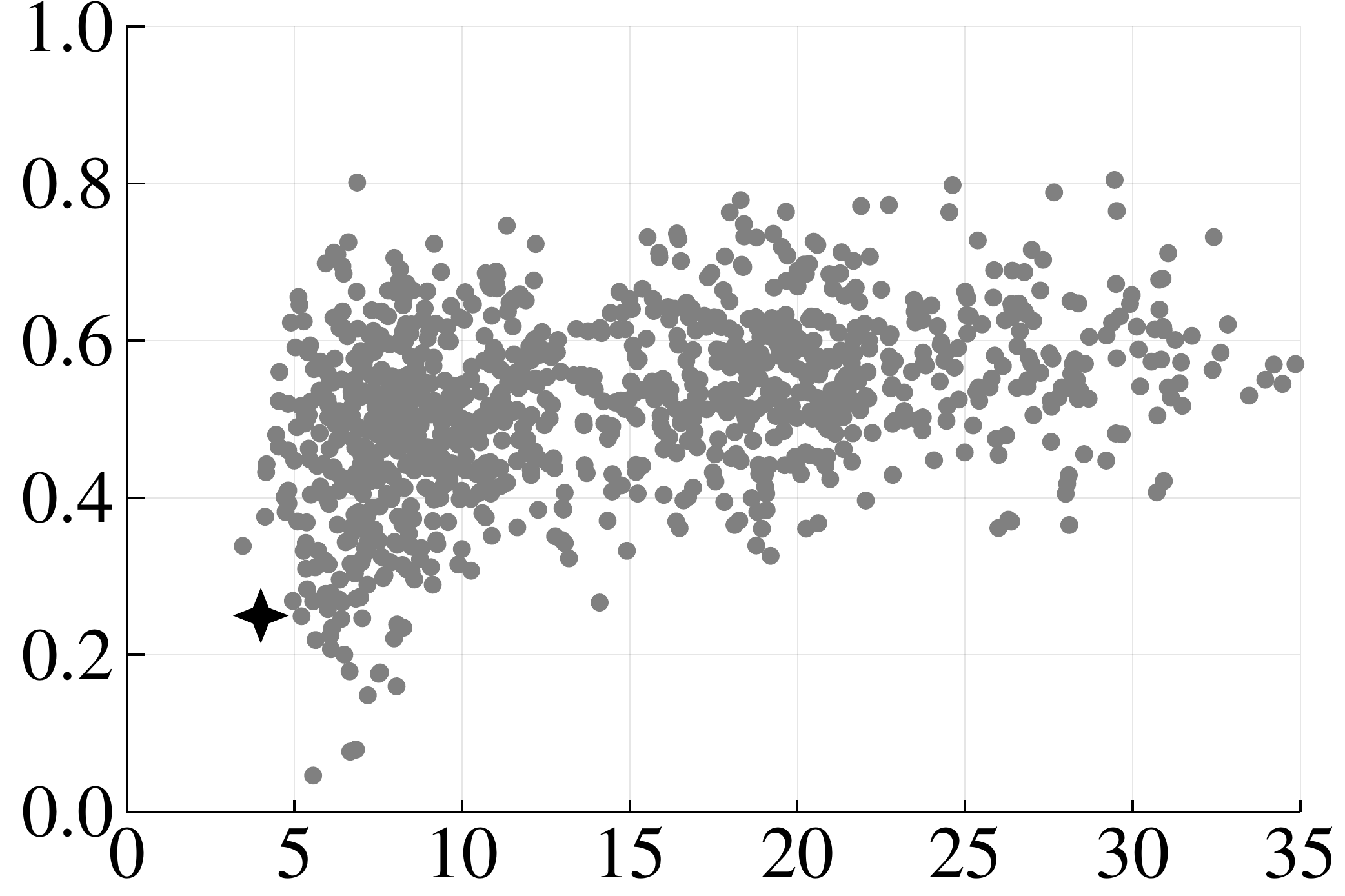}
    };
    \node at (0,-3.0) {
      \includegraphics[width=0.3\textwidth]{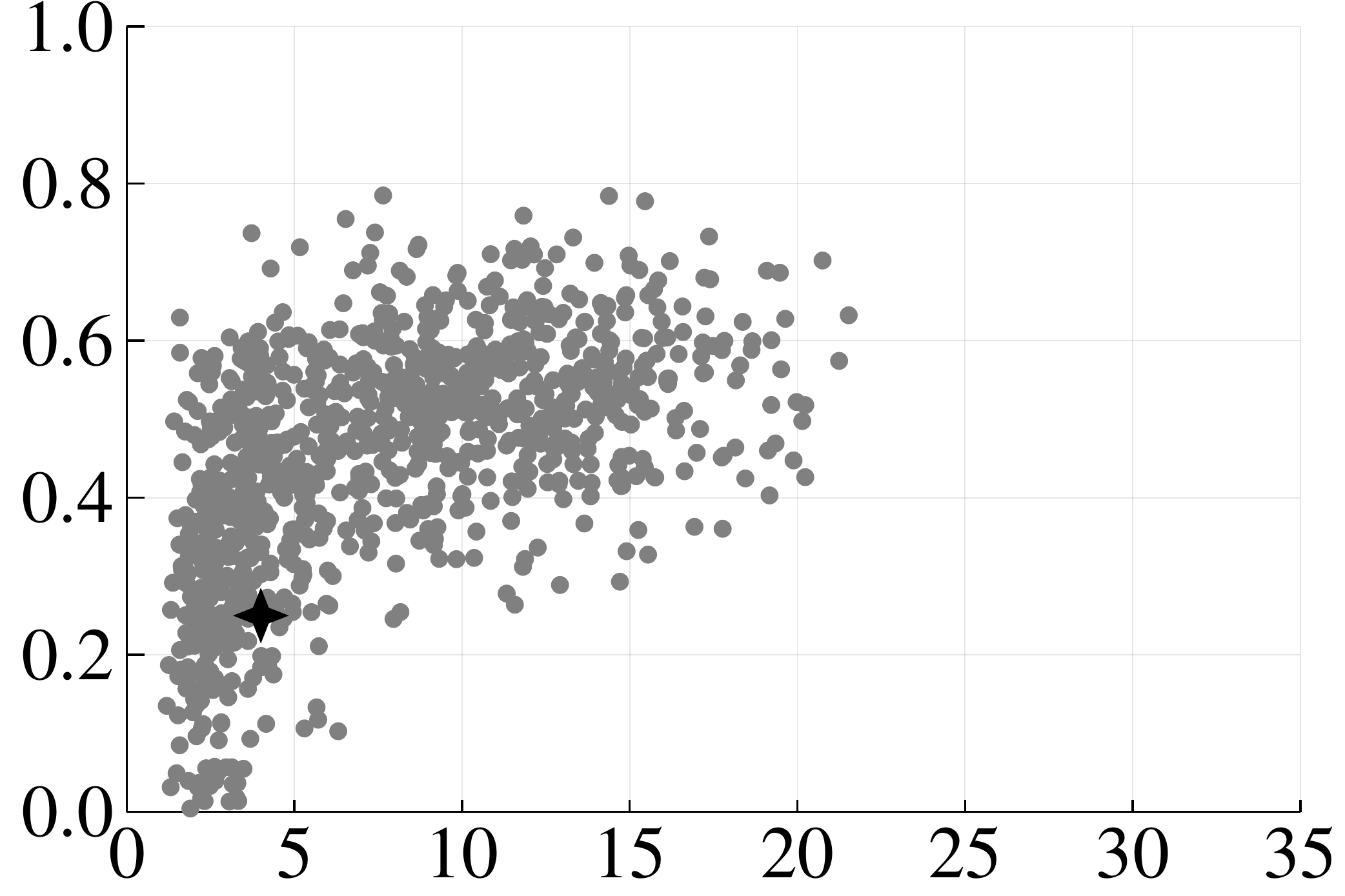}
    };
    \node at (4.6,-3.0) {
      \includegraphics[width=0.3\textwidth]{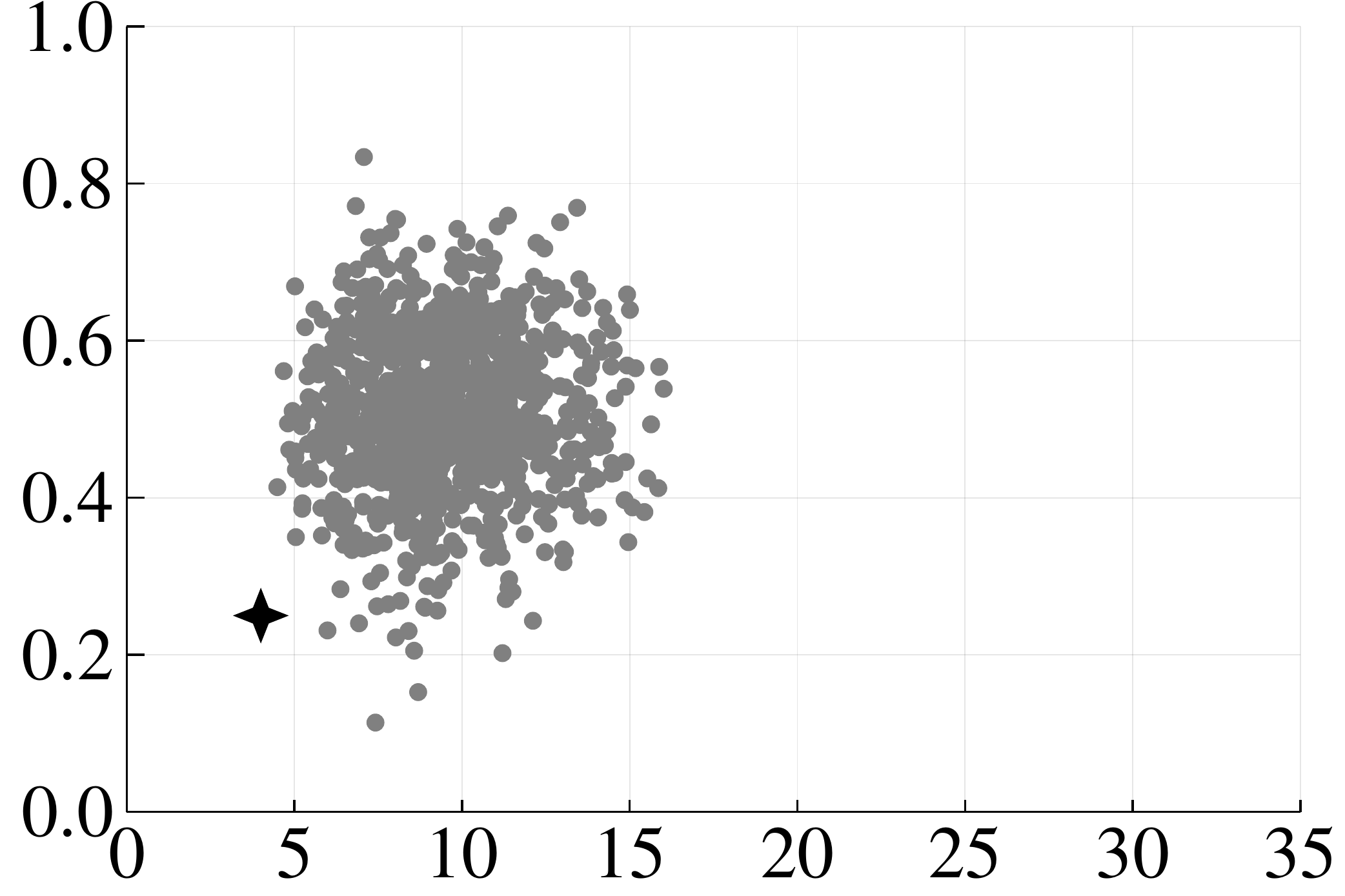}
    };
    \node at (9.2,-3.0) {
      \includegraphics[width=0.3\textwidth]{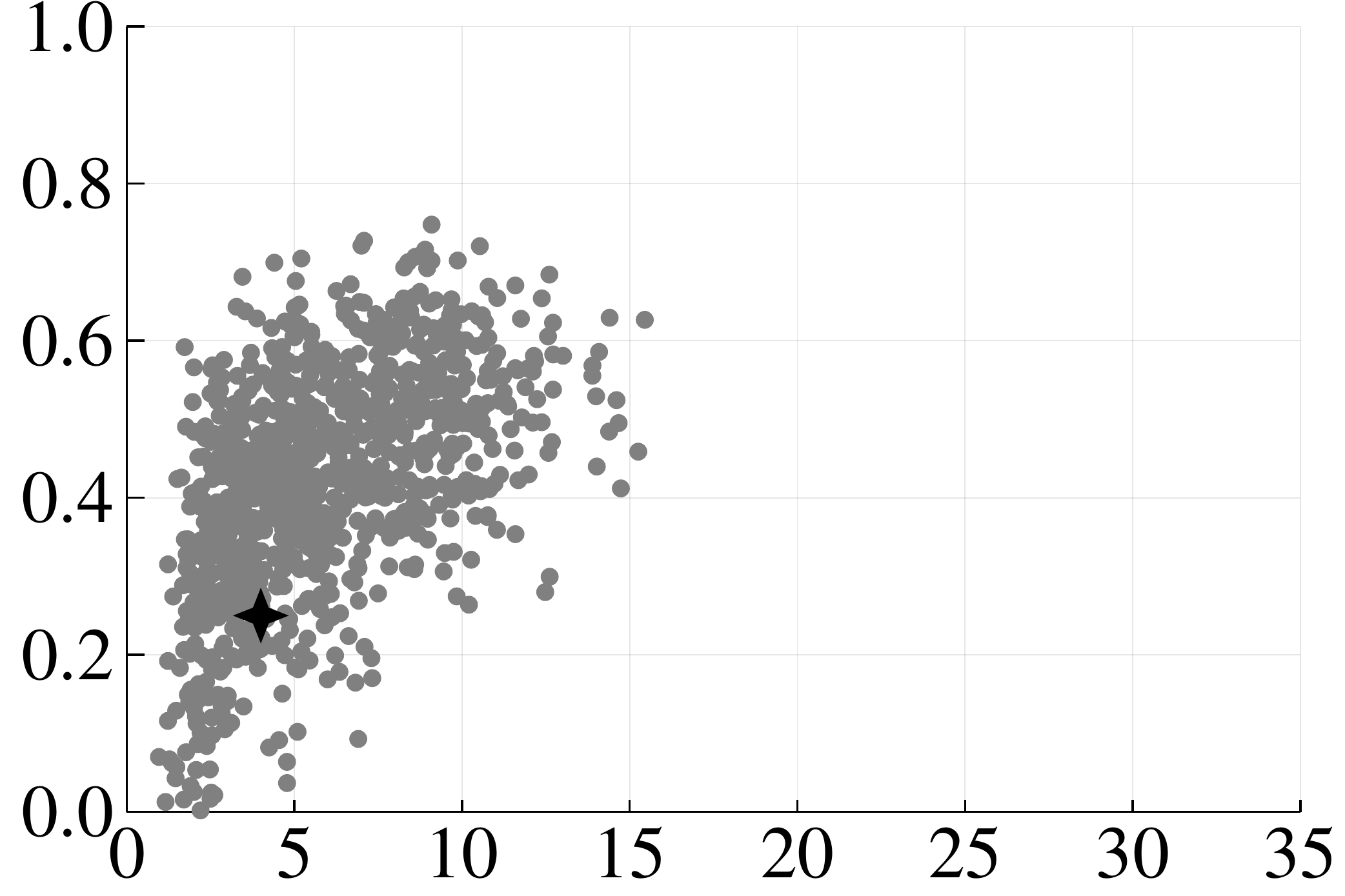}
    };
    \node at (0,-6.0) {
      \includegraphics[width=0.3\textwidth]{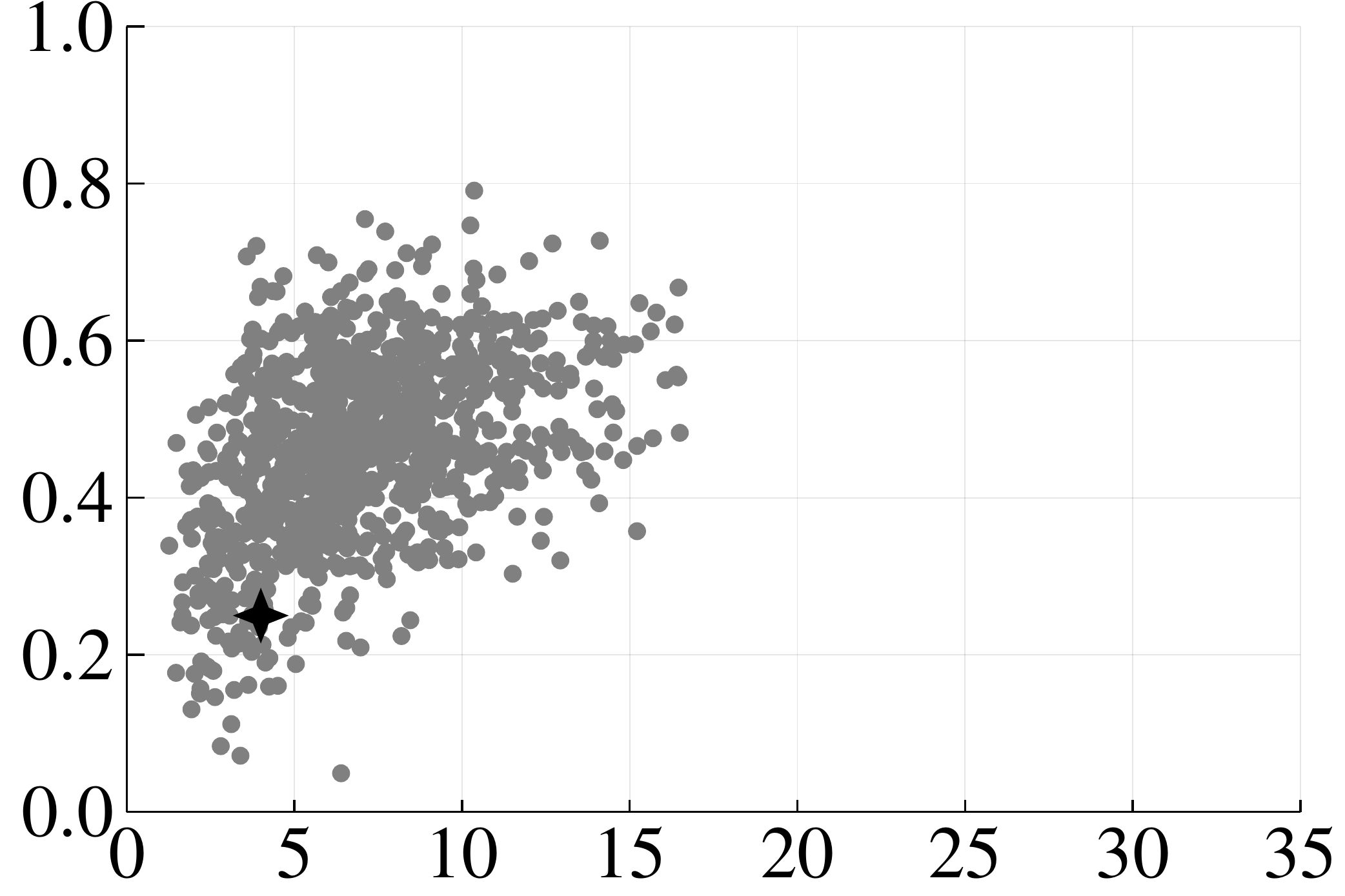}
    };
    \node at (4.6,-6.0) {
      \includegraphics[width=0.3\textwidth]{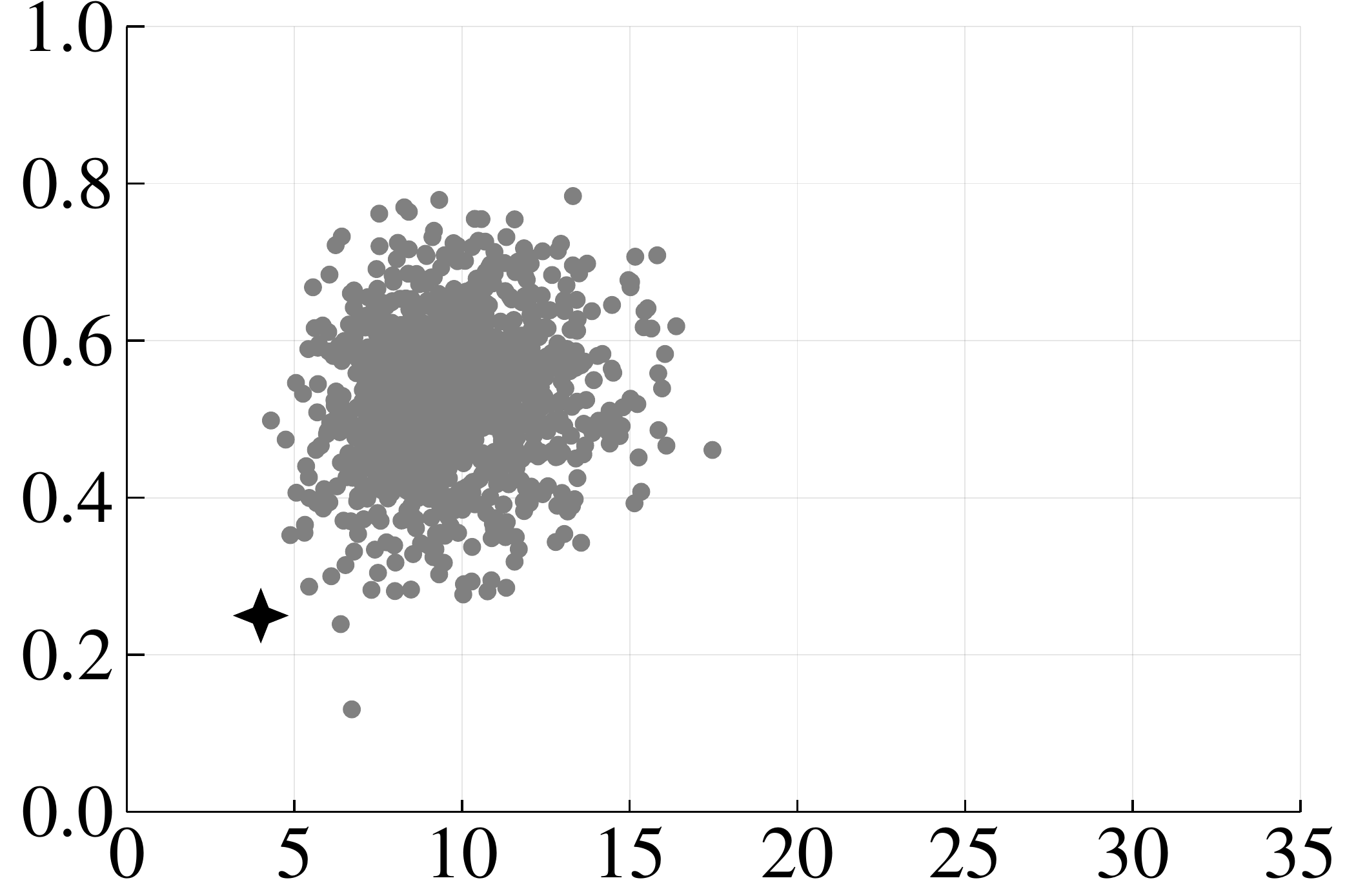}
    };
    \node at (9.2,-6.0) {
      \includegraphics[width=0.3\textwidth]{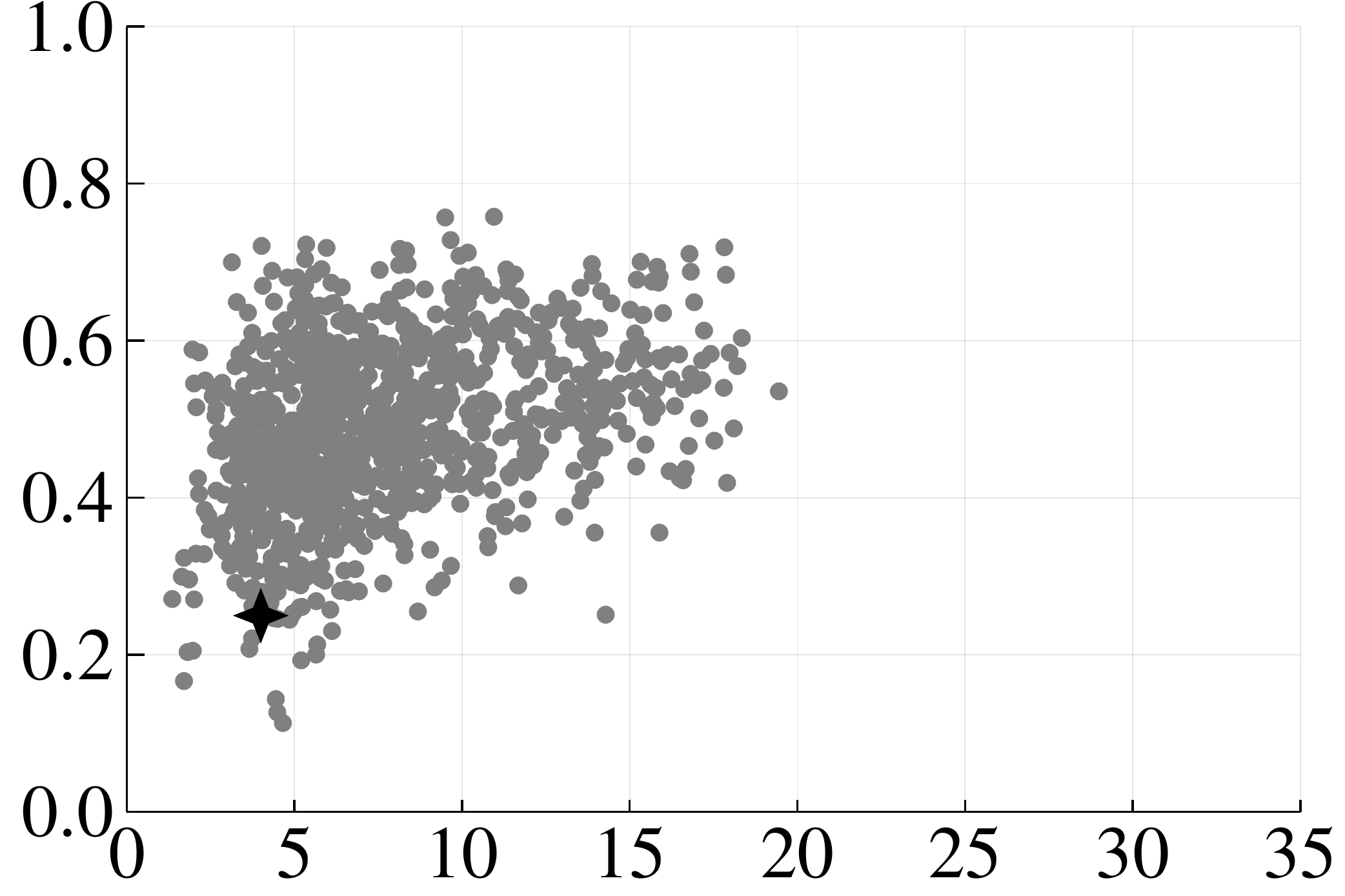}
    };
    
    \node at (-2.5,0) {\footnotesize $\alpha$};
    \node at (-2.5,-3.0) {\footnotesize $\alpha$};
    \node at (-2.5,-6.0) {\footnotesize $\alpha$};

    \node at (0,-7.8) {\footnotesize $\lambda$};
    \node at (4.6,-7.8) {\footnotesize $\lambda$};
    \node at (9.2,-7.8) {\footnotesize $\lambda$};

    \node[rotate=-90] at (11.8,0) {\footnotesize ${\alpha\sim\text{Beta}(\frac{1}{2},\frac{1}{2})}$};
    \node[rotate=-90] at (11.8,-3.0) {\footnotesize ${\alpha\sim\text{Beta}(1,1)}$};
    \node[rotate=-90] at (11.8,-6.0) {\footnotesize ${\alpha\sim\text{Beta}(2,2)}$};

    \node at (0,2) {\footnotesize ${\lambda\sim\text{Gamma}(1,\frac{1}{4})}$};
    \node at (4.6,2) {\footnotesize ${\lambda\sim\text{Gamma}(20,2)}$};
    \node at (9.2,2) {\footnotesize ${\lambda\sim\text{Gamma}(\frac{1}{100},\frac{1}{100})}$};
    \end{scope}
    \begin{scope}[yshift=-11cm]
    \node at (0,0) {
      \includegraphics[width=0.3\textwidth]{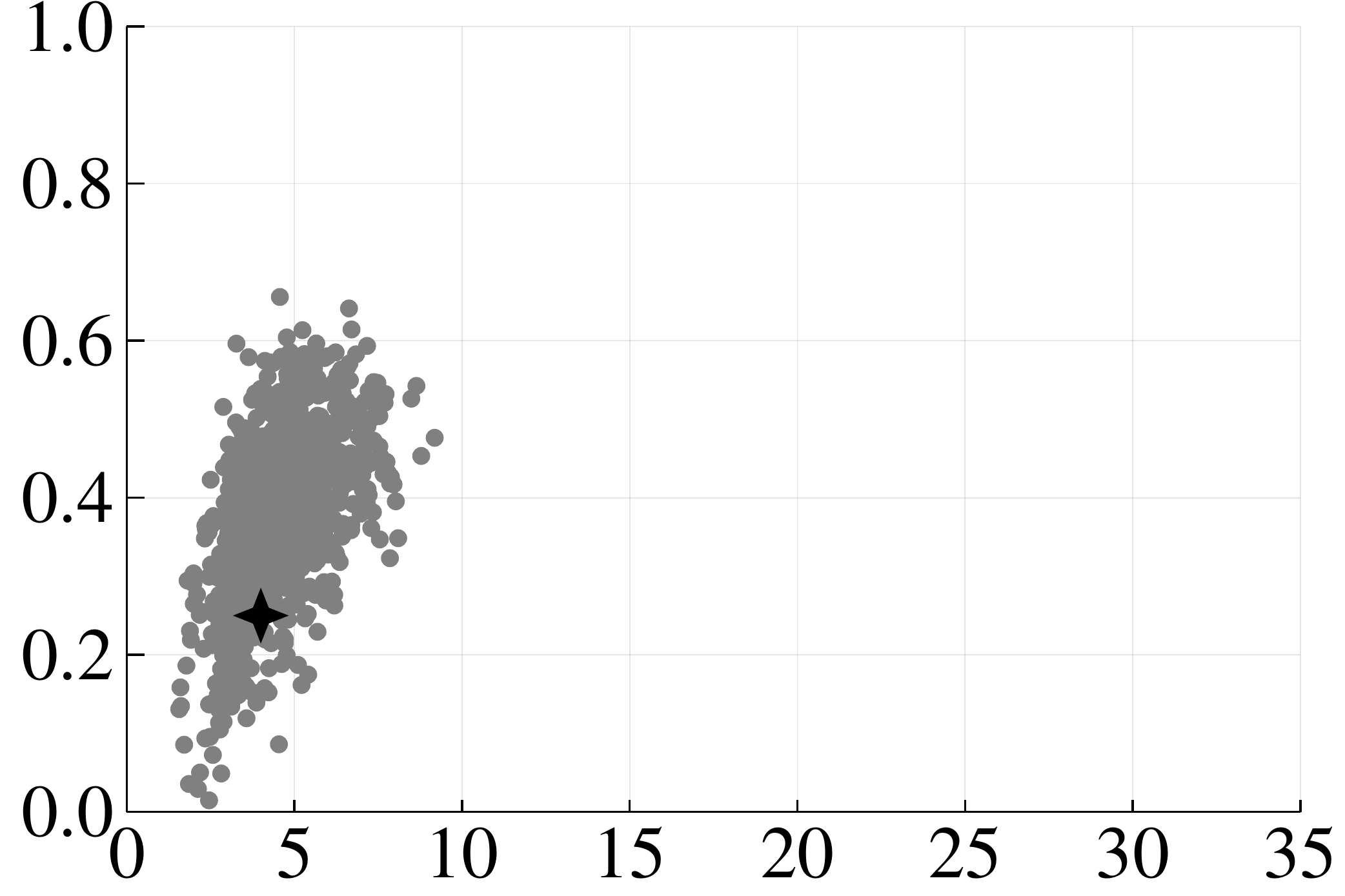}
    };
    \node at (4.6,0) {
      \includegraphics[width=0.3\textwidth]{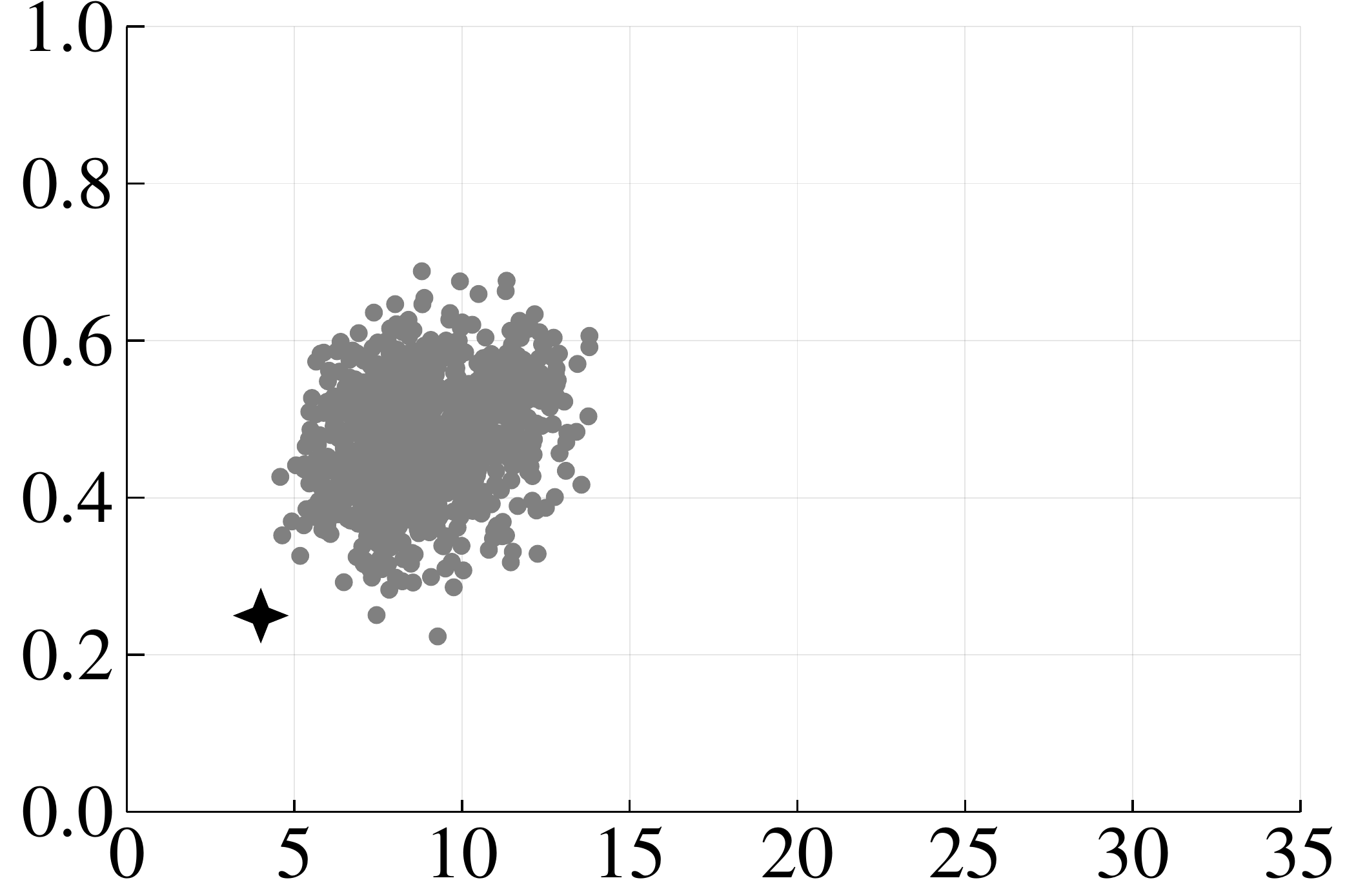}
    };
    \node at (9.2,0) {
      \includegraphics[width=0.3\textwidth]{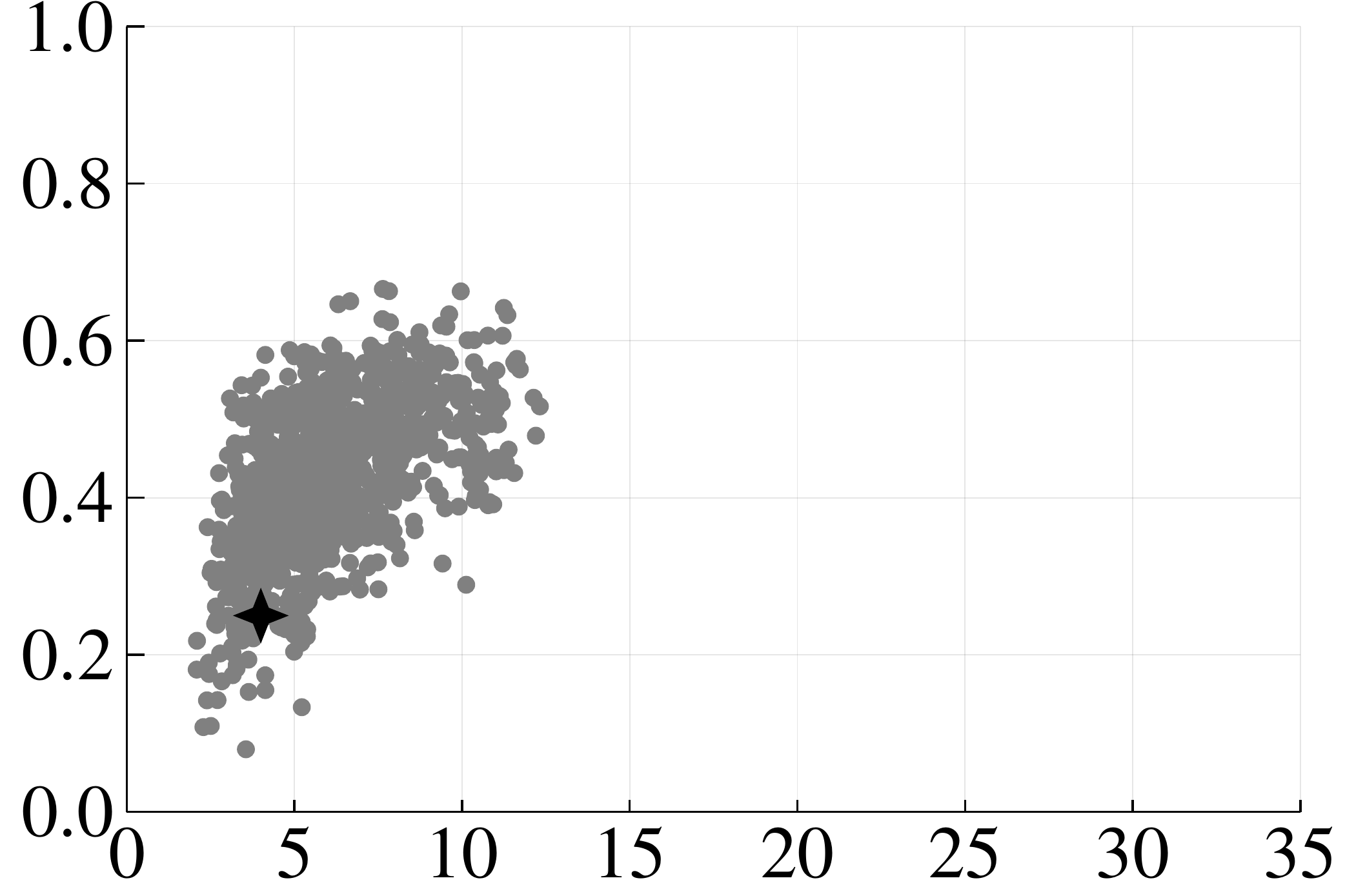}
    };
    \node at (0,-3.0) {
      \includegraphics[width=0.3\textwidth]{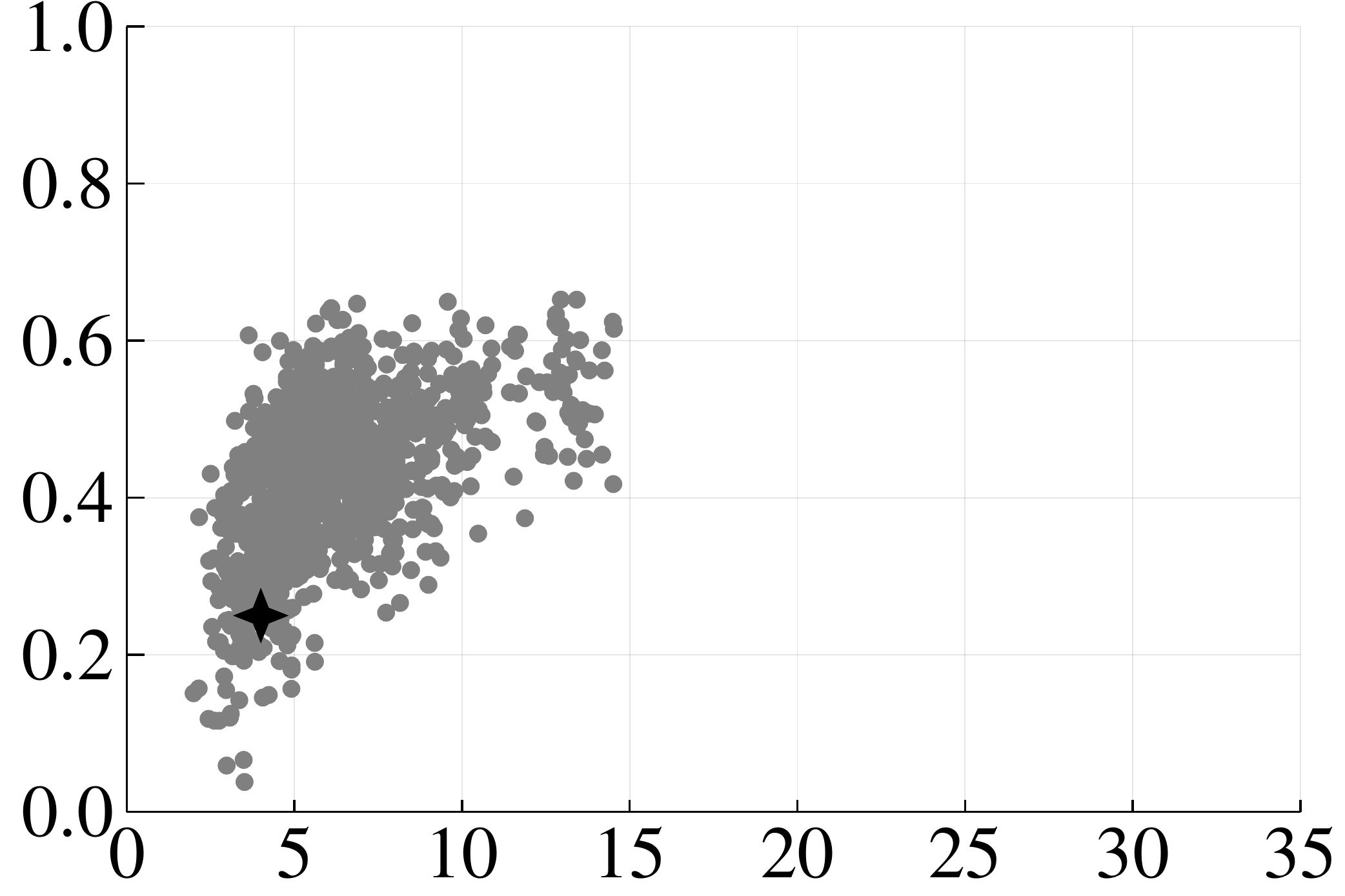}
    };
    \node at (4.6,-3.0) {
      \includegraphics[width=0.3\textwidth]{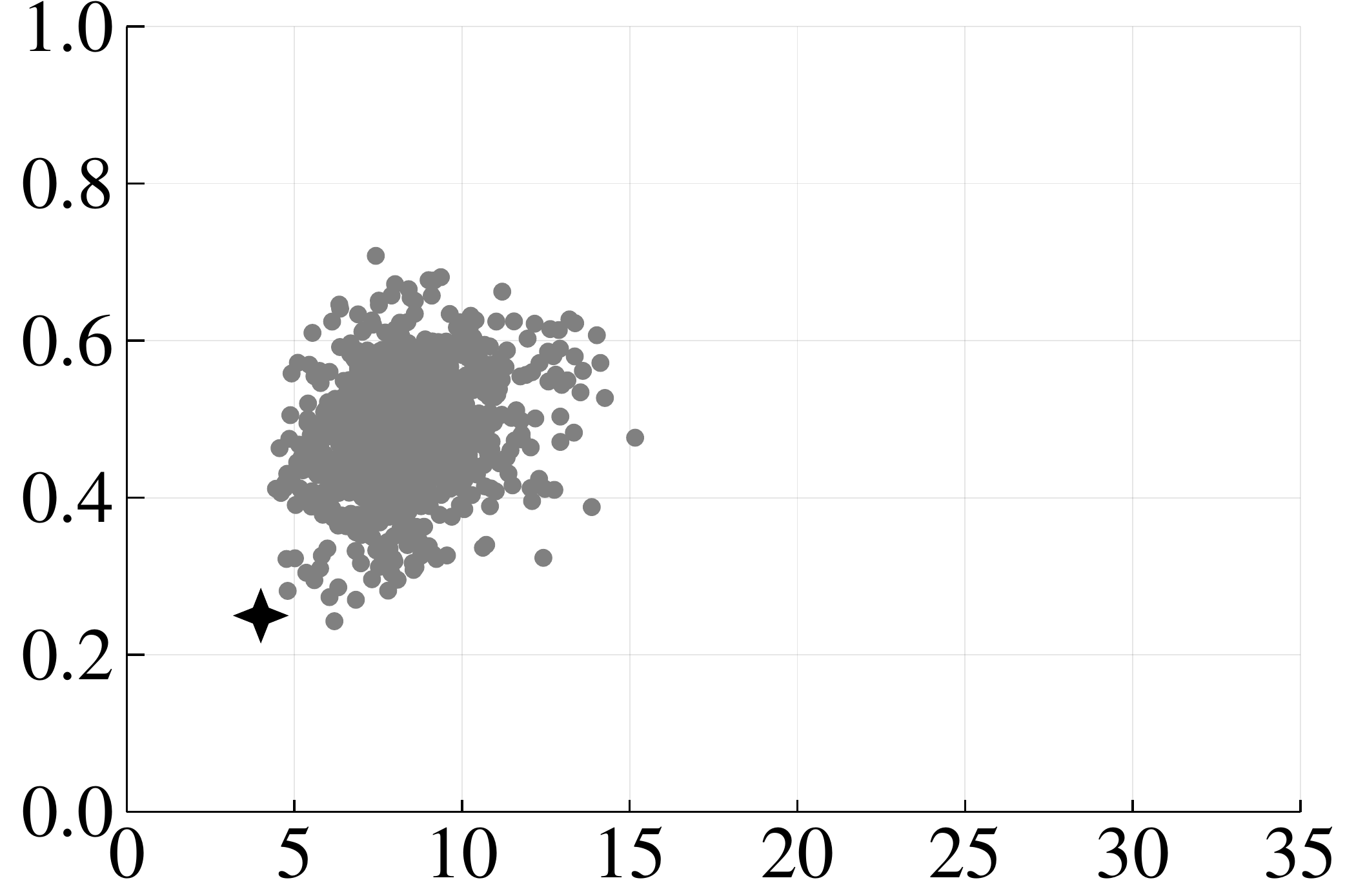}
    };
    \node at (9.2,-3.0) {
      \includegraphics[width=0.3\textwidth]{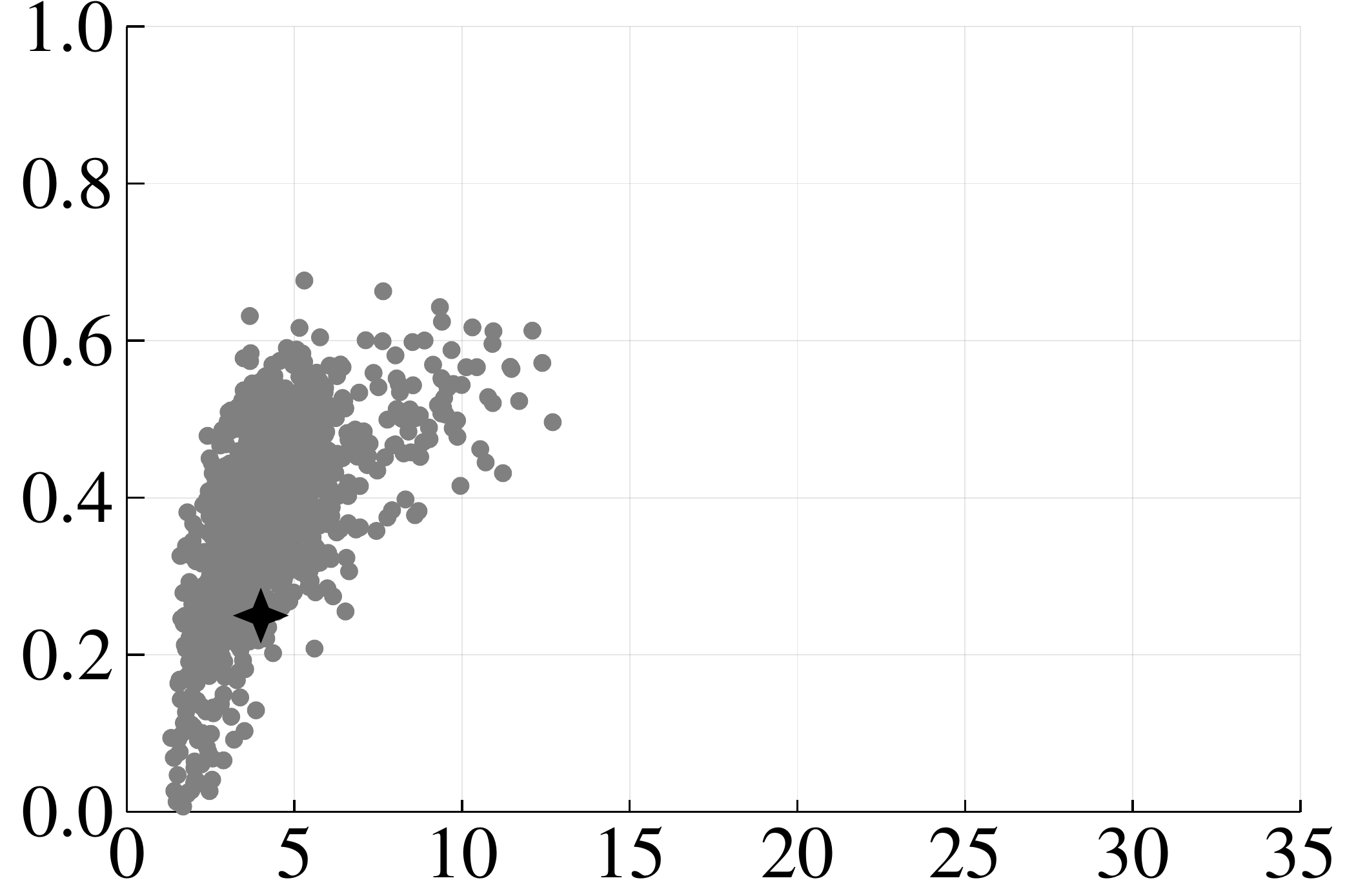}
    };
    \node at (0,-6.0) {
      \includegraphics[width=0.3\textwidth]{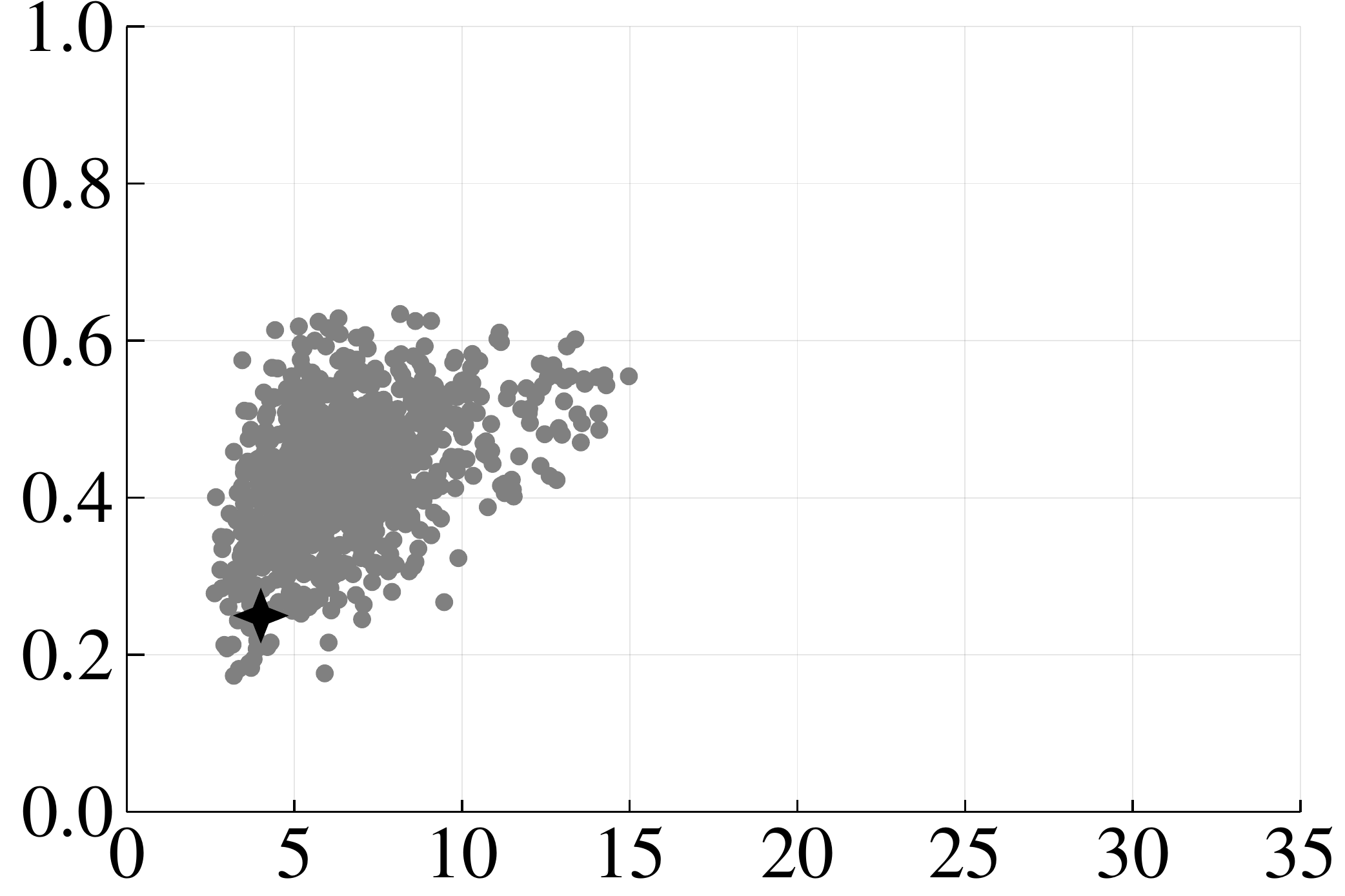}
    };
    \node at (4.6,-6.0) {
      \includegraphics[width=0.3\textwidth]{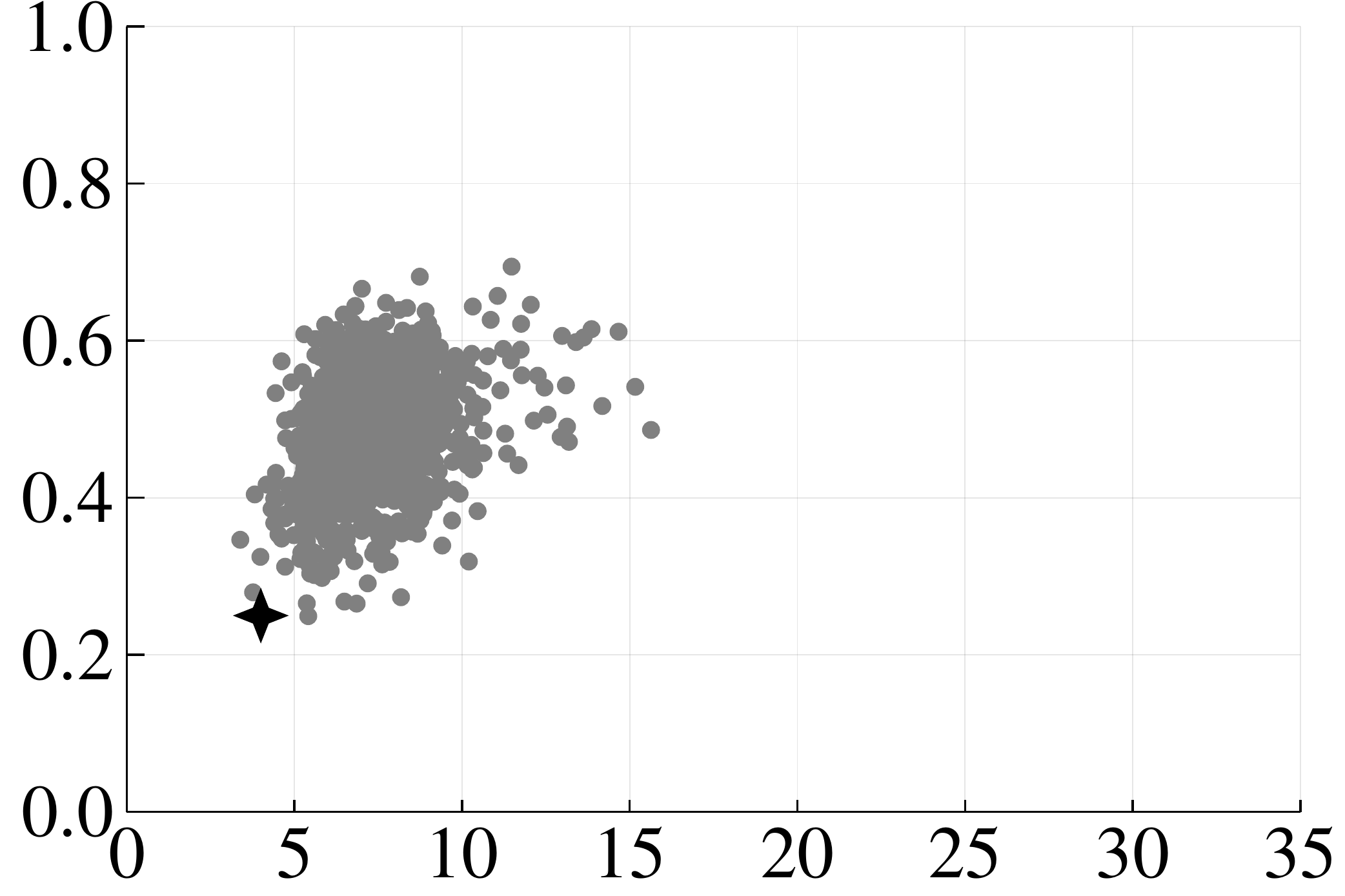}
    };
    \node at (9.2,-6.0) {
      \includegraphics[width=0.3\textwidth]{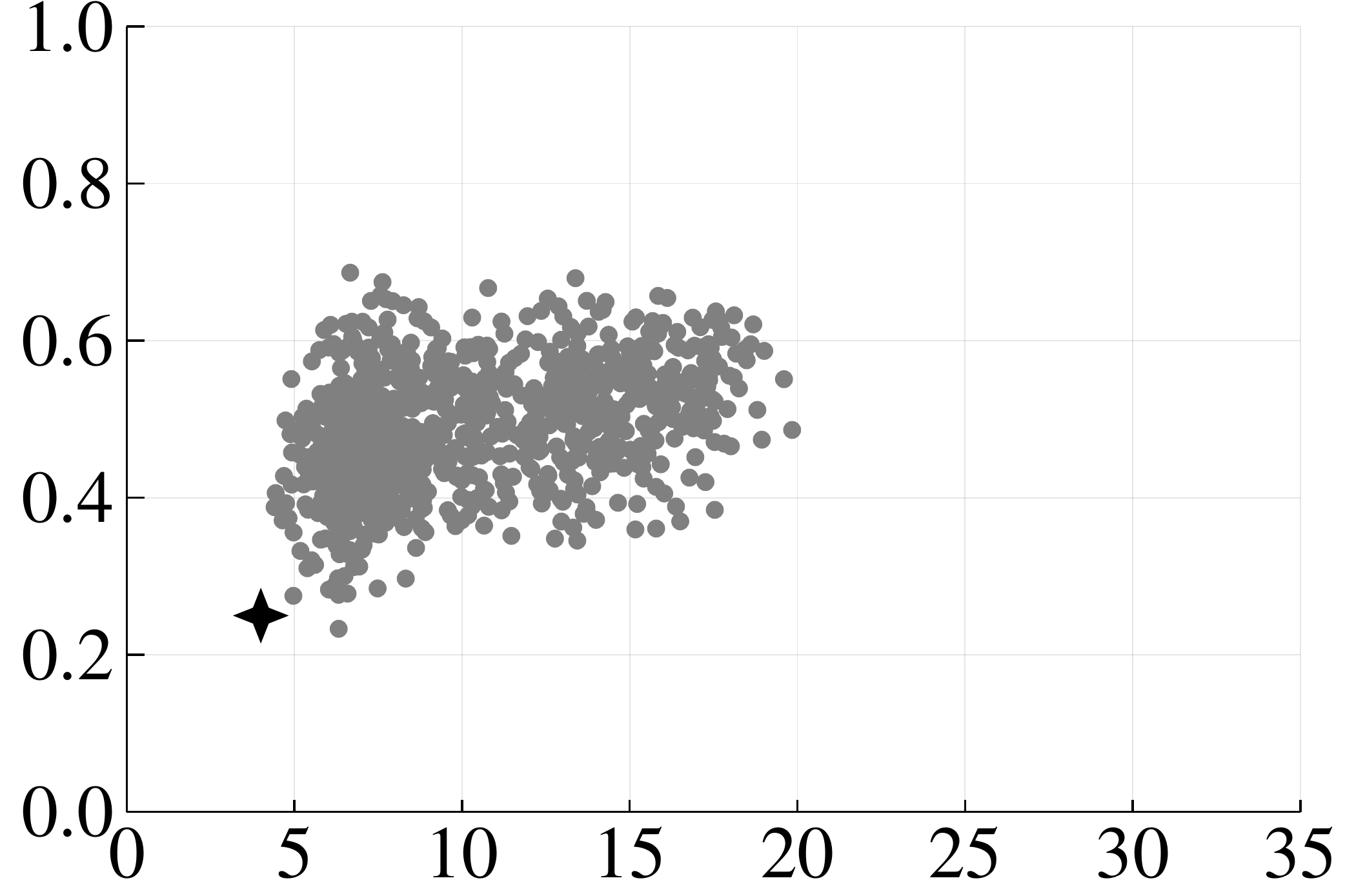}
    };

    \node at (-2.5,0) {\footnotesize $\alpha$};
    \node at (-2.5,-3.0) {\footnotesize $\alpha$};
    \node at (-2.5,-6.0) {\footnotesize $\alpha$};

    \node at (0,-7.8) {\footnotesize $\lambda$};
    \node at (4.6,-7.8) {\footnotesize $\lambda$};
    \node at (9.2,-7.8) {\footnotesize $\lambda$};

    \node[rotate=-90] at (11.8,0) {\footnotesize ${\alpha\sim\text{Beta}(\frac{1}{2},\frac{1}{2})}$};
    \node[rotate=-90] at (11.8,-3.0) {\footnotesize ${\alpha\sim\text{Beta}(1,1)}$};
    \node[rotate=-90] at (11.8,-6.0) {\footnotesize ${\alpha\sim\text{Beta}(2,2)}$};

    \node at (0,2) {\footnotesize ${\lambda\sim\text{Gamma}(1,\frac{1}{4})}$};
    \node at (4.6,2) {\footnotesize ${\lambda\sim\text{Gamma}(20,2)}$};
    \node at (9.2,2) {\footnotesize ${\lambda\sim\text{Gamma}(\frac{1}{100},\frac{1}{100})}$};
    \end{scope}
  \end{tikzpicture}
}}
  \caption{Prior sensitivity, for sample graphs of size $T=50$ (top) and $T=100$ (bottom). 
    Grey dots are posterior samples; the black star indicates the parameter values used to generate the data. 
  }
  \label{fig:prior_sensitivity}
  \vspace{-.5cm}
\end{figure}

\clearpage

%%%%%%%% bibliography %%%%%%%%

\bibliographystyle{natbib}
\bibliography{references}

\end{document}

%% file: fig_red_circle.tex
\begin{tikzpicture}
  \fill[red, draw=gray] (2,2) circle (0.15cm);
\end{tikzpicture}

%% file: fig_blue_circle.tex
\begin{tikzpicture}
  \fill[blue, draw=gray] (2,2) circle (0.13cm);
\end{tikzpicture}

%% file: fig_green_circle.tex
\begin{tikzpicture}
  \fill[green, draw=gray] (2,2) circle (0.11cm);
\end{tikzpicture}

%% file: fig_orange_circle.tex
\begin{tikzpicture}
  \fill[orange, draw=gray] (2,2) circle (0.085cm);
\end{tikzpicture}

%% file: fig_cyan_circle.tex
\begin{tikzpicture}
  \fill[cyan, draw=gray] (2,2) circle (0.075cm);
\end{tikzpicture}

%% file: fig_pink_circle.tex
\begin{tikzpicture}
  \fill[pink, draw=gray] (2,2) circle (0.06cm);
\end{tikzpicture}